\documentclass[11pt,letterpaper]{article}
\usepackage{mehdis}
\newif\ifnotes\notestrue
\ifnotes
\usepackage{color}
\definecolor{mygrey}{gray}{0.50}
\newcommand{\notename}[2]{{\textcolor{mygrey}{\footnotesize{\bf (#1:} {#2}{\bf ) }}}}

\newcommand{\pnote}[1]{{\endnote{#1}}}

\else

\newcommand{\notename}[2]{{}}

\newcommand{\pnote}[1]{}

\fi

\newcommand{\mnote}[1]{\textcolor{blue}{\small {\textbf{(Mehdi:} #1\textbf{)}}}}
\newcommand{\anote}[1]{\textcolor{red}{\small {\textbf{(Anurag:} #1\textbf{) }}}}

\newcommand{\snote}[1]{\textcolor{magenta}{\small {\textbf{(Srini:} #1\textbf{) }}}}

\usepackage{umoline}
\usepackage{xcolor}
\definecolor{cobalt}{rgb}{0.0, 0.28, 0.67}

\newcommand{\Ene}{\mathcal{E}}

\newcommand{\vL}{v_{\operatorname{LR}}}
 
\begin{document}

\title{Sample-efficient learning of quantum many-body systems}
 \author{Anurag\\ Anshu\thanks{Institute for Quantum Computing and Department of Combinatorics and Optimization, University of Waterloo, Canada and Perimeter Institute for Theoretical Physics, Canada. \href{mailto:aanshu@uwaterloo.ca}{aanshu@uwaterloo.ca}} \and 
 Srinivasan \\Arunachalam\thanks{IBM Research. \href{mailto:Srinivasan.Arunachalam@ibm.com}{Srinivasan.Arunachalam@ibm.com}} 
 \and 
 Tomotaka\\ Kuwahara\thanks{Mathematical Science Team, RIKEN Center for Advanced Intelligence Project (AIP), Japan and Interdisciplinary Theoretical \& Mathematical Sciences Program (iTHEMS) RIKEN, Japan. \href{mailto:tomotaka.kuwahara@riken.jp}{tomotaka.kuwahara@riken.jp}}
 \and 
 Mehdi\\ Soleimanifar\thanks{Center for Theoretical Physics, MIT. \href{mailto:mehdis@mit.edu}{mehdis@mit.edu}} }

\date{\today}
 
\clearpage\maketitle
\thispagestyle{empty}
\begin{abstract}
We study the problem of learning the Hamiltonian of a quantum many-body system given samples from its Gibbs (thermal) state. The classical analog of this problem, known as learning graphical models or Boltzmann machines, is a well-studied question in machine learning and statistics. In this work, we give the first sample-efficient algorithm for the quantum Hamiltonian learning problem. In particular, we prove that polynomially many samples in the number of particles (qudits) are necessary and sufficient for learning the parameters of a spatially local Hamiltonian in $\ell_2$-norm. 

Our main contribution is in establishing the strong convexity of the log-partition function of quantum many-body systems, which along with the maximum entropy estimation yields our sample-efficient algorithm. Classically, the strong convexity for partition functions follows from the Markov property of Gibbs distributions. This is, however, known to be violated in its exact form in the quantum case. We introduce several new ideas to obtain an unconditional result that avoids relying on the Markov property of quantum systems, at the cost of a slightly weaker bound. In particular, we prove a lower bound on the variance of quasi-local operators with respect to the Gibbs state, which might be of independent interest. Our work paves the way toward a more rigorous application of machine learning techniques to quantum many-body problems.
\end{abstract}
\newpage
{
  \addtocontents{toc}{\protect\enlargethispage{2\baselineskip}}
  \hypersetup{linkcolor=black,linktoc=all}
   \setcounter{tocdepth}{2} 
  \tableofcontents
}
 
\pagenumbering{gobble}
\clearpage
\pagenumbering{arabic}
\setcounter{page}{1}

\section{Introduction}
The success of machine learning algorithms in analyzing high-dimensional data, has resulted in a surge of interest in applying these algorithms to study quantum many-body systems whose description requires dealing with an exponentially large state space. One important problem in this direction is the \emph{quantum Hamiltonian learning} problem, which has been the focus of many recent theoretical and experimental works \cite{Aradlearning,Arad_learning_open_dynamic,Wiebe2014hamiltonian,wiebe2014LearningImperfect,Flammia2019BaysianHamLearning,experimental_Hamiltonian_learning_nature}. Here, one would like to learn the underlying Hamiltonian of a quantum system given multiple identical copies of its Gibbs (thermal) state. The classical analog of this problem is a central problem in machine learning and modern statistical inference, known as \emph{learning graphical models} or \emph{Boltzmann machines} (aka Ising models). Classically, understanding the learnability of Boltzmann machines was initiated by the works of Hinton and others in the 80s~\cite{ackley1985learning,hinton1986learning}. In the past few years, there has been renewed interest in this subject and has seen significant progress resulting in \emph{efficient provable} learning algorithms for graphical models with optimal sample and time complexity especially for sparse and bounded-degree graphs~\cite{Bresler_learning,Klivans_learning,Interaction_screening,Hamilton2017graphical_models,Wright_graphical_model_regularized}. Thus far, a rigorous analysis of the \emph{quantum} Hamiltonian learning problem with guaranteed sample complexity has been lacking. The main contribution of this work is to provide the first  \emph{sample-efficient} algorithm for this~task.

We now introduce the quantum Hamiltonian learning problem. Consider a $\k$-local Hamiltonian~$H$ acting on $n$ qudits. In general, we can parameterize $H$ by
\begin{align}
    H(\m)=\sum_{\ell=1}^m \m_{\ell} E_{\ell}\nn
\end{align}
where $\m_{\ell}\in\bbR$ and the operators $E_{\ell}$ are Hermitian and $\{E_\ell\}$ forms an orthogonal basis for the space of operators. For instance in the case of qubits, $E_{\ell}$ are tensor product of at most $\k$ Pauli operators that act non-trivially only on spatially contiguous qubits. We let the vector $\m=(\m_1,\dots,\m_m)^\top$ be the vector of \emph{interaction coefficients}. In our setup, without loss of generality we assume the Hamiltonian is traceless, i.e., for the identity operator $E_{\ell}=\iden$, the coefficient $\mu_{\ell}=0$. At a \emph{inverse-temperature} $\b$, the qudits are in the \emph{Gibbs state}  defined as
\begin{align}
    \r_{\b}(\m)=\frac{e^{-\b H(\mu)}}{\Tr[e^{-\b H(\mu)}]}.\nn
\end{align}
In the learning problem, we are given multiple copies of $\r_{\b}(\m)$ and can perform arbitrary \emph{local measurements} on them. In particular, we can obtain all the $\k$-local \emph{marginals} of $\r_{\b}(\m)$ denoted by 
\ba
e_{\ell}=\Tr[\r_{\b}(\m) E_{\ell}]\quad \textit{for }\ell\in[m].\nn
\ea
The goal is to learn the coefficients $\m_{\ell}$ of the Hamiltonian $H$ using the result of these measurements. We call this the $\HLP$. Before stating our main results, we provide further motivations for looking at this problem.

\paragraph{Physics perspective.}
Quantum many-body systems consist of many quantum particles (qudits) that \emph{locally} interact with each other. The interactions between these particles are described by the Hamiltonian of the system. Even though the interactions in the Hamiltonian are local, the state of the whole system can be highly \emph{entangled}. This is not only true at low temperatures when the system is in the lowest energy eigenstate of its Hamiltonian (the ground state), but remains true even at finite temperatures when the state is a mixture of different eigenstates of the Hamiltonian known as the Gibbs or thermal state.

While the underlying fundamental interactions in these systems are long known to be given by Coulomb forces between electrons and nuclei, they are too complicated to be grasped in entirety. Physicists are primarily interested in ``effective interactions" that, if accurately chalked out, can be used to describe a variety of properties of the system. How can such effective interactions be learned in a system as complicated as, for example, the high temperature superconductor? Algorithms for Hamiltonian learning can directly address this problem and provide a suitable approximation to the effective interactions. 

\paragraph{Verification of quantum devices.}
The size of the available quantum computers is increasing and they are becoming capable of running more intricate quantum algorithms or preparing highly entangled states over larger number of qubits.   Due to the noise in these devices, a major challenge that accompanies the scalable development of quantum devices is to efficiently \emph{certify} their~functionality. 
In recent times, one widely used subroutine in quantum algorithms is \emph{quantum Gibbs sampling}. Preparing and measuring the Gibbs state of a given Hamiltonian is used in quantum algorithms for solving semi-definite programs~\cite{Brandao_sdp,Gilyen_sdp,Brandao_sdp2,Gilyen_sdp_2,brandao_sdp_quadratic}, quantum simulated annealing \cite{montanaro2015speading_up_MCMC,harrow2020MCMC}, metropolis sampling~\cite{Temme_metropolis}, quantum machine learning~\cite{wiebe2014quantum}, or quantum simulations at finite temperature \cite{GarnetChen2020Gibbs}. Given near term quantum devices will be noisy, an important problem when implementing these quantum subroutines  is to certify the performance of the quantum Gibbs samplers and to calibrate them. More specifically, it would be ideal to have a \emph{classical} algorithm that given samples from the output of a Gibbs sampler determines if the correct Hamiltonian has been implemented. 

\paragraph{Quantum machine learning for quantum data.}
A popular family of models for describing classical distributions are \emph{graphical models} or \emph{Markov random fields}. These models naturally encode the causal structure between random variables and have found widespread applications in various areas such as social networks, computer vision, signal processing, and statistics (see~\cite{Wright_graphical_model_regularized} for a survey). A simple and extremely well-studied example of such a family is the classical \emph{Ising model} (also known as the \emph{Boltzmann machine}) defined over a graph whose vertices correspond to the random variables $x_i$. A natural distribution that one can associate to this model~is
\ba
\pr[X=x]=\frac{1}{Z}\exp\Big(\sum_{i\sim j} J_{ij} x_i x_j+ \sum_i h_i x_i\Big)\label{eq:s1}
\ea
where $J_{ij},h_i \in \bbR$ are real coefficients and the normalization factor $Z$ is called the \emph{partition function}. This distribution in Eq.~\eqref{eq:s1} is also known as the \emph{Gibbs distribution}. There is a rich body of work on learnability of Ising models given samples from the Gibbs distribution. Remarkably, a sequence of works concluded in showing a classical \emph{efficient} algorithm with a running time quadratic in the number of vertices that outputs estimates of the coefficients $J_{ij}$ and $h_i$ \cite{Bresler_learning,Klivans_learning,Interaction_screening}. Similar results have been also proved for more general graphical models.

Considering these achievements in learning theory and the broad practical application of machine learning algorithms, there has been a rising interest in connecting these techniques to problems in quantum computing and many-body physics. This along with other related problems is loosely referred to as \emph{quantum machine learning}. Is there a natural problem that we can rigorously establish such a connection for it? Thus far, almost-all the proposals we are aware of in this direction are mostly based on  heuristic grounds. One proposal that stands out due to its similarity to the classical case is the problem of learning quantum Ising model (aka quantum Boltzmann machine) or more generally the $\HLP$. 

In this paper, we rigorously show that by applying tools from statistics and machine learning such as maximum entropy estimation, one can get a sample complexity for the $\HLP$ that is \emph{polynomial} in the number of qudits. To the best of our knowledge, this is the first such result that unconditionally obtains a non-trivial sample complexity. We believe our work opens the doors to further study of this problem using insight from machine learning and optimization theory.

\section{Main result}
Motivated by these applications, we now formally define the Hamiltonian learning problem. 
\begin{problem}[Hamiltonian learning problem]
\label{prob:hamiltonianlearning}
 
Consider a $\k$-local Hamiltonian $H(\m)=\sum_{\ell=1}^m \m_\ell E_\ell$ that acts on $n$ qudits and consists of $m$ local terms such that  $\max_{\ell\in[m]}|\m_\ell| \leq 1$. In the $\HLP$, we are given $N$ copies of the Gibbs state of this Hamiltonian
 $$
 \r_{\b}(\m)=\frac{e^{-\beta H(\mu)}}{\tr[e^{-\b H(\mu)}]}
 $$ 
at a fixed inverse-temperature $\b$. Our goal is to obtain an estimate $\hat{\m}=(\hat{\m}_1,\dots,\hat{\m}_m)$ of the coefficients $\mu_k$ such that with probability at least $1-\delta$,
 \begin{align}
    \norm{\m- \hat{\m}}_2 \leq \e,\nn
\end{align}
where $\norm{ \m- \hat{\m}}_2=\left(\sum_{\ell=1}^m |\mu_\ell-\hat{\mu}_\ell|^2\right)^{\frac{1}{2}}$ is the $\ell_2$-norm of the difference of $\mu$ and $\hat{\mu}$.
\end{problem}
Our main result is a sample-efficient algorithm for the $\HLP$.  
\begin{thm}[Sample-efficient Hamiltonian learning]
\label{thm:main}
The $\HLP$~\ref{prob:hamiltonianlearning} can be solved using  
\ba
N=\mathcal{O}\left(\frac{ e^{\orderof{\b^c}}}{\b^{\tilde{c}}\e^2} \cdot m^3\cdot \log\Big(\frac{m}{\d}\Big)\right) \label{eq:s22intro}
\ea
copies of the Gibbs state $\r_{\b}(\m)=e^{-\beta H(\m)}/\tr[e^{-\b H(\m)}]$, where $c,\tilde{c}\geq 1$ are constants that depend on the geometry of the Hamiltonian.
\end{thm} 
As far as we are aware, our work is the first to establish unconditional and rigorous upper bounds on the sample complexity of the $\HLP$. For spatially local Hamiltonians the number of interaction terms $m$ scales as $O(n)$. Hence, our result in \thmref{main} implies a sample complexity polynomial in the number of qudits. 

The number of samples in \eqref{eq:s22intro} increases as $\b \rightarrow \infty$ or $\b \rightarrow~0$. As the temperature increases ($\b \rightarrow~0$), the Gibbs state approaches the maximally mixed state independent of the choice of parameters $\m$. At low temperatures ($\b\rightarrow \infty$), the Gibbs state is in the vicinity of the ground space, which for instance, could be a product state $\ket{0}^{\ot n}$ for the various choices of~$\m$. In either cases, more sample are required to distinguish the parameters $\m$.

To complement our upper bound, we  also obtain a $\Omega(\sqrt{m})$ lower bound for the $\HLP$ with $\ell_2$ norm using a simple reduction to the state discrimination problem. The proof appears in Appendix \ref{learning_lower_bound}. Hence, our upper bound in Theorem~\ref{thm:main} is tight up to polynomial factors.
\begin{thm}
\label{thm:lower_bound}
The number of copies $N$ of the Gibbs state needed to solve the $\HLP$ and outputs a $\hat{\mu}$ satisfying $\|\hat{\mu}-\mu\|_2\leq \varepsilon$ with probability $1-\delta$ is lower bounded by
$$
 N\geq \Omega\Big(\frac{\sqrt{m}+\log(1-\delta)}{\beta\varepsilon}\Big).
$$
\end{thm}
\section{Proof overview}
In order to prove our main result, we introduce several new ideas. In this section, we provide a sketch of the main ingredients in our proof.

\subsection{Maximum entropy estimation and sufficient statistics}\label{sec:Maximum entropy estimation and sufficient statistics}
In statistical learning theory, a conventional method for obtaining the parameters of a probability distribution from data relies on the concepts of \emph{sufficient statistics} and the \emph{maximum entropy estimation}. Suppose $p(x;\m)$ is a family of probability distributions parameterized by  $\m$ that we want to learn. This family could for instance be various normal distributions with different mean or variance. Let $X_1,\dots,X_m\sim p(x;\m)$ be $m$ samples from a distribution in this family. A sufficient \emph{statistic} is a function $T$ of these samples $T(X_1,\dots,X_m)$ such that conditioned on that, the original date set $X_1,\dots,X_m$ does not depend on the parameter $\m$. For example, the sample mean and variance are well known sufficient statistic functions.

After obtaining the sufficient statistic of a given data set given classical samples, there is a natural algorithm for estimating the parameter $\m$: among all the distributions that match the observed statistic $T(X)$ find the one that maximizes the Shannon entropy. Intuitively, this provides us with the least biased estimate given the current samples \cite{jaynes1957informationClassical,Jaynes1982Max_entropy}. This algorithm, which is closely related to the maximum likelihood estimation, is commonly used for analyzing the sample complexity of classical statistical~problems. 

Our first observation when addressing the $\HLP$ is that this method can be naturally extended to the quantum problem \cite{jaynes1957informationQuantum}. Indeed, the maximum entropy principle has already appeared in other quantum algorithms such as \cite{Brandao_sdp2}. More formally, we first show that the marginals $\tr[E_\ell \r]$ for $\ell\in[m]$ form a sufficient statistic for the $\HLP$. 

\begin{prop}[Matching local marginals implies global equivalence]\label{prop:matching marginals}
Consider the following two Gibbs states
\begin{align}
    \r_{\b}(\m)=\frac{e^{-\b\sum_\ell \m_\ell E_\ell}}{\Tr[e^{-\b\sum_\ell \m_\ell E_\ell}]},\quad  \r_{\b}(\l)=\frac{e^{-\b\sum_\ell \l_\ell E_\ell}}{\Tr[e^{-\b\sum_\ell \l_\ell E_\ell}]} \label{eq:b2intro}
\end{align}
such that $\Tr[\r_{\b}(\l) E_\ell]=\Tr[\r_{\b}(\m) E_\ell]$ for all $\ell\in [m]$, i.e. all the $\k$-local marginals of $\r_{\b}(\l)$ match that of~$\r_{\b}(\m)$. Then, we have $\r_{\b}(\l)=\r_{\b}(\m)$, which in turns implies $\l_\ell=\m_\ell$ for $\ell\in [m]$.
\end{prop}
Similar to the classical case discussed above, one implication of \propref{matching marginals} is a method for learning the Hamiltonian $H$: first measure all the $\k$-local marginals of the Gibbs state $e_\ell$, then among all the states of the form \eqref{eq:b2intro}, find the one that matches those marginals. Finding such a state can be naturally formulated in terms of an optimization problem known as the \emph{maximum entropy problem}:
\begin{equation}
\begin{aligned}
\max_{\s} \quad & S(\s)\\
\textrm{s.t.} \quad & \Tr[\s E_\ell]=e_\ell, \quad \forall \ell\in[m]\\
& \s>0, \quad \Tr[\s]=1.\label{eq:r2}
\end{aligned}
\end{equation}
 where $S(\s)=-\tr[\s\log \s]$ is the \emph{von Neumann entropy} of the state $\s$. The optimal solution of this program is a quantum state with a familiar structure \cite{jaynes1957informationQuantum}. Namely, it is a Gibbs state $\r(\l)$ for some set of coefficients $\l=(\l_1,\dots,\l_m)$. The coefficients $\l$ are the \emph{Lagrange multipliers} corresponding to the dual of this program. Indeed, we can write the dual program of Eq.~\eqref{eq:r2} as follows:
\begin{equation}
\begin{aligned}
\mu=\argmin_{\l=(\l_1,\dots,\l_m)} \log Z_{\beta}(\l)+\b\cdot\sum_{\ell=1}^m \l_\ell e_\ell,
\quad & \label{eq:4}
\end{aligned}
\end{equation}
where $Z_{\beta}(\lambda)=\tr\big(e^{-\beta\cdot\sum_\ell \l_\ell E_\ell}\big)$ is the \emph{partition function} at inverse-temperature $\b$. In principle, according to the result of \propref{matching marginals}, we could solve the $\HLP$ by finding the optimal solution of the dual program in~\eqref{eq:4}. Of course, the issue with this approach is that since we have access to limited number of samples of the original Gibbs state $\r_\b(\m)$, instead of the {exact} marginals $e_\ell$, we can only \emph{approximately} estimate the $e_\ell$s. We denote these estimates by $\hat{e}_\ell$. This means instead of solving the dual program \eqref{eq:4}, we solve its \emph{empirical}~version
\begin{equation}
\begin{aligned}
\hat{\mu}=\argmin_{\l=(\l_1,\dots,\l_m)} \quad &\log Z_{\b}(\l)+\b \cdot \sum_{\ell=1}^m \l_\ell \hat{e}_\ell. \label{eq:9}
\end{aligned}
\end{equation}

The main technical problem that we address in this work is analyzing the robustness of the programs \eqref{eq:r2} and \eqref{eq:4} to the statistical error in the marginals as appears in \eqref{eq:9}. This is an instance of a \emph{stochastic optimization} which is a well-studied problem in optimization. In the next section, we review the ingredients from convex optimization that we need in our analysis. 

\subsection{Strong convexity}
One approach to incorporate the effect of the statistical errors in the marginals $e_\ell$ into the estimates for $\m_\ell$ is to use \propref{matching marginals}. It is not hard to extend this proposition to show that if a Gibbs states $\r_{\b}(\l)$ \emph{approximately} matches the marginals of $\r_{\b}(\m)$ up to some error~$\e$, then $\norm{\r_{\b}(\m)-\r_{\b}(\l)}^2_1\leq \mathcal{O}(m \e)$ (see \secref{Local Hamiltonians and quantum Gibbs states} for more details). This bound, however, is not strong enough for our purposes. This is because if we try to turn this bound to a one on the coefficients $\mu_\ell$ of the Hamiltonian, we need to bound $\norm{\log \r_{\b}(\m)-\log \r_{\b}(\l)}$. Unfortunately, the function $\log(x)$ does not have a bounded gradient (i.e., it is not Lipschitz) over its domain and in general $\norm{\log \r_{\b}(\m)-\log \r_{\b}(\l)}$ can be exponentially worse than $\norm{\r_{\b}(\m)-\r_{\b}(\l)}_1$. In order to overcome the non-Lipschitz nature of the logarithmic function and bound $\norm{\log \r_{\b}(\m)-\log \r_{\b}(\l)}$, we prove a property of the dual objective function  \eqref{eq:4} known as the \emph{strong convexity}, which we define now.

\begin{definition}\label{def:strong convexity}
Consider a convex function $f:\bbR^{m}\mapsto \bbR$ with gradient $\nabla f(x)$ and Hessian $\nabla^2 f(x)$ at a point $x$.\footnote{Recall that the entries of the Hessian matrix $\nabla^2 f(x)$ are given by $\frac{\partial^2}{\partial x_i \partial x_j} f(x)$} This function $f$ is said to be $\a$-strongly convex in its domain if it is differentiable and for all $x,y$,
\begin{align}
    f(y)\geq f(x)+\nabla f(x)^\top (y-x) +\frac{1}{2}\a \norm{y-x}^2_2,\nn
\end{align}
or equivalently if its Hessian satisfies
\begin{align}
    \nabla^2 f(x)\succeq \a \iden.\footnotemark\label{eq:6_informal}
\end{align}
\footnotetext{By $A\succeq B$ we mean $A-B$ is positive semidefinite.}
In other words, for any vector $v\in \bbR^m$, it holds that $\sum_{i,j} v_i v_j \frac{\partial^2}{\partial x_i \partial x_j}f(x)\geq \a \norm{v}^2_2$.
\end{definition}

Roughly speaking, strong convexity puts a limit on how \emph{slow} a convex function $f(x)$ changes.\footnote{This should not be confused with a related property called the smoothness which limits how fast the function~grows.} This is particularly useful because given two points $x,y$ and an upper bound on $|f(y)-f(x)|$ and $\nabla f(x)^\top (y-x)$, it allows us to infer an upper bound on $\norm{y-x}_2$. 

For our application, we think of $f$ as being $\log Z_{\b}(\cdot)$. Then the difference $|f(y)-f(x)|$ is the difference between the optimal solution of the original program in Eq.~\eqref{eq:4} and that of its empirical version in Eq.~\eqref{eq:9} which includes the statistical error. We apply this framework to our optimization \eqref{eq:9} in two steps: 
\begin{itemize}[leftmargin=*]
\item[1)] Proving the strong convexity of the objective function: This is equivalent to showing that the log-partition function (aka the free energy) is strongly convex, i.e., $\nabla^2 \log Z_{\b}(\l) \succeq \a \iden$ for some positive coefficient $\alpha$. In particular, this means that the optimization \eqref{eq:9} is a convex program. This result is the main technical contribution of our work and is stated in the following theorem:
\begin{thm}[Informal: strong convexity of log-partition function]\label{thm:strong convexity of log-partition function_informal} Let $H=\sum_{\ell=1}^m \m_\ell E_\ell$ be a $\k$-local Hamiltonian over a finite dimensional lattice with $\norm{\m}\leq 1$. For a given inverse-temperature $\b$, there are constants $c, c'>3$ depending on the geometric properties of the lattice  such that 
\ba
\nabla^2 \log Z_{\b}(\m) \succeq  e^{-\mathcal{O}(\b^c)} \frac{\b^{c'}}{m} \cdot \iden, \label{eq:convexity_bound}
\ea
i.e., for every vector $v\in \bbR^m$ we have $v^T \cdot \nabla^2 \log Z_{\b}(\m) \cdot v \geq e^{-\mathcal{O}(\b^c)}\frac{\b^{c'}}{m} \cdot \norm{v}^2_2$.
\end{thm}
\item[2)] 
Bounding the error in estimating $\m$ in terms of the error in estimating the marginals $e_\ell$:  In this step we show that as long as the statistical error of the marginals is small, using the strong convexity property from step (1), we can still prove an upper bound on the difference between the solutions  of the convex programs~\eqref{eq:4},~\eqref{eq:9}.

We discuss this in more details later in \secref{Stochastic convex optimization_prelim}. The result can be stated as follows:
\begin{thm}[Error bound from strong convexity]\label{thm:Error bound from strong convexity informal}
Let $\delta,\alpha > 0$. Suppose the marginals $e_\ell$ are determined up to error $\d$, i.e., $|e_\ell-\hat{e}_\ell|\leq \d$ for all $\ell\in[m]$. Additionally assume $\nabla^2 \log Z_{\b}(\l) \succeq \a \iden$ and $\norm{\l}\leq 1$. Then the optimal solution to the program \eqref{eq:9} satisfies 
$$
\norm{\m-\hat{\m}}_2\leq \frac{2\b \sqrt{m}\d}{\alpha}
$$
\end{thm}
\end{itemize}
Combining \thmref{strong convexity of log-partition function_informal} and \thmref{Error bound from strong convexity informal}, we obtain the main result of our paper. We now proceed to sketch the proof of~\thmref{strong convexity of log-partition function_informal}.

\subsection{Strong convexity of log-partition function: Review of the classical case}\label{sec:Strong convexity of of log z classical}
In order to better understand the motivation behind our quantum proof, it is insightful to start with the \emph{classical}  Hamiltonian learning problem. This helps us better describe various  subtleties and what goes wrong when trying to adapt the classical techniques to the quantum case. We continue using the quantum notation here, but the reader can replace the Hamiltonian $H$, for instance, with the classical Ising model $H=\sum_{i\sim j} J_{ij} x_i x_j$ (where $x_i\in \{-1,1\}$ and $J_{ij}\in \bbR$).

The entries of the Hessian $\nabla^2 \log Z_{\b}(\m)$ for classical Hamiltonians are given by 
\ba
\frac{\partial^2}{\partial \m_i \partial \m_j} \Big[\log Z_{\b}(\m)\Big]=\Cov[E_i,E_j]\label{eq:s12}
\ea
where $\Cov$ is the covariance function which is defined as $\Cov[E_i,E_j]=\langle E_i E_j\rangle-\langle E_i\rangle \langle E_j\rangle$ with the expectation taken with respect to the Gibbs distribution (i.e., $\langle E\rangle=\Tr[E\cdot \rho_\beta(\m)]$). To prove the strong convexity of the log-partition function at a constant $\b$, using \eqref{eq:s12} it suffices to show that for every vector $v$, we have 
\ba
\sum_{i,j} v_i v_j \frac{\partial^2}{\partial \m_i \partial \m_j} \log Z_{\b}(\m)=\Var\left[\sum_{\ell=1}^m v_\ell E_\ell\right]\geq \Omega(1) \cdot \sum_{\ell=1}^m v_\ell^2.\label{eq:s11}
\ea
Although the operator $\sum_\ell v_\ell E_\ell$ is a local Hamiltonian, note the mismatch between this operator and the original Hamiltonian in the Gibbs state $\sum_{\ell=1}^m \m_\ell E_\ell$. Also note that compared to the inequality \eqref{eq:convexity_bound}, the inequality \eqref{eq:s11} claims a stronger lower bound of $\Omega(1)$.

Before proving Eq.~\eqref{eq:s11}, we remark that an \emph{upper bound} of $\Var[\sum_{\ell=1}^m v_\ell E_\ell]\leq \mathcal{O}(1) \norm{v}_2^2$ is known in literature, under various conditions like the decay of correlations both in classical and quantum settings \cite{Araki1969,Gross1979,Park1995,ueltschi2004cluster,PhysRevX.4.031019,frohlich2015some}. This upper bound intuitively makes sense because the variance of the thermal state of a Hamiltonian and other local observables are expected to be \emph{extensive}, i.e., they scale with the number of particles (spins) or norm of the Hamiltonian, which is replaced by $\norm{v}_2^2$ in our setup. However, in the classical Hamiltonian learning problem, we are interested in obtaining a \emph{lower bound} on the variance. To this end, a crucial property of the (classical) Gibbs distributions that allows us to prove the inequality \eqref{eq:s11} is the conditional independence or the \emph{Markov property} of classical systems.
\begin{definition}[Markov property]\label{def:Markov property}
Suppose the interaction graph is decomposed into three disjoint regions $A$, $B$, and $C$ such that region $B$ ``shields" $A$ from $C$, i.e., the vertices in region $A$ are not connected to those in $C$. Then, conditioned on the sites in region $B$, the distribution of sites in $A$ is independent of those in $C$. This is often conveniently expressed in terms of the conditional mutual information by $I(A:C|B)=0$.
\end{definition}
It is known by the virtue of the Hammersley-Clifford theorem \cite{Hammersley_Clifford} that the family of distributions with the Markov property coincides with the Gibbs distributions. Using this property, we can lower bound $\Var\left[\sum_{\ell=1}^m v_\ell E_\ell\right]$ in terms of variance of local terms $E_\ell$ by \emph{conditioning} on a subset of sites. To this end, we consider a partition of the interaction graph into two sets $A$ and~$B$. The set~$B$ is chosen, suggestively, such that the vertices in $A$ are not connected (via any edges) to each other. We denote the spin configuration of sites in $B$ collectively by $s_B$. Then using the concavity of the variance and the Markov property, we have

\ba
\Var\left[\sum_{\ell=1}^m v_\ell E_\ell\right]&\overset{(1)}\geq \bbE_{s_B}\left[\Var\left[\sum_{\ell=1}^m v_\ell E_\ell \Bigm\vert s_B\right]\right]\nn\\
&\overset{(2)}= \sum_{x\in A} \bbE_{s_B}\left[\Var\left[\sum_{\ell:E_{\ell} \textnormal{ acts on } x} v_{\ell} E_{\ell}\Bigm\vert s_B \right]\right]\nn\\
&\overset{(3)}\geq \Omega(1)\sum_{\ell=1}^m v_\ell^2,\label{eq:s13}
\ea
where inequality $(1)$ follows from the law of total variance, equality $(2)$ can be justified as follows: by construction, the local terms $E_\ell$ either completely lie inside region $B$ or intersect with only one of the sites in region $A$. In the former, the local conditional variance $\Var\left[E_\ell\left\vert s_B\right.\right]$ vanishes, whereas in the latter, the interaction terms $E_{\ell}$ that act on different sites $x\in A$ become uncorrelated and the global variance decomposes into a sum of local variance. Finally, inequality $(3)$ is derived by noticing that at any constant inverse-temperature $\b$, the local variance is lower bounded by a constant that scales as $e^{-\mathcal{O}(\b)}$. By carefully choosing the partitions $A$ and $B$ such that $|A|=\mathcal{O}(n)$, we can make sure that the variance in inequality $(2)$ is a constant fraction of the $\sum_{\ell=1}^m v_\ell^2$ as in \eqref{eq:s13} (see \cite{Montanari_sufficient_statistics,Interaction_screening} for details). This lower bound on variance results in a sample complexity $\mathcal{O}\big( e^{\mathcal{O}(\b)} m(\log m)\e^{-2}\big)$, which compared to our result in \thmref{main} is more efficient (by only a polynomial factor in $m$).

\subsection{Strong convexity of log-partition function: Proof of the quantum case}\label{sec:Strong convexity of of log z quantum}
If we try to directly quantize the proof strategy of the classical case in the previous section, we immediately face several issues. We now describe the challenges in obtaining a quantum proof along with our techniques to overcome them.

\subsubsection{Relating the Hessian to a variance}
The first problem is that we cannot simply express the entries of the Hessian matrix $\nabla^2 \log Z_{\b}(\m)$ in terms of $\Cov[E_i,E_j]$ as in \eqref{eq:s12}. This expression in~\eqref{eq:s12} only holds for Hamiltonians with \emph{commuting} terms, i.e., $[E_i,E_j]=0$ for all $i,j\in[m]$. The Hessian for the non-commuting Hamiltonians takes a complicated form (see \lemref{expression for Hessia} for the full expression) that makes its analysis difficult. Our first contribution is to recover a similar result to \eqref{eq:s11} in the quantum case by showing that, for every $v$, we can still \emph{lower bound} $v^{\top}\cdot \nabla^2 \log Z_{\b}(\m)\cdot v$ by the variance of a suitably defined \emph{quasi-local} operator. We later define what we mean by ``quasi-local'' more formally (see \defref{Quasi-local operators} in the body), but for now one can assume such an operator is, up to some small error, sum of local terms.

\begin{lem}[A lower bound on $\pmb{\nabla^2 \log Z_{\b}(\m)}$]
\label{lem:variance lower bound on Hessianintro}
For any vector $v \in \bbR^m$, we define a quasi-local operator $\quW=\sum_{\ell=1}^m v_\ell \widetilde{E}_\ell$, where the operators $\tilde{E}_\ell$ are defined by 
\begin{align} 
\widetilde{E}_\ell= \int_{-\infty}^\infty f_\beta(t)\ e^{-iHt}\ E_\ell\ e^{iHt} dt.
\end{align}
Here $f_\beta(t)=\frac{2}{\beta\pi}\log \frac{e^{\pi |t|/\beta}+1}{e^{\pi |t|/\beta}-1}$ is defined such that $f_\beta(t)$ scales as $\frac{1}{\b} e^{-\pi |t|/\beta}$ for large $t$ and $f_\beta(t) \propto \log(1/t)$ for $t\to +0$. We claim that 
\ba
\sum_{i,j} v_i v_j\frac{\partial^2}{\partial \m_i \partial \m_j}\log Z_{\b}(\m) \geq \beta^2\Var [\quW] \label{eq:r6_informal}
\ea
\end{lem}
\subsubsection{Lower bounding the variance}\label{sec:Lower bounding the variance}
As a result of \lemref{variance lower bound on Hessianintro}, we see that from here onwards, it suffices to lower bound the variance of the quasi-local operator $\quW=\sum_{\ell=1}^m v_\ell \tilde{E}_\ell$. One may expect the same strategy based on the Markov property in \eqref{eq:s13} yields the desired lower bound. Unfortunately, it is known that a natural extension of this property to the quantum case, expressed in terms of the \emph{quantum conditional mutual information} (qCMI), does not hold. In particular, example Hamiltonians are constructed in \cite{Poulin2008quantumGraphicalModels} such that for a tri-partition $A,B,C$ as in \defref{Markov property}, their Gibbs states have non-zero qCMI, i.e., $I(A:C|B)>0$. Nevertheless, it is \emph{conjectured} that an \emph{approximate} version of this property can be recovered i.e., $I(A:C|B)\leq e^{-\Omega(\dist(A,C))}$. In other words, the approximate property claims that qCMI is exponentially small in the \emph{width} of the shielding region $B$. Thus far, this conjecture has been proved only at sufficiently high temperatures~\cite{Kohtaro_cmi_cluster} and on 1D chains~\cite{Brandao_gibbs_cmi}. Even assuming this conjecture is true, we currently do not know how to recover the argument in \eqref{eq:s13}. We get back to this point in \secref{Discussion and open questions}. Given this issue we ask,
\begin{quote}
\centering
        \emph{Can we obtain an unconditional lower bound on the variance of a quasi-local observable at any inverse-temperature $\b$ without assuming quantum conditional independence?}
\end{quote}
Our next contribution is to give an affirmative answer to this question. To achieve this, we modify the classical strategy as explained below. 

\paragraph{From global to local variance.} One ingredient in the classical proof is to lower bound the global variance $\Var[\sum_\ell v_\ell E_\ell]$ by sum of local conditional variances $\Var[E_\ell|s_B]$ as in \eqref{eq:s13}. We prove a similar but slightly weaker result in the quantum regime. To simplify our discussion, let us ignore the fact that $\quW=\sum_\ell v_\ell \tilde{E}_\ell$ is a quasi-local operator and view it as (strictly) local. Consider a special case in which $v$ is such that the operator $\quW$ is supported on a small number of sites. For instance, it could be that $v_1>0$ while $v_2, \ldots ,v_m=0$. Then the variance $\Var[\quW]$ can be easily related to the local variance $\Var[E_1]$ and since $E_1^2=\iden$, $|\Tr[E_1\rho_{\beta}]|< 1$, we get
$$\Var[\quW]=v_1^2\cdot \br{\Tr[E_1^2\rho_{\beta}]-\Tr[E_1\rho_{\beta}]^2}\geq \Omega(1)\cdot v_1^2$$
We show that even in the general case, where $v_1,\dots,v_m$ are all non-zero, we can still relate $\Var[\quW]$ to the variance of a local operator supported on a constant region. Compared to the classical case in \eqref{eq:s13}, where the lower bound on $\Var[W]$ includes a sum of $\mathcal{O}(m)$ local terms, our reduction to a \emph{single} local variance costs ``an extra factor of $m$" in the strong convexity bound in \thmref{strong convexity of log-partition function_informal}.

Our reduction to local variance is based on the following observation. By applying Haar-random local unitaries, we can remove all the terms of the operator $\quW$ except those that act on an arbitrary qudit at site $i$. We denote the remainder terms by $\widetilde{W}_{(i)}$ defined via 
$$\quW_{(i)} =\quW-\bbE_{U_i\sim \mathrm{Haar}}[U_i^{\dag} \quW U_i].$$
By using triangle inequality this relation implies 
\ba
\Var[\quW] \geq \frac{1}{2}\tr[\quW_{(i)}^2 \r_{\b}]-\bbE_{U_i} \left[\tr[ \quW^2 \cdot U_i\r_{\b}U_i^{\dag}]\right].\label{eq:s20}
\ea
Hence, if we could carefully analyze the effect of the term $\bbE_{U_i}[ \tr[ \quW^2 \cdot U_i\r _{\b}U_i^{\dag}]]$, this will allow us to relate the global variance $\Var[\quW]$ to the local variance $\tr[\quW_{(i)}^2\r_{\b}]$. We discuss this next.
\paragraph{Bounding the effect of local unitaries.} While applying the above reduction helps us to go to an easier local problem, we need to deal with the changes in the spectrum of the Gibbs state due to applying the random local unitaries $U_i$. Could it be that the unitaries $U_i$ severely change the spectral relation between $\quW$ and $\rho_{\b}$? We show that this is not the case, relying on the facts: (1) local unitaries cannot mix up subspaces of $\quW$ and $H$ that are energetically far away and (2) the weight given by the Gibbs state $\rho_\beta$ to nearby subspaces of $H$ are very similar at small $\beta$. Thus, (1) allows us to focus the subspaces that are close in energy and (2) shows that similar weights of these subspaces do not change the variance by much. In summary, we prove:

\begin{prop}[Invariance under local unitaries, informal]\label{prop:Invariance under local unitaries, informal}
Let $U_X$ be a local unitary operator acting on region $X$ that has a constant size. There exists a constant $c\leq 1$ such that 
\begin{align}
&\Tr\left[\quW^2 \cdot U_X \r_{\b} U_X^{\dagger}\right]\leq \left(\Var[\quW]\right)^c. \label{eq:s22}
\end{align} 
\end{prop}
When combined with \eqref{eq:s20}, the inequality \eqref{eq:s22} implies the following loosely stated local lower bound on the global variance:
\begin{equation}
\label{eq:informallocallb}
\Var[\quW]\geq \br{\Tr\left[\quW_{(i)}^2\r_{\b}\right]}^{\frac{1}{c}}.\nn
\end{equation}
With this reduction, it remains to find a constant lower bound on $\Tr[\quW_{(i)}^2\r_{\b}]$. This can be done, again, by applying a local unitary $U$. Roughly speaking, we use this unitary to perform a ``change of basis'' that relates the local variance at finite temperature to its infinite-temperature version. The spectrum of $\r_{\b}$ majorizes the maximally mixed state $\eta$. Hence, by applying a local unitary, we can rearrange the eigenvalues of $\quW^2_{(i)}$ in the same order as that of $\r_{\b}$ such that when applied to both $\r_{\b}$ and $\eta$, we have $\tr[\quW^2_{(i)} U \r_{\b}U^{\dag}]\geq \tr[\quW^2_{(i)}\eta]$. Formally, we show that

\begin{prop}[Lower bound on the local variance, informal]\label{prop:Lower bound on the local variance, informal}
There exists a unitary $U$ supported on $\orderof{1}$ sites such that
$$\Tr\left[\quW_{(i)}^2U\r_{\b} U^{\dagger}\right] \geq \Tr\left[\quW_{(i)}^2\eta\right],$$
where $\eta$ is the maximally mixed state or the infinite temperature Gibbs state.
\end{prop}
In summary, starting from \eqref{eq:s20} and following \propref{Invariance under local unitaries, informal} and \propref{Lower bound on the local variance, informal}, the lower bound on the global variance takes the following local form:
$$\Tr\left[\quW^2\r_{\b}\right]\geq \br{\Tr\left[\quW_{(i)}^2\eta\right]}^{\orderof{1}}.$$
Lower bounding the quantity $\Tr\left[\quW_{(i)}^2\eta\right]$ by a constant is now an easier task, which we explain in more details later in Lemma \ref{lem:Wprimelowb} and Theorem \ref{claim_W_norm_W'_norm}. 

\section{Further discussions}
\subsection{Connection to previous work}\label{sec:Connection to previous work}

A similar problem to $\HLP$ known as the shadow tomography has been considered before \cite{Aaronson_shadow_tomography,Aaronson_shadow_tomography_2,Brandao_sdp2} where instead of the coefficients $\m_\ell$, we want to find a state~$\s$ that approximately matches $\tr[E_\ell\s]\approx_{\e} \tr[E_\ell \r]$ given multiple copies of an unknown state $\r$. It was shown  $\poly(\log m,\log d^n,1/\e)$ copies of $\r$ are sufficient for tomography. The $\HLP$ differs from the shadow tomography problem. Our goal is to estimate the Hamiltonian (i.e. the coefficients $\m_\ell$) within some given error bound. The shadow tomography protocol only concerns with estimating the marginals $\tr[E_\ell \r]$ up to a fixed error and by itself does not imply a bound on the Hamiltonian. Moreover, since the Hamiltonians we consider are spatially local, we only need to measure \emph{local} observables~$E_\ell$. This means we do not need to rely on the whole machinery of the shadow tomography which is applicable even when $E_\ell$ are non-local. We instead use a variant of this method introduced in \cite{Kueng_shadows} or other approaches such as those in \cite{cotler2020tomographyCommuting,bonet2019tomographyCommuting} to estimate $\tr[E_{\ell} \r_{\b}]$.

There have been a number of proposals for the $\HLP$ in the past. In \cite{Aradlearning,Flammia2019BaysianHamLearning,Qi2019learningFromGroundState} learning the Hamiltonian from local measurements is considered. Their approach is based on setting up a linear system of equations whose constraints (i.e., the matrix of coefficients) are determined from the measurement outcomes. The solution of these equations is the parameter $\m_k$ of the Hamiltonian. The sample complexity in this approach depends inverse polynomially on the ``spectral gap'' of the matrix of coefficients which thus far has not been rigorously bounded. Another line of work considers learning the Hamiltonian using a trusted quantum simulator \cite{Wiebe2014hamiltonian,wiebe2014LearningImperfect,Verdon2019HamLearning} which is analyzed using a combination of numerical evidence and heuristic arguments. Amin et al.~\cite{quantum_Boltzmann_machines} quantized  classical Boltzmann machines and proposed a method to train and learn quantum Boltzmann machines using gradient~descent. 

As mentioned earlier, there has been a fruitful series of works on the classical analog of the $\HLP$ (see e.g. \cite{Bresler_learning,Klivans_learning,Interaction_screening}). In our work, we assume it is a priori known that the interaction graph of the Hamiltonian is spatially local. We then estimate the parameters in $\ell_2$-norm using $\poly(n)$ samples which is polynomially tight even for classical Hamiltonians. If we instead consider estimation in $\ell_{\infty}$-norm, the classical algorithms can achieve a stronger result. That is, given $\mathcal{O}(\log n)$ samples, they succeed in efficiently learning the structure of the underlying graph and its parameters in $\ell_{\infty}$-norm. If we apply our current analysis to this setup, we cannot improve our $\poly(n)$ sample complexity to $\mathcal{O}(\log n)$. This is in part because the classical results devise a more efficient convex program that learns the parameters node-wise (this relies on the commutativity of the Hamiltonian terms), and partly because their required strong convexity assumptions is based on the Markov property, none of which are known to be quantizable. 

 \subsection{Open questions}\label{sec:Discussion and open questions}
In \secref{Maximum entropy estimation and sufficient statistics} we explained our approach to analyzing the $\HLP$ based on reducing data to its sufficient statistics and using maximum entropy estimation. An issue with this approach is the blowup in the computationally complexity. It is shown in  \cite{Montanari_sufficient_statistics} that this approach basically requires approximating the partition function which is $\NP$-hard. Ideally, one would like to have an algorithm for the $\HLP$ that requires small number of samples, but also has an efficient running time. Satisfying both these constraints for all inverse-temperatures $\b$ even in the classical learning problems is quite challenge. It was only recently that more efficient algorithms are devised for learning graphical models \cite{Klivans_learning,Interaction_screening}. In this work, we focus on the less demanding but still non-trivial question of bounding the sample complexity and leave obtaining an efficient running time for future work. Below we mention some of the open problems in this direction.

Our lower bound on the variance in \secref{Lower bounding the variance} is obtained for any constant inverse-temperature $\b$. It is an interesting open question to improve this bound, ideally to a constant independent of system size, assuming physically-motivated conditions such as the decay of correlations or the decay of conditional mutual information. Another approach might be to derive such a bound at high temperatures where powerful tools such as cluster expansions are available \cite{Kohtaro_cmi_cluster}. We also expect our bounds can be improved for commuting Hamiltonians. Indeed, using structural results such as \cite{bravyi2003commutative,Eldar_commuting}, one should be able to follow the same strategy as in \secref{Strong convexity of of log z classical} to find a constant lower bound on the variance of commuting Hamiltonians.

There are recent results on efficiently computing the partition function of quantum many-body systems under various assumptions \cite{Bravyi_ferro,SoleimanifarComplexZeros,Kohtaro_cmi_cluster}. We expect by combining these results with our maximum entropy estimation algorithm in \secref{Maximum entropy estimation and sufficient statistics}, one can obtain efficient classical algorithms for the $\HLP$. Another approach might be to use calibrated quantum computers (or Gibbs samplers) as in \cite{Brandao_gibbs_preparing,Brandao_sdp2} to solve the maximum entropy optimization using multiplicative weight update method and learn the parameters of another quantum~device.

Finally, an important future direction is to devise more refined objective functions for the $\HLP$ that matches the performance of the learning algorithms for the classical problem as discussed in \secref{Connection to previous work}. Given the non-commutative nature of quantum Hamiltonians, this seems to require substantially new ideas and advances in characterizing the information theoretic properties of the quantum Gibbs states. 

\paragraph{Acknowledgements.}
We thank Aram Harrow, Yichen Huang, Rolando La Placa, Sathyawageeswar Subramanian, John Wright, and Henry Yuen for helpful discussions. Part of this work was done when SA and TK were visiting Perimeter Institute. SA was supported in part by the Army Research Laboratory and the Army Research Office under grant number W1975117. AA is supported by the Canadian Institute for Advanced Research, through funding provided to the Institute for Quantum Computing by the Government of Canada and the Province of Ontario. Perimeter Institute is also supported in part by the Government of Canada and the Province of Ontario. TK was supported by the RIKEN Center for AIP and JSPS KAKENHI Grant No. 18K13475. MS was supported by NSF grant CCF-1729369 and a Samsung Advanced Institute of Technology Global Research Partnership.  
\section{Preliminaries}

\subsection{Some mathematical facts}
Here we summarize some of the basic mathematical facts used in the proof.  Let $A,B$ be arbitrary operator. The \emph{operator norm} of $A$ which is its largest singular value is denoted by $\|A\|$.  We also often use the  \emph{Frobenius norm} $\|A\|_F:= \sqrt{\Tr[A^{\dagger}A]}$ and more generally the Hilbert-Schmidt inner product between $A,B$ defined by $\tr[A^{\dag}B]$. Additionally using H\"older's inequality we have, 
\begin{align}
\label{fact:relatingfrobeniusAB}
\|A B\|_F= \sqrt{ \tr (B^\dagger A A^\dagger B) } \le \sqrt{\|B\|^2 \tr (A A^\dagger) }= \|B\| \cdot \|A\|_F.
\end{align}
We define the von Neumann entropy of a quantum state $\s$ by $S(\s)=-\tr[\s \log \s]$ and the relative entropy between two states $\s_1$ and $\s_2$ by $S(\s_1\| \s_2)=-\tr[\s_1\log \s_2]-S(\s_1)$. 

The gradient of a real function $f:\bbR^m\mapsto\bbR$ is denoted by $\nabla f(x)$ and its \emph{Hessian} (second derivative) matrix by $\nabla^2 f(x)$. The entries of the Hessian matrix are given by $\frac{\partial^2}{\partial x_i \partial x_j} f(x)$.

We write $A\succeq 0$ to represent a \emph{positive semi-definite} (PSD) operator $A$, one such example of a PSD operator is the  Hessian matrix $\nabla^2 f(x)$.

For convenience, we will also gather a collection of infinite sums over exponentials. For $t>0$,~let 
$$\Gamma(t):= \int_{0}^{\infty}x^{t-1}e^{-x}dx = \frac{1}{t}\int_{0}^{\infty}e^{-x}d\br{x^t}= \frac{1}{t}\int_{0}^{\infty}e^{-y^{\frac{1}{t}}}dy$$
be the \emph{gamma function}. It holds that $\Gamma(t)\leq t^t$. This can be used to simplify several summations that we encounter later. Finally, we collect a few useful summations that we use in our proofs in the following fact. The proof is postponed until \appref{proof of integral facts}.

\begin{fact}
\label{fact:integrals}
Let $a,c,p>0$ be reals and $b$ be a positive integer. Then
\begin{enumerate}
\item[1)] $\sum_{j=0}^{\infty} e^{-cj} \leq \frac{e^c}{c}$.
\item[2)] $\sum_{j=0}^{\infty} j^be^{-cj^p} \leq \frac{2}{p}\cdot\br{\frac{b+1}{cp}}^{\frac{b+1}{p}}$.
\item[3)] $\sum_{j=0}^{\infty} e^{-c(a+j)^p} \leq e^{-\frac{c}{2}a^p}\br{1+\frac{1}{p }\br{\frac{2}{cp}}^{\frac{1}{p}}}$.
\end{enumerate}
\end{fact}

\subsection{Local Hamiltonians and quantum Gibbs states}\label{sec:Local Hamiltonians and quantum Gibbs states}
\paragraph{Local Hamiltonians.} In this work, we focus on Hamiltonians that are \emph{geometrically local}. That is, the interactions terms in the Hamiltonian act on a constant number of qudits that are in the neighborhood of each other. To describe this notion more precisely, we consider a $D$-dimensional lattice $\L \subset \mathbb{Z}^D$ that contains $n$ sites with a $d$-dimensional qudit (spin) on each site. We denote the dimension of the Hilbert space associated to the lattice $\Lambda$ by $\mathcal{D}_{\Lambda}$. The Hamiltonian of this system is 
\ba
H=\sum_{X\subset \L} H_X.\nn
\ea
Each term $H_X$ acts only on the sites in $X$ and $X$ is restricted to be a connected set with respect to $\L$. We also define the Hamiltonian restricted to a region $A\subseteq~\L$ by $H_A=\sum_{X \subseteq A} H_X$. Let $B(r,i):=\{j\in \Lambda| \dist(i,j)\le r\}$ denotes a ball (under the Manhattan distance in the lattice) of size~$r$ centered at site $i$. For a given connected set $X\in \Lambda$, let $\diam(X):=\max\{\dist(i,j):i,j\in X\}$ denote the \emph{diameter} of this set, $X^c:= \Lambda\setminus X$ denote the complement of this set and $\partial X$ denote its boundary. Given two sets $X,Y \in \Lambda$, we define $\dist(X,Y):=\min\br{\dist(i,j): i\in X, j\in Y}$. 

In order describe our Hamiltonians, we consider an orthogonal Hermitian basis for the space of operators acting on each qudit. For instance, for qubits this basis consists of the Pauli operators. By decomposing each local term $H_X$ in terms of the tensor product of such basis operators, we find the following canonical form for the Hamiltonian $H$:
\begin{definition}[Canonical representation for $\k$-local Hamiltonians]\label{def:geometrically-local Hamiltonians} 
A $\k$-local Hamiltonian $H$ on a lattice $\L$ is sum of $m$ Hermitian operators $E_\ell$ each acting non-trivially on $\k$ qudits. That is,
\ba
H=\sum_{\ell=1}^m \m_\ell E_\ell.
\ea
where $\m_\ell\in \bbR$ and we assume $\norm{E_\ell}\leq 1$, $\Tr[E_\ell^2]=\mathcal{D}_\Lambda$, $E_\ell^{\dagger}=E_\ell$ for $\ell\in[m]$, and 
\ba
\tr[E_k E_{\ell}]= 0\quad \text{for}\ k\neq \ell.\label{eq:s9}
\ea
\end{definition}
Since $H$ is geometrically local, it holds that $m=\orderof{|\L|}=\orderof{n}$. As discussed earlier, we extensively use the notion of {quasi-local operators}, which we now formally define.
\begin{definition}[Quasi-local operators]\label{def:Quasi-local operators}
An operator $A$ is said to be \emph{$(\tau, a_1, a_2, \zeta)$-quasi-local} if it can be written as 
\begin{align}
&A= \sum_{\ell=1}^n  {g}_{\ell}  \bar{A}_\ell  \quad {\rm with } \quad {g}_{\ell}  \le a_1\cdot  \exp(-a_2 \ell^\tau), \notag \\
&\bar{A}_\ell=\sum_{|Z|= \ell} a_Z, \quad  \max_{i\in \Lambda}\left( \sum_{Z: Z\ni i} \|a_Z\| \right) \leq \zeta,
\label{quasi_locality of A}
\end{align}
where the sets $Z\subset \Lambda$ are restricted to be balls.\footnote{The assumption that $Z$ is a ball suffices for us. Our results on quasi-local operators also generalize to the case where~$Z$ is an arbitrary \textit{regular} shape, for example when the radii of the balls inscribing and inscribed by $Z$ are of constant proportion to each other.}
\end{definition}
Although local operators are morally a special case of quasi-local operators (when $\tau=\infty$), we will reserve the above notation for operators with $\tau=\mathcal{O}(1)$.
A useful tool for analyzing quasi-locality is the Lieb-Robinson bound, which shows a light-cone like behavior of the time evolution operator.
\begin{fact}[Lieb-Robinson bound~\cite{Lieb1972}, \cite{10.1007/978-90-481-2810-5_39}]
\label{fact:LRB}
Let $P,Q$ be operators supported on regions $X, Y$ of the $D$ dimensional lattice $\Lambda$ respectively. Let $H$ be a $(\z,\k)$-geometrically local Hamiltonian. There exist constants $\vL, f, c$ that only depend on $\z,\k$ and $D$ such that 
$$\|[e^{iHt}Ae^{-iHt}, B]\|\leq f\|A\|\|B\|\cdot\min\br{|\partial X|, |\partial Y|}\cdot \min\br{e^{c\br{\vL|t|-\dist(X,Y)}},1}.$$
\end{fact}

\paragraph{Gibbs states.} At an \emph{inverse-temperature $\b$}, a quantum many-body system with the Hamiltonian $H(\m)$ is in the Gibbs (thermal) state 
\ba
\r_{\b}(\m)=\frac{e^{-\b H(\m)}}{\tr[e^{-\b H(\m)}]}.\label{eq:s15}
\ea
The partition function of this system is defined by $Z_{\b}(\m)=\tr[e^{-\b H(\m)}]$.
\begin{rem} In our notation, we sometimes drop the dependency of the partition function or the Gibbs state on $\m$. We also often simply use the term local Hamiltonian $H$ or quasi-local operator $A$ when referring to \defref{geometrically-local Hamiltonians} and \defref{Quasi-local operators}.
\end{rem}
As discussed earlier, local marginals of the Gibbs states can be used to uniquely specify them. This provides us with ``sufficient statistics'' for learning the Hamiltonians from $\r_{\b}$. More precisely, we have:
\begin{prop}[Restatement of \propref{matching marginals}]
Consider the following two Gibbs states
\begin{align}
    \r_{\b}(\m)=\frac{e^{-\b\sum_\ell \m_\ell E_\ell}}{\Tr[e^{-\b\sum_\ell \m_\ell E_\ell}]},\quad  \r_{\b}(\l)=\frac{e^{-\b\sum_\ell \l_\ell E_\ell}}{\Tr[e^{-\b\sum_\ell \l_\ell E_\ell}]} \label{eq:b2}
\end{align}
such that $\Tr[\r_{\b}(\l) E_j]=\Tr[\r_{\b}(\m) E_j]$ for all $j\in [m]$, i.e. all the $\k$-local marginals of $\r_{\b}(\l)$ match that of $\r_{\b}(\m)$. Then, we have $\r_{\b}(\l)=\r_{\b}(\m)$, which in turns implies $\l_\ell=\m_\ell$ for $\ell\in [m]$. 
\end{prop}
\begin{proof}
We consider the relative entropy between $\r_{\b}(\l)$ and the Gibbs state $\r_{\b}(\m)$. We have
\begin{align}
    S\left(\r_{\b}(\m)\| \r_{\b}(\l)\right)&=\Tr\left[\r_{\b}(\m)\left(\log \r_{\b}(\m)- \log \r_{\b}(\l)\right)\right] \nn\\
    &=-S(\r_{\b}(\m))+\b\cdot \Tr\left[\r_{\b}(\m)\sum_\ell \l_\ell E_\ell\right]+\log Z(\l)\\
    &\overset{(1)}=-S(\r_{\b}(\m))+\b\sum_\ell \l_\ell \Tr[\r_{\b}(\l) E_\ell]+\log Z(\l) \\
    &=-S(\r_{\b}(\m))+S(\r_{\b}(\l))\nn\\
    &\overset{(2)}\geq 0, && \label{eq:7}
\end{align}
where $(1)$ follows because $\Tr[\r_{\b}(\m) E_\ell]=\Tr[\r_{\b}(\l) E_\ell]$ for all $\ell\in [m]$ and $(2)$ used the positivity of relative entropy. 
Similarly, we can exchange the role of $\r(\m)$ and $\r(\l)$ in \eqref{eq:7} and get
\begin{align}
    S\left(\r_{\b}(\l)\|\r_{\b}(\m)\right)=-S(\r_{\b}(\l))+S(\r_{\b}(\m))\geq 0.
\end{align}
Combining these bounds imply $S(\r_{\b}(\m))=S(\r_{\b}(\l))$ and hence from Eq.~\eqref{eq:7}, we get $S(\r_{\b}(\m)\| \r_{\b}(\l))=~0$. It is known that the relative entropy of two distribution is zero only when $\r_{\b}(\m)=\r_{\b}(\l)$. Hence, we also have $\log \r_{\b}(\m) = \log \r_{\b}(\l)$ or equivalently up to an additive term $\sum_{\ell=1}^m \m_\ell E_\ell=\sum_{\ell=1}^m \l_\ell E_\ell$. Since the operators $E_\ell$ form  an orthogonal basis (see Eq.~\eqref{eq:s9}), we see that $\l_\ell=\m_\ell$ for all $\ell\in[m]$.
\end{proof}
\begin{rem}
When the marginals of the two Gibbs states only approximately match, i.e., 
$$
|\tr[\r_{\b}(\m) E_\ell]-\tr[\r_{\b}(\l) E_\ell]|\leq \e
$$ 
for $\ell\in [m]$, then a similar argument to \eqref{eq:7} shows that $S(\r_{\b}(\m)\| \r_{\b}(\l))\leq \mathcal{O}(m\e)$. By applying Pinsker's inequality, we get $\norm{\r_{\b}(\m)-\r_{\b}(\l)}^2_1\leq \mathcal{O}(m \e)$.\footnote{Pinsker's inequality states that for two density matrices $\rho,\sigma$, we have $\|\rho-\sigma\|_1^2\leq 2\ln 2\cdot S(\rho\|\sigma)$.}
\end{rem}
\subsection{Quantum belief propagation} 
Earlier we saw that we could express the Gibbs state of a Hamiltonian $H$ by Eq.~\eqref{eq:s15}. Suppose we alter this Hamiltonian by adding a term $V$ such that
\begin{align}
    H(s)=H+sV,\quad s\in [0,1].
\end{align} 
How does the Gibbs state associated with this Hamiltonian change? If the new term $V$ commutes with the Hamiltonian $H$, i.e., $[H,V]=0$, then the derivative of the Gibbs state of $H(s)$ is given by
\begin{align}
 \frac{d}{ds} e^{-\b H(s)}=-\b e^{-\b H(s)}V= -\frac{\b}{2}\left\{e^{-\b H(s)},V\right\},
\end{align}
where $\{e^{-\b H(s)},V\}=e^{-\b H(s)}V+V e^{-\b H(s)}$ denotes the anti-commutator. In the non-commuting case though, finding this derivative is more complicated. The \emph{quantum belief propagation} is a framework developed in \cite{hastings_belief_propagation,Kim_gibbs,Brandao_gibbs_cmi} for finding such derivatives in a way that reflects the locality of the system. 

\begin{definition}[Quantum belief propagation operator] \label{def:r1}
For every $s\in[0,1]$, $\beta\in \mathbb{R}$, define $H(s)=H+s V$ where $V=\sum_{j,k} V_{j,k}\ketbra{j}{k} $ is a Hermitian operator. Also let $f_{\b}(t)$ be a function whose Fourier transform is 
\ba
\tilde{f}_{\b}(\omega)=\frac{\tanh(\b \omega/2)}{\b \omega/2},
\label{def:tilde_f_beta}
\ea
i.e., $f_{\b}(t)=\frac{1}{2\pi} \int d\omega \tilde{f}_{\b}(\omega) e^{i\omega t}$. The quantum belief propagation operator $\Phi_{H(s)}(V)$ is defined by
\ba
\Phi_{H(s)}(V)&=\int_{-\infty}^{\infty} dt f_{\b}(t)\ e^{-iH(s)t}\ V \ e^{iH(s)t}.\nn
\ea
Equivalently, in the energy basis of $H(s)=\sum_{j}\ep_j(s)\ \ketbra{j}{j}$, we can write
\ba
\Phi_{H(s)}(V)&=\sum_{j,k} \ketbra{j}{k}\  V_{j,k}\ \tilde{f_\b}(\Ene_j(s)-\Ene_k(s)).
\ea
\end{definition}
\begin{prop}[cf. \cite{hastings_belief_propagation}] \label{prop:QBP}
In the same setup as \defref{r1}, it holds that
\ba
\frac{d}{ds}e^{-\b H(s)}=-\frac{\b}{2}\left\{e^{-\b H(s)},\Phi_{H(s)}(V)\right\}.
\ea
\end{prop}

\subsection{Change in the spectrum after applying local operators}
For a Hamiltonian $H$, let $P^H_{\le x}$ and $P^H_{\ge y}$ be projection operators onto the eigenspaces of $H$ whose energies are in $\le x$ and $\ge y$, respectively (we use similar notation $P^A_{\le x}, P^A_{\ge y}$ for the quasi-local operator $A$). Consider a quantum state $\ket{\psi}$ in the low-energy part of the spectrum such that $P^H_{\le x}\ket{\psi}=\ket{\psi}$. Suppose this states $\ket{\psi}$ is perturbed by applying a local operator~$O_X$ on a subset $X\subset \Lambda$ of its qudits. Intuitively, we expect that the operator $O_X$ only affects the energy of $\ket{\psi}$ up to $\orderof{|X|}$, i.e., $\| P^H_{\ge y} O_X\ket{\psi}\| \approx 0$ for $y\gg x+ |X|$. A simple example is when~$\ket{\psi}$ is the eigenstate of a classical spin system. By applying a local operation that flips the spins in a small region $X$, the energy changes at most by $\mathcal{O}(|X|)$. The following lemma rigorously formulates the same classical intuition for quantum Hamiltonians.
\begin{lem}[Theorem~2.1 of \cite{Arad_connecting_global_local_dist}]
\label{lem:AKL16}
Let $H$ be an arbitrary $\kappa$-local operator such that 
\ba
H = \sum_{|Z|\le \kappa} h_Z,
\ea
and each $i\in \L$ supports at most $g$ terms $h_Z$. Then, for an arbitrary operator $O_X$ which is supported on $X\subseteq \Lambda$, the operator norm of $P^H_{\ge y} O_X P^H_{\le x}$ is upper-bounded by 
\ba
\| P^{H}_{\ge y} O_X P^{H}_{\le x}\| \le \|O_X\| \cdot \exp \big(-\frac{1}{2g\kappa}(y-x - 2g|X|)\big) .
\ea
\end{lem}
In our analysis, we need an different version of this lemma for \emph{quasi-local} operators instead of $\k$-local operators. The new lemma will play a central role in lower-bounding the variance of {quasi-local operators}. 
The proof follows by the analysis of a certain moment function (as opposed to the moment generating function in \cite{Arad_connecting_global_local_dist}). Due to formal similarities between the proofs, we defer the proof of the next lemma to Appendix~\ref{sec:Derivation of norm quasi local}. 
\begin{lem}[Variation of \cite{Arad_connecting_global_local_dist} for quasi-local operators]
\label{multicommutator_norm_quasi_local}
Let $A$ be a $(\tau, a_1, a_2, 1)$-quasi-local operator, as given in Eq.~\eqref{quasi_locality of A}, with $\tau\le 1$. For an arbitrary operator $O_X$ supported on a subset $X\subseteq \Lambda$  with $|X|=k_0$ and $\|O_X\|=1$, we have
\begin{align}
\| P^{A}_{\ge x+y} O_X  P^{A}_{\le x}  \|  \le c_5\cdot  k_0 \exp \Big(-(\lambda_1 y/k_0)^{1/\tau_1}\Big), 
\label{Concentration_lemma_subexponential1}
\end{align}
where $\tau_1:=\frac{2}{\tau}-1$,  $c_5$ and $\lambda_1$ are constants depending on $a_1$ and $a_2$  as $c_5 \propto a_2^{2/\tau}$ and $\lambda_1\propto a_2^{-2/\tau}$ respectively.
\end{lem}

\subsection{Local reduction of global operators}
\label{global_to_local}
An important notion in our proofs will be a reduction of a global operator to a local one, which has influence on a site $i$. Fix a subset $Z\subseteq \Lambda$ and an operator $O$ supported on $Z$. Define 
\begin{align}
    \label{defnofOi}
O_{\locci}:= O- \tr_i[O]\otimes \frac{\iden_i}{d}
\end{align}
where operator $\iden_i$ is the identity operator on the $i$th site, $d$ is the local dimension, $\tr_i$ is the partial trace operation with respect to the site $i$. Note that $O_\locci$ removes all the terms in $O$ that do not act on the $i$th site. This can be explicitly seen by introducing a basis $\{E^{\alpha}_Y\}_{\alpha \in \mathbb{N},Y\in Z}$ of Hermitian operators, where $Y$ labels the support of $E^{\alpha}_Y$ and $\alpha$ labels several possible operators on the same support. We can assume that $\Tr[\br{E^{\alpha}_Y}^2]=1, \Tr_i[E^{\alpha}_Y]=0$ for every $i\in Y$, and the orthogonality condition $\Tr[E^{\alpha}_Y E^{\alpha'}_{Y'}]=0$ holds if $\alpha\neq \alpha'$ or $Y\neq Y'$. These conditions are satisfied by the appropriately normalized Pauli operators.~Expand
$$
O= \sum_{\alpha, Y} g_{\alpha, Y}E^{\alpha}_Y.
$$
Then
$$
O_\locci=\sum_{\alpha, Y} g_{\alpha, Y}E^{\alpha}_Y - \sum_{\alpha, Y} g_{\alpha, Y}\Tr_{i}[E^{\alpha}_Y]\otimes \frac{\iden_i}{d}=\sum_{\substack{\alpha, Y: Y\ni i}} g_{\alpha, Y}E^{\alpha}_Y.
$$
Thus, $O_\locci$ is an operator derived from $O$, by removing all $E^{\alpha}_Y$ which act as identity on $i$. The following claim shows that the Frobenius norm of a typical $O_\locci$ is not much small in comparison to the Frobenius norm of $O$.  
\begin{claim} \label{claim_global_norm_local_norm}
For every operator $O$ and $O_\locci$ defined in Eq.~\eqref{defnofOi}, it holds that
\begin{align} 
\max_{i\in Z}\|O_\locci\|_F^2 \ge \frac{1}{|Z|}\sum_{i\in Z}\|O_\locci\|^2_F \ge \frac{1}{|Z|} \|O\|_F^2,
\label{main ineq_claim_global_norm_local_norm}
\end{align}
\end{claim}
\begin{proof}
Using the identities $\Tr[E^{\alpha}_YE^{\alpha'}_{Y'}]=0$ and $\Tr[\br{E^{\alpha}_Y}^2]=1$, we have
\begin{align*} 
\|O\|_F^2 =\sum_{Y,\alpha} g^2_{\alpha, Y} \le \sum_{i\in Z} \sum_{\alpha, Y:Y\ni i}g^2_{\alpha, Y} =\sum_{i\in Z}  \| O_\locci \|_F^2.
\end{align*}
This completes the proof.
\end{proof}

\subsection{Stochastic convex optimization applied to Hamiltonian learning}\label{sec:Stochastic convex optimization_prelim}
Suppose we want to solve the optimization 
$$
\max_{x\in \bbR^m} f(x)
$$ 
for a function $f:\bbR^m\rightarrow \bbR$ which is of the form $f(x)=\bbE_{y\sim \cD}[g(x,y)]$. Here $g(x,y)$ is some convex function and the expectation $\bbE_{y\sim \cD}$ is taken with respect to an unknown distribution $\cD$. Algorithms for this maximization problem  are based on obtaining i.i.d.~\emph{samples} $y$ drawn from the distribution~$\cD$. In practice, we can only receive finite samples $y_1,y_2,\dots,y_\ell$ from such a distribution. Hence, instead of the original optimization, we solve an \emph{empirical} version 
$$
\max_{x\in \bbR^m} \frac{1}{\ell}\sum_{k=1}^\ell  g(x,y_k).
$$
The natural question therefore is: How many samples $\ell$ do we need to guarantee the output of the empirical optimization is close to the original solution? One answer to this problem relies on a property of the objective function known as \emph{strong convexity}. 
\begin{definition}[Restatement of \defref{strong convexity}]
Consider a convex function $f:\bbR^{m}\mapsto \bbR$ with gradient $\nabla f(x)$ and Hessian $\nabla^2 f(x)$ at $x$. The function $f$ is said to be $\a$-strongly convex in its domain if it is differentiable and for all $x,y$, and
\begin{align}
    f(y)\geq f(x)+\nabla f(x)^\top (y-x) +\frac{1}{2}\a \norm{y-x}^2_2,\nn
\end{align}
or equivalently if its Hessian satisfies
\begin{align}
    \nabla^2 f(x)\succeq \a \iden.
\end{align}
In other words, for any vector $v \in \bbR^m$ it holds that $\sum_{i,j} v_i v_j \frac{\partial^2}{\partial x_i \partial x_j}f(x)\geq \a \norm{v}^2_2$.
\end{definition}
Next, we discuss how the framework of convex optimization and in particular strong convexity, can be applied to the $\HLP$. To this end, we define the following optimization problems. 
\begin{definition}[Optimization program for learning the Hamiltonian]\label{def:dual functions}
We denote the objective function in the $\HLP$ and its approximate version (equations \eqref{eq:4} and \eqref{eq:9} in \secref{Maximum entropy estimation and sufficient statistics}) by $L(\l)$ and $\hat{L}(\l)$ respectively, i.e.,
\ba
L(\l)=\log Z_{\b}(\l)+\beta\cdot\sum_{\ell=1}^m \l_\ell e_\ell,\quad 
\hat{L}(\l)=\log Z_{\b}(\l)+\beta\cdot\sum_{\ell=1}^m \l_k \hat{e}_\ell,
\ea
where the partition function is given by $Z_{\b}(\lambda)=\tr\big(e^{-\b \sum_{\ell=1}^m  \lambda_\ell E_\ell}\big)$. The parameters of the Hamiltonian that we intend to learn are $\m=\argmin_{\l\in \bbR^m: \|\l\|\leq 1} L(\l)$. As before, we also define the empirical version of this optimization by
\ba
\hat{\m}=\argmin_{\l\in \bbR^m: \norm{\l}\leq 1} \hat{L}(\l).\label{eq:b9}
\ea
\end{definition}
We prove later in \lemref{variance lower bound on Hessian} that $\log Z_{\b}(\l)$ is a convex function in parameters $\l$ and thus, the optimization in \eqref{eq:b9} is a convex program whose solution can be in principle found. In this work, we do not constraint ourselves with the running time of solving \eqref{eq:b9}. We instead obtain sample complexity bounds as formulated more formally in the next theorem.
\begin{thm}[Error in $\pmb{\m}$ from error in marginals $\pmb{e_\ell}$]\label{thm:required error in marginals}
Let $\delta,\alpha > 0$. Suppose the marginals $e_\ell$ are determined up to error $\d$, i.e., $|e_\ell-\hat{e}_\ell|\leq \d$ for all $\ell\in[m]$. Additionally assume $\nabla^2 \log Z(\l) \succeq \a \iden$ for $\norm{\l}\leq 1$. Then the optimal solution to the program \eqref{eq:b9} satisfies 
$$
\norm{\m-\hat{\m}}_2\leq \frac{2\b\sqrt{m}\d}{\alpha}
$$
\end{thm}

\begin{proof}
From the definition of $\hat{\m}$ as the optimal solution of $\hat{L}$ in \eqref{eq:b9}, we see that $\hat{L}(\hat{\m})\leq \hat{L}(\m)$. Thus, we get
\ba
\log Z_{\b}(\hat{\m}) +\beta\cdot \sum_{\ell=1}^m \hat{\m}_\ell \hat{e}_\ell \leq \log Z_{\b}(\m) +\beta\cdot \sum_{\ell=1}^m \m_\ell \hat{e}_\ell.\nn
\ea
or equivalently,
\ba
\log Z_{\b}(\hat{\m}) \leq \log Z_{\b}(\m)+\beta\cdot\sum_{\ell=1}^m (\m_\ell-\hat{\m}_\ell) \hat{e}_\ell.\label{eq:b11}
\ea
We show later in \lemref{expression for Hessia} that for every $\ell\in [m]$, we have $\frac{\partial}{\partial \m_\ell}\log Z_{\b}(\m)=-\b e_\ell$.\footnote{In particular, see Eq.~\eqref{eq:firstderivativeofpartition}, where we showed $\frac{\partial}{\partial \m_\ell}\log Z_{\b}(\m)=-\beta \cdot \Tr[E_\ell \rho_\beta(\m)]=-\beta e_\ell$.} This along with the assumption $\nabla^2\log Z(\m)\succeq \alpha\iden$ in the theorem statement, implies that for every $\m'$ with $\norm{\m'}\leq 1$ 
\ba
\log Z_{\b}(\m')\geq \log Z_{\b}(\m)-\beta\cdot \sum_{\ell=1}^m (\m'_\ell-\m_\ell)e_\ell  +\frac{1}{2}\a \norm{\m'-\m}_2^2.\label{eq:b10}
\ea
Hence, by choosing $\m'=\hat{\m}$ and combining \eqref{eq:b10} and \eqref{eq:b11}, we get
$$
\log Z_{\b}(\m)-\beta\cdot \sum_{\ell=1}^m (\hat{\m}_\ell-\m_\ell)e_\ell  +\frac{1}{2}\a \norm{\hat{\m}-\m}_2^2\leq \log Z_{\b}(\m)+\beta\cdot\sum_{\ell=1}^m (\m_\ell-\hat{\m}_\ell) \hat{e}_\ell
$$
which further implies that
\ba
\frac{1}{2}\a \norm{\hat{\m}-\m}_2^2 &\leq \beta\cdot \sum_{\ell=1}^m (\hat{\m}_\ell-\m_\ell) (e_\ell-\hat{e}_\ell)\nn,\\
&\leq \beta\cdot\norm{\hat{\m}-\m}_2\cdot \norm{\hat{e}-e}_2.\nn
\ea
Hence, we have
\ba
\norm{\hat{\m}-\m}_2\leq \frac{2\b}{\a}\norm{\hat{e}- e}_2\leq \frac{2\b\sqrt{m} \d}{\a}.\nn
\ea
\end{proof}

\begin{cor}[Sample complexity from strong convexity]
Under the same conditions as in \thmref{required error in marginals}, the number of copies of the Gibbs state $\r_{\b}$ that suffice to solve the $\HLP$~is
\ba
N=O\left(\frac{\b^2 2^{\mathcal{O}(\k)}}{\a^2 \e^2}m \log m\right)\nn.
\ea
\end{cor}
\begin{proof}
First observe that, using Theorem~\ref{thm:required error in marginals}, as long as the error in estimating the marginals $e_\ell$ are
\ba
\d \leq \frac{\a \e}{2\b \sqrt{m}},\label{eq:s19}
\ea
we estimate the coefficients $\m$ by $\hat{\m}$ such that $\norm{\hat{\m}-\m}_2\leq \e$. 
The marginals $e_\ell$ can be estimated in various ways. One method considered in \cite{cotler2020tomographyCommuting,bonet2019tomographyCommuting} is to group the operators $E_\ell$ into sets of mutually commuting observables and simultaneously measure them at once. Alternatively, we can use the recent procedure in \cite[Theorem 1]{Kueng_shadows} based on a variant of shadow tomography. In either case, the number of copies of the state needed to find all the marginals with accuracy $\d$ is 
$$
N=O\left(\frac{2^{\mathcal{O}(\k)}}{\d^2}\log m\right),
$$
where recall that $\k$ is the locality of the Hamiltonian. Plugging in Eq.~\eqref{eq:s19} gives us the final bound
\ba
N=O\left(\frac{\b^2 2^{\mathcal{O}(\k)}}{\a^2 \e^2}m \log m\right)\nn.
\ea
\end{proof}
\section{Strong convexity of $\pmb{\log Z_{\b}(\lambda)}$}
\label{sec:Strong convexity of log z proofs}
 We now state our main theorem which proves the strong convexity of the logarithm of the partition function. Recall that  for a vector $\lambda=(\lambda_1,\ldots,\lambda_m)\in \bbR^m$, Hamiltonian  $H(\lambda)=\sum_i \lambda_i E_i$ where 
$E_i$ are tensor product of Pauli operators with weight at most $\k$ and $\rho_{\b}(\lambda)=\frac{1}{Z_{\b}(\lambda)}e^{-\beta H(\lambda)}$, we defined the partition function as $Z_\b(\lambda)=\tr(e^{-\beta H(\lambda)})$. 
We now prove our main theorem for this section.
\begin{thm}[Restatement of Theorem \ref{thm:strong convexity of log-partition function_informal}: $\pmb{\log\ Z(\lambda)}$ is strongly convex]\label{thm:strong convexity of log-partition function} Let $H=\sum_{\ell=1}^m \m_\ell E_\ell$ be a $\k$-local Hamiltonian over a finite dimensional lattice. For a given inverse-temperature $\b$, there are constants $c, c'>3$ depending on the geometric properties of the lattice  such that 
\ba
\nabla^2 \log Z_{\b}(\m) \succeq e^{-\mathcal{O}(\b^c)} \cdot \frac{\b^{c'}}{m} \cdot \iden, 
\ea
i.e., for every vector $v\in \bbR^m$ we have $v^T \cdot \nabla^2 \log Z_{\b}(\m) \cdot v \geq \b^{c'}e^{-\mathcal{O}(\b^c)}/m \cdot \norm{v}^2_2$.
\end{thm}

The proof of Theorem~\ref{thm:strong convexity of log-partition function}  is divided into multiple lemmas that we state and prove both in this and the following sections. We begin with finding an expression for the Hessian of $\log Z_{\b}(\l)$. 
\begin{lem}\label{lem:expression for Hessia}
For every vector $v\in \bbR^{m}$, define the local operator $W_v=\sum_{i=1}^m v_i E_i$ (for notational convenience, later on we stop subscripting $W$ by $v$). The Hessian $\nabla^2 \log Z_{\b}(\l)$ satisfies
\ba
v^\top\cdot\left(\nabla^2 \log Z_{\b}(\l)\right)\cdot v&= \frac{\beta^2}{2}\Tr\Big[\big\{W_v, \Phi_{H(\l)}(W_v)\big\}\rho_{\b}(\l)\Big] - \beta^2\big(\Tr\left[W_v\rho_{\b}(\l)\right]\big)^2,
\ea
\end{lem}
\begin{proof}

Since the terms in the Hamiltonian are non-commuting, we  use \propref{QBP} to find the derivatives of $\log Z_{\b}(\l)$. We get
\begin{align}
\label{eq:firstderivativeofpartition}
 \frac{\partial}{\partial \l_j} \log Z_{\b}(\l)&=\frac{1}{Z_{\b}(\l)}\tr\left[-\frac{\b}{2}\left\{e^{-\b H(\l)},\Phi_{H(\l)}(E_j)\right\}\right]\nn\\
    &=\frac{-\b}{Z_{\b}(\l)}\Tr\left[e^{-\b H(\l)}\int_{-\infty}^{\infty} dt f_{\b}(t)e^{-iH(\l)t}\ E_j\ e^{iH(\l)t}\right]\nn\\
    &=-\b\ \Tr \left[E_j \frac{e^{-\b H(\l)}}{Z_{\b}(\l)}\right],
\end{align}
where the second equality used the definition of the quantum belief propagation operator $\Phi_{H(\lambda)}(E_j)=\int_{-\infty}^{\infty} dt f_{\b}(t)\ e^{-iH(\lambda)t}\ E_j \ e^{iH(\lambda)t}$ with $f_\b$ as given in Definition~\ref{def:r1}. The third equality used the fact that $e^{iH(\lambda)t}$ commutes with $e^{-\beta H(\lambda)}$. Similarly, we have 
\begin{align}
    \frac{\partial^2}{\partial \l_k\partial \l_j} \log Z_{\b}(\l) &=-\b\ \Tr \left[E_j \cdot \frac{\partial}{\partial\l_k}\left(\frac{e^{-\b H(\l)}}{Z_{\b}(\l)}\right)\right]\nn\\
    &=-\b\ \Tr \left[E_j \cdot \frac{1}{Z_{\b}(\l)}\frac{\partial}{\partial\l_k}\left(e^{-\b H(\l)}\right)\right]+\b\Tr \left[E_j \cdot \frac{e^{-\b H(\l)}}{Z_{\b}(\l)}\right]\cdot \frac{1}{Z_{\b}(\l)}\frac{\partial}{\partial\l_k}Z_{\b}(\l)\nn\\
    &=\frac{\b^2}{2}\tr\left[ E_j\cdot\left\{ \r_{\b}(\l),\Phi_{H(\l)}(E_k)\right\}\right]-\b^2\ \tr[E_k \r_{\b}(\l)]\ \tr[E_j\r_{\b}(\l)]\nn\\
    &=\frac{\b^2}{2}\tr\left[\left\{ E_j ,\Phi_{H(\l)}(E_k)\right\}\cdot\r_{\b}(\l)\right]-\b^2\ \tr[E_k \r_{\b}(\l)]\ \tr[E_j\r_{\b}(\l)]\nn.
\end{align}
One can see from this equation that $\nabla^2 \log Z_{\b}(H)$ is a symmetric real matrix, \footnote{The terms $\tr[E_k \r_{\b}(\l)]$  and $\tr[E_j\r_{\b}(\l)]$ are real, being expectations of Hermitian matrices. Moreover, $\left\{E_j ,\Phi_{H(\l)}(E_k)\right\}$ is a Hermitian operator, being an anti-commutator of two Hermitian operators. Hence $\tr\left[ \left\{E_j ,\Phi_{H(\l)}(E_k)\right\}\r_{\b}(\l)\right]$ is real too.} and hence its eigenvectors have real entries. Finally, we get 
\ba
v^{\top}\cdot \left(\nabla^2 \log Z_{\b}(\l)\right) \cdot v&=\sum_{j,k} v_j v_k \frac{\partial^2}{\partial \l_k\partial \l_j} \log Z_{\b}(\l)\\ &=\frac{\beta^2}{2}\Tr\Big[\big\{W_v, \Phi_{H(\l)}(W_v)\big\}\rho_{\b}(\l)\Big] - \big(\beta\Tr\left[W_v\rho_{\b}(\l)\right]\big)^2.\nn
\ea
\end{proof}
The statement of Lemma \ref{lem:expression for Hessia} does not make it clear that the Hessian is a variance of a suitable operator, or even is positive. The following lemma shows how to lower bound the Hessian by a variance of a quasi-local operator. The intuition for the proof arises by writing the Hessian in a manner that makes its positivity clear. This in particular, shows that $\log Z_{\b}(\m)$ is a convex function in parameters $\m$ -- we later improve this to being strongly convex.

\begin{lem}[A lower bound on $\pmb{\nabla^2 \log Z_{\b}(\l)}$]
\label{lem:variance lower bound on Hessian}
For every $v\in \bbR^m$ and local operator $W_v=\sum_i v_i E_i$, define another local operator $\quW$ such that
\begin{align} 
\quW_v= \int_{-\infty}^\infty f_\beta(t)\ e^{-iHt}\ W\ e^{iHt} dt,\label{eq:r8}
\end{align}
where
\begin{align}
f_\beta(t)= \frac{2}{\beta\pi}\log \frac{e^{\pi |t|/\beta}+1}{e^{\pi |t|/\beta}-1}
\end{align}
is defined such that $f_\beta(t)$ scales as $\frac{4}{\beta\pi}e^{-\pi |t|/\beta}$ for large $t$.
We claim
\ba
\frac{1}{2}\Tr\Big[\big\{W_v, \Phi_{H(\l)}(W_v)\big\}\rho_{\b}(\l)\Big] - \big(\Tr\left[W_v\rho_{\b}(\l)\right]\big)^2 \geq \Tr\left[(\quW_v)^2\rho_{\b}(\l)\right]- \left(\Tr\left[\quW_v\rho_{\b}(\l)\right]\right)^2\label{eq:r6}
\ea
 
\end{lem}
\begin{rem}
For the rest of the paper, we are going to fix an arbitrary $v \in \bbR^m$, in order to avoid subscripting $W,\quW$ by $v$. 
\end{rem}
\begin{proof}[Proof of Lemma~\ref{lem:variance lower bound on Hessian}]
 Let us start by proving a simpler version of Eq.~\eqref{eq:r6}, where we only show 
\ba
\frac{1}{2}\Tr\Big[\big\{W, \Phi_{H(\l)}(W)\big\}\rho_{\b}(\l)\Big] - \big(\Tr\left[W\rho_{\b}(\l)\right]\big)^2\geq 0.
\ea
 Since $v$ is an arbitrary vector, this shows that, as expected, $\nabla^2 \log Z_{\b}(\l)$ is a positive semidefinite operator.

Consider the spectral decomposition of the Gibbs state $\r_{\b}(\l)$: $\r_{\b}(\l)=\sum_j r_j(\l) \ketbra{j}{j}$. Then observe that 
\ba
&\frac{1}{2}\Tr\Big[\big\{W, \Phi_{H(\l)}(W)\big\}\rho_{\b}(\l)\Big] - \big(\Tr\left[W\rho_{\b}(\l)\right]\big)^2\\
&= \frac{1}{2}\sum_j r_j(\l)\bra{j}\{W, \Phi^H(W)\}\ket{j} - \br{\sum_jr_j(\l)W_{j,j}}^2\nn\\
&= \frac{1}{2}\sum_{j,k} r_j(\l)\br{W_{j,k}\bra{k}\Phi^H(W)\ket{j}+\bra{j}\Phi^H(W)\ket{k}W_{k,j}} - \br{\sum_jr_j(\l)W_{j,j}}^2\nn\\
&\overset{(1)}=\frac{1}{2}\sum_{j,k} r_j(\l)\br{W_{j,k}W_{k,j}\tilde{f}_\beta(\Ene_k-\Ene_j)+W_{j,k}W_{k,j}\tilde{f}_\beta(\Ene_j-\Ene_k)} - \br{\sum_jr_j(\l)W_{j,j}}^2\nn\\
&\overset{(2)}=\sum_{j,k} r_j(\l)|W_{j,k}|^2\tilde{f}_\beta(|\Ene_j-\Ene_k|) - \br{\sum_jr_j(\l)W_{j,j}}^2.\label{eq:r9}
\ea
In equality $(1)$, we use Definition \ref{def:r1} and in equality $(2)$ we use the facts that $W$ is Hermitian and $\tilde{f}_\beta(t)=\tilde{f}_\beta(-t)$. Since $\tilde{f}_\beta(0)=1$ and $\tilde{f}_\beta(t)>0$ for all $t$, it is now evident from last equation that
$$\Tr\br{\frac{1}{2}\{W, \Phi^H(W)\}\rho_{\b}} - \Tr\br{W\rho_{\b}}^2\geq \sum_{j}r_j(\l)|W_{j,j}|^2 - \br{\sum_jr_j(\l)W_{j,j}}^2\geq 0.$$
We can improve this bound by using the operator $\quW$ in \eqref{eq:r8}. The function $f_\beta(t)$ in \eqref{eq:r8} is chosen carefully such that its Fourier transform satisfies $\tilde{f}_{\beta}(\omega)=\tilde{f}_{\beta}(|\omega|)$. Then, we have that
\begin{align} 
\quW= \int_{-\infty}^\infty f_\beta(t)\ e^{-iHt}\ W\ e^{iHt} dt = \sum_{j,k}\ket{j}\bra{k}\ W_{j,k} \ \tilde{f}_\beta(|\Ene_j-\Ene_k|).
\end{align}
Similar to \eqref{eq:r9} we get
\begin{align} 
\tr (\quW^2\rho_{\b}(\l)) - [\tr (\quW\rho_{\b}(\l))]^2& = \sum_{j,k} r_j(\l) |W_{j,k}|^2 \tilde{f}_\beta(|\Ene_j-\Ene_k|)^2 - \left( \sum_{j} r_j(\l) W_{j,j} \right)^2 \notag \\
&\le \sum_{j,k} r_j(\l) |W_{j,k}|^2 \tilde{f}_\beta(|\Ene_j-\Ene_k|) - \left( \sum_{j} r_j(\l) W_{j,j} \right)^2   \notag \\
&=\tr \left(\frac{\{ W , \Phi^H(W) \}}{2} \rho_{\b}(\l) \right) - [\tr(W \rho_{\b}(\l))]^2,
\end{align}
where the inequality is derived from $\tilde{f}_\beta(x)^2 \le \tilde{f}_\beta(x)$ for arbitrary $-\infty < x < \infty$. 
Thus, we get
$$\br{\sum_{i=1}^m v_i\frac{\partial}{\partial \lambda_i}}^2 \log Z_{\b}(\l) \geq \Tr(\quW^2\rho_{\b}(\l))- \Tr(\quW\rho_{\b}(\l))^2.$$
\end{proof}

\section{Lower bound on the variance of quasi-local operators}
In the previous section, we showed how to give a lower bound on the Hessian of the logarithm of the partition function. To be precise,  for a vector $\lambda=(\lambda_1,\ldots,\lambda_m)\in \bbR^m$ with $\|\lambda\|\leq 1$, Hamiltonian  $H(\lambda)=\sum_i \lambda_i E_i$ and $\rho_{\b}(\lambda)=\frac{1}{Z_{\b}(\lambda)}e^{-\beta H(\lambda)}$ (where $Z_{\b}(\lambda)=\tr(e^{-\beta H(\lambda)})$), we showed in Lemma~\ref{lem:variance lower bound on Hessian} how to carefully choose a local operator $\quW$ such that for every $v$, we have
\begin{align}
\label{eq:recallnablaZ>=W'}
    v^\top \cdot \big(\nabla^2 \log Z_{\b}(\l)\big) \cdot v  \geq \beta^2\Var[\quW]. 
\end{align}
In this section, we further prove that the variance of $\quW$ with respect to $\rho_{\b}(\lambda)$ can be bounded from below by a large enough quantity.  Before looking at the highly non-trivial case of finite temperature, lets look at a simpler case of infinite temperature limit.

\begin{figure}[ht]
  \centering
  
\usetikzlibrary{shapes,arrows,fit,backgrounds,calc}
\tikzstyle{box} = [rectangle, very thick, rounded corners, minimum width=3cm, minimum height=1cm,text centered, draw=black, fill=blue!20!white, inner sep=6pt, inner ysep=5pt]

\tikzstyle{arrow} = [double,double equal sign distance,-implies]

\begin{tikzpicture}[node distance=2.45cm]
\node (box0) [box] [draw, align=center]{
        Theorem \ref{claim_W_norm_W'_norm}\\
        \emph{Variance at finite temperature}};
\node (box1) [box, below of=box0] [draw, align=center]{
        \emph{Variance of locally rotated operator at finite temperature}};
\node (box2) [box, below of=box1] [draw, align=center]{
        \emph{Variance of operator with small support at finite temperature}};
\node (box3) [box, below of=box2] [draw, align=center]{
       Lemma \ref{lem:Wprimelowb} \\
         \emph{Variance of operator with small support at infinite temperature} };
\node (box4) [box, below of=box3] [draw, align=center]{
       Theorem \ref{claim_W_norm_W'_norm}\\
        \emph{Variance of operator with small support at infinite temperature}};

\draw [arrow] (box1) -- node [anchor=west] {Claims~\ref{claim1_lower} and \ref{claim2_lower}}(box0);
\draw [arrow] (box2) -- node [anchor=west] {Subsection \ref{subsec:globaltolocal}, Equation~\eqref{eq:upperboundedagAurho1}}(box1);
\draw [arrow] (box3) -- node [anchor=west] {Claim~\ref{claim:existenceofuXi0}}(box2);
\draw [arrow] (box4) -- node [anchor=west] {Claim \ref{claim_Wr_norm_W'r_norm}}(box3);

\end{tikzpicture}
   \caption{Flow of the argument in the proof of Theorem \ref{thm:variancelowerbound}.}
    \label{fig:flowchart}
\end{figure}
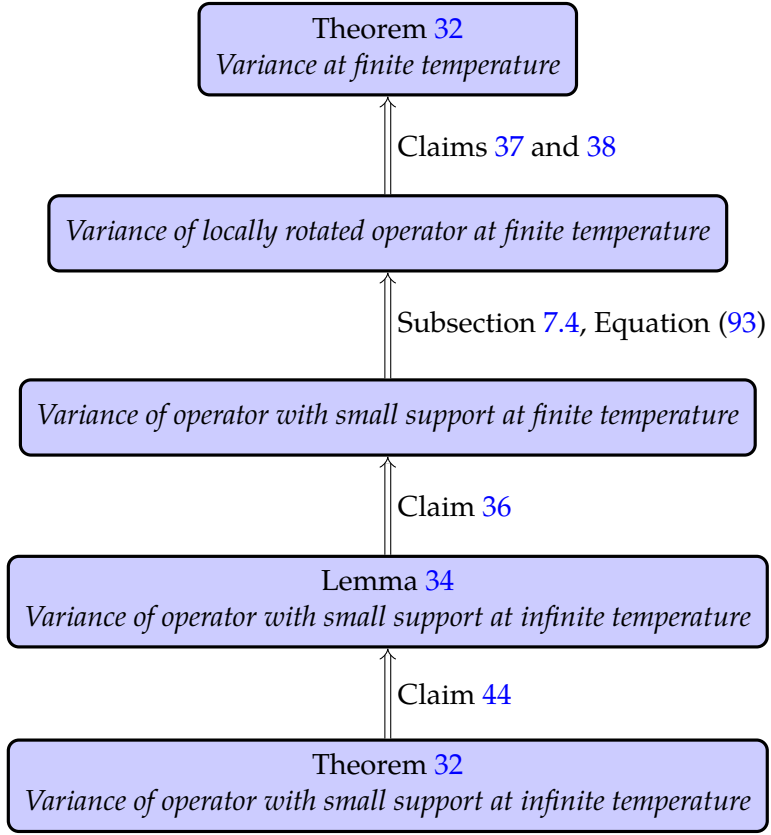
\subsection{Warm-up: Variance at infinite temperature}

Consider the infinite temperature state (i.e., the maximally mixed state) $\eta =\frac{\iden_\Lambda}{\mathcal{D}_{\Lambda}}$. We have the following theorem, where we assume that the locality of $W$, namely $\kappa=\orderof{1}$.

\begin{thm} \label{claim_W_norm_W'_norm}
For $\quW$ as defined in Lemma~\ref{lem:variance lower bound on Hessian}, we have
\begin{align} 
\Tr[(\quW)^2\eta] - \Tr[\quW\eta]^2 \ge \frac{\Omega(1)}{(\beta \log(m)+1)^2}\sum_{i=1}^m v_i^2.
\label{main_ineq_claim_W_norm_W'_norm}
\end{align} 
\end{thm}
The intuition behind the theorem is as follows. In the above statement, if $\quW$ is replaced by $W$, then the lower bound is immediate (see Eq.~\eqref{Wupp/bound_1} below). Similarly, if $H$ and $W$ were commuting, then $\quW$ would be the same as $W$ and the statement would follow. In order to show \eqref{main_ineq_claim_W_norm_W'_norm} for $\quW$ in general, we expand it in the energy basis of the Hamiltonian and use the locality of $W$ to bound the contribution of cross terms (using Lemma \ref{lem:AKL16}). This accounts for the contributions arising due to non-commutativity of $W$ and $H$.

\begin{proof} [Proof of Theorem \ref{claim_W_norm_W'_norm}]
Recall from Lemma~\ref{lem:variance lower bound on Hessian} that $W=\sum_i v_i E_i$. We first note that $\Tr[\quW\eta]= \Tr[W\eta]=0$. 
From the definition, we have $\Tr[(\quW)^2\eta]=\frac{1}{\mathcal{D}_\Lambda}\|(\quW)^2\|_F^2$. To begin with, we observe
\begin{align} 
&\|W\|_F^2=\mathcal{D}_\Lambda \sum_{i=1}^m v_i^2, \label{Wupp/bound_1}
\end{align} 
which holds since the basis $E_i$ satisfies $\|E_i\|_F^2=\mathcal{D}_\Lambda$ and $\Tr[E_iE_j]=0$ if $i\neq j$. Define $P_s^H$ as the projection onto the energy range $(s,s+1]$ of $H$.  
\begin{align} 
P_{s}^H:=\sum_{j: \Ene_j \in (s , s+1]} \ket{j}\bra{j}.
\end{align} 
Using the identity $\sum_s P_s^H=\iden_\Lambda$ and the definition of $\quW$, let us expand
\begin{align} 
\|W\|_F^2&= \sum_{s,s'=-\infty}^\infty \| P_{s'}^H W P_{s}^H \|_F^2 , \notag \\
\|\quW\|_F^2&= \sum_{s,s'=-\infty}^\infty \left \|\int_{-\infty}^\infty dt   P_{s'}^H  f_\beta(t)e^{-iHt} W e^{iHt} P_{s}^H  \right\|_F^2\\
& =\sum_{s,s'=-\infty}^\infty \left \| \sum_{\substack{j: \Ene_j \in (s , s+1]\\ k: \Ene_k \in (s , s+1]}} W_{j,k}  \tilde{f}_\beta(|\Ene_j-\Ene_k|)  P_{s'}^H\ket{j}\bra{k} P_s^H\right\|_F^2.\label{eq:frobeniusonW'}
\end{align} 
 By using the inequality 
$$
\tilde{f}_\beta(\omega)= \frac{\tanh(\beta \omega/2)}{\beta \omega/2} \geq \frac{1}{\frac{\beta}{2}|\omega|+1},$$ we have 
\begin{align} 
&\left \| \sum_{\substack{j: \Ene_j \in (s , s+1]\\ k: \Ene_k \in (s , s+1]}} W_{j,k}  \tilde{f}_\beta(|\Ene_j-\Ene_k|)  P_{s'}^H\ket{j}\bra{k} P_s^H\right\|_F^2  \notag \\
=& \sum_{j: \Ene_j \in (s , s+1]} \sum_{k: \Ene_j \in (s' , s'+1]} [\tilde{f}_\beta(|\Ene_j-\Ene_k|)]^2  |W_{j,k}|^2\notag \\
\ge& \frac{1}{\big(\frac{\beta}{2} (|s-s'|+1) +1\big)^2} \sum_{j: \Ene_j \in (s , s+1]} \sum_{k: \Ene_j \in (s' , s'+1]}  |W_{j,k}|^2
= \frac{\| P_{s'}^H W P_{s}^H \|_F^2}{\big(\frac{\beta}{2} (|s-s'|+1) +1\big)^2}.
\end{align} 
Plugging this lower bound in Eq.~\eqref{eq:frobeniusonW'} gives the following lower bound for $\|\quW\|_F^2$:
\begin{align} 
\|\quW\|_F^2 \ge  \sum_{s,s'=-\infty}^\infty \frac{\| P_{s'}^H W P_{s}^H \|_F^2}{\big(\frac{\beta}{2} (|s-s'|+1) +1\big)^2} 
= \sum_{s_0=-\infty}^\infty\sum_{s_1=-\infty}^\infty   \frac{\| P_{(s_0+s_1)/2}^H W P_{(s_0-s_1)/2}^H \|_F^2}{\big(\frac{\beta}{2} (|s_1|+1) +1\big)^2}, 
\label{W'upp/bound_1}
\end{align} 
where we have introduced $s_0=s+s', s_1=s-s'$. Let us consider the last expression for a fixed~$s_0$, introducing a cut-off parameter $\bar{s}$ which we fix eventually: 
\begin{align} 
&\sum_{s_1=-\infty}^\infty   \frac{\| P_{(s_0+s_1)/2}^H W P_{(s_0-s_1)/2}^H \|_F^2}{[\frac{\beta}{2}(|s_1|+1) +1]^2} \notag \\
\ge&  \frac{1}{\big(\frac{\beta}{2}(\bar{s}+1) +1\big)^2} 
\left( \sum_{|s_1|\le \bar{s}} \| P_{(s_0+s_1)/2}^H W P_{(s_0-s_1)/2}^H \|_F^2\right) \notag \\
=&\frac{1}{\big(\frac{\beta}{2}(\bar{s}+1) +1\big)^2} 
\left( \sum_{s_1=-\infty}^\infty \| P_{(s_0+s_1)/2}^H W P_{(s_0-s_1)/2}^H \|_F^2
-\sum_{|s_1|> \bar{s}} \| P_{(s_0+s_1)/2}^H W P_{(s_0-s_1)/2}^H \|_F^2  \right) .
\label{W'upp/bound_2}
\end{align} 
By combining the inequalities~\eqref{W'upp/bound_1} and \eqref{W'upp/bound_2}, we obtain
\begin{align} 
\|\quW\|_F^2 \ge \frac{\|W\|_F^2}{\big(\frac{\beta}{2}(\bar{s}+1) +1\big)^2}  - \frac{1}{\big(\frac{\beta}{2}(\bar{s}+1) +1\big)^2}\sum_{s_0=-\infty}^\infty\sum_{|s_1|> \bar{s}} \| P_{(s_0+s_1)/2}^H W P_{(s_0-s_1)/2}^H \|_F^2  .\label{W'upp/bound_3}
\end{align} 
Now, we will estimate the second term in \eqref{W'upp/bound_3}. Since the subspaces $P^H_{(s_0+s_1)/2}$ and $P^H_{(s_0-s_1)/2}$ are sufficiently far apart in energy, we can use the exponential concentration on the spectrum~\cite{Arad_connecting_global_local_dist} (as stated in Lemma~\ref{lem:AKL16}) to obtain the following: for $W=\sum_i v_i E_i$, we have 
\begin{align}
\begin{aligned}
\| P_{(s_0+s_1)/2}^H  W P_{(s_0-s_1)/2}^H \|&\le \sum_{i=1}^m v_i \| P_{(s_0+s_1)/2}^H  E_i P_{(s_0-s_1)/2}^H \|\\
&\le
 C e^{-\lambda (|s_1|-1-\kappa)} \sum_{i=1}^m |v_i| \le  C m e^{-\lambda (|s_1|-1-\kappa)} \max_{i} |v_i| .
 \end{aligned}
 \label{eq:usingAKL1}
\end{align}
 where we use the condition that $E_i$ are tensor product of Pauli operators with weight at most $\k$, and the parameters $C$ and $\lambda$ are $\orderof{1}$ constants (see Lemma~\ref{lem:AKL16} for their explicit forms).
 Then, the second term in \eqref{W'upp/bound_3} can be upper-bounded by
\begin{align} 
&\sum_{s_0=-\infty}^\infty\sum_{|s_1|> \bar{s}} \| P_{(s_0+s_1)/2}^H W P_{(s_0-s_1)/2}^H \|_F^2 \notag\\
&\overset{(1)}\le \sum_{s_0=-\infty}^\infty\sum_{|s_1|> \bar{s}} \| P_{(s_0+s_1)/2}^H W P_{(s_0-s_1)/2}^H \|^2\cdot \| P_{(s_0+s_1)/2}^H \|_F^2\notag\\ 
&\overset{(2)}\le \sum_{|s_1|> \bar{s}} \sum_{s_0=-\infty}^\infty \| P_{(s_0+s_1)/2}^H \|_F^2 \cdot  C^2 m^2 e^{-2\lambda (|s_1|-1-\kappa)} \max_{i} v_i^2  \notag \\
&= \mathcal{D}_\Lambda  C^2 m^2 e^{2\lambda (1+\kappa)}\max_{i} v_i^2 \sum_{|s_1|\ge \bar{s}+1} e^{-2\lambda |s_1|} \notag \\
&\overset{(3)}\leq \mathcal{D}_\Lambda \max_{i} v_i^2\frac{C^2 m^2 e^{2\lambda(\kappa+1)} }{\lambda} e^{-2\lambda \bar{s}} \overset{(4)}\leq \mathcal{D}_\Lambda \frac{C^2 m^2 e^{2\lambda(\kappa+1)} }{\lambda} e^{-2\lambda \bar{s}} \br{\sum_i v_i^2},\label{W'upp/bound_4}
\end{align}
where inequality $(1)$ follows from Eq.~\eqref{fact:relatingfrobeniusAB},  $(2)$ follows from Eq.~\eqref{eq:usingAKL1},  $(3)$ follows from Fact \ref{fact:integrals} and 
$(4)$ follows from $\max_{i} v_i^2\le \sum_i v_i^2$.  
Therefore, by applying Eq.~\eqref{Wupp/bound_1} and \eqref{W'upp/bound_4} to \eqref{W'upp/bound_3}, we arrive at the lower bound as 
\begin{align} 
\|\quW\|_F^2 \ge \frac{\mathcal{D}_\Lambda}{[\beta (\bar{s}+1)/2 +1]^2} \left(\sum_{i=1}^m v_i^2 \right) \left(1 - \frac{C^2m^2 e^{2\lambda(\kappa+1)} }{\lambda} e^{-2\lambda \bar{s}} \right).
\end{align} 
Since $\lambda, C, \kappa=\orderof{1}$, by choosing $\bar{s}=\orderof{\log(m)}$, we obtain the main inequality~\eqref{main_ineq_claim_W_norm_W'_norm}. 
This completes the proof. $\square$
\end{proof}

\subsection{Variance at finite temperature}

Next, we show how to prove a variance lower bound at finite temperature. This is achieved by the following general theorem on the variance of arbitrary local operator, which will reduce the problem to estimating a ``variance-like'' quantity at the infinite temperature case (observe the occurrence of the maximally mixed state $\eta$ in the theorem below). 
\begin{thm}
\label{thm:variancelowerbound}
Let $\beta>0$, $H$ be a $\k$-local Hamiltonian  on the lattice $\Lambda$ and $\rho_\beta=\frac{e^{-\beta H}}{\tr(e^{-\beta H})}$. 
Let $A$ be a $(\tau, a_1, a_2, 1)$-quasi-local operator (see Eq.~\eqref{quasi_locality of A}) where $a_2=\orderof{1/\beta}, a_1=\orderof{1}$ are constants and $Z$ are restricted to be connected sets within $\Lambda$. Suppose $\Tr[A\rho_\beta]=0$ and $\tau\leq 1$. We have
$$
\langle A^2 \rangle =\tr(A^2\rho_\beta)\geq \br{\max_{i\in \Lambda}\Tr[A_\locci^2\eta]}^{\beta^{\Omega(1)}}.
$$
\end{thm}
We remark that the theorem statement above  hides several terms that depend on the lattice, such as the lattice dimension, the degree of the graph and the locality of Hamiltonian (which we have fixed to be a constant). Additionally, the assumptions $a_2=\orderof{1/\beta}, a_1=\orderof{1}$ are made  in order to show that an operator $A^*$ to which we apply the theorem, satisfies the assumptions of the theorem.\footnote{We remark that we can also apply the theorem for  other choices of $a_1, a_2$, with small modifications to the proof.} Before proving the theorem, we first discuss how to use this theorem in order to prove a lower bound on the Hessian of $\log Z(\l)$.

For an arbitrary $v \in \bbR^m$, let $W=\sum_i v_i E_i$ and $\quW$ be the operators defined in Lemma~\ref{lem:variance lower bound on Hessian}. In Appendix~\ref{append:Wquasi} we show that $\quW$ is a 
$\br{1/D, \orderof{1}, \orderof{1/\beta}, c_\ast\beta^{2D+1}  \br{\max_{i\in \Lambda}|v_i|}}$-quasi-local operator, for $c_\ast=\orderof{1}$. Thus, the following operator 
$$
A^*=\frac{\beta^{-2D-1}}{c_\ast \max_{i\in \Lambda} |v_i| }(\quW-\Tr[\rho_\beta \quW]\iden),
$$
is $(\orderof{1}, \orderof{1}, \orderof{1/\beta}, 1)$-quasi-local and satisfies $\Tr[A\rho_{\beta}]=0$.
We now apply Theorem \ref{thm:variancelowerbound} to the operator $A^*$ to prove Theorem \ref{thm:strong convexity of log-partition function}. We need to estimate $\max_i\Tr[A_\locci^{\ast 2}\eta]$. Consider 
$$\max_i\Tr[(A^*_\locci)^2\eta]= \frac{\beta^{-4D-2}}{c_\ast^2\br{\max_{i\in \Lambda} |v_i|}^2}\br{\max_i\Tr[(\quW_\locci)^2\eta]}.$$
The following lemma is shown in Appendix \ref{append:lowerboundwprimei}.
\begin{lem}
\label{lem:Wprimelowb}
It holds that 
$$
\max_{i\in \Lambda}(\Tr[(\quW_\locci)^2\eta]) =   \frac{\Omega(1)}{\br{\beta\log(\beta)+1}^{2D+2}} \br{\max_{i\in \Lambda}v_i^2}.
$$
\end{lem}
This implies
$$
\max_i\Tr[(A^*_\locci)^2\eta]= \frac{\Omega(1)}{\beta^{4D+2}\br{\beta\log(\beta)+1}^{2D+2}\br{\max_{i\in \Lambda} |v_i|}^2}\br{\max_{i\in \Lambda}v_i^2}=\frac{1}{\beta^{\Omega(1)}}.
$$
Using this lower bound in Theorem \ref{thm:variancelowerbound}, we find
\begin{align*}
\langle (\quW)^2 \rangle- \br{\langle \quW \rangle}^2 &= \frac{\beta^{4D+2} \br{\max_{i\in \Lambda} |v_i|}^2}{c_\ast^2}\br{\langle (A^*)^2 \rangle- \br{\langle A^* \rangle}^2}\\ 
&=\beta^{\Omega(1)}\br{\max_{i\in \Lambda} |v_i|}^2\cdot \br{\max_i\Tr[(A^*_\locci)^2\eta]}^{\beta^{\Omega(1)}} \\
&= \beta^{\Omega(1)}\cdot \br{\frac{1}{\beta^{\mathcal{O}(1)}}}^{\beta^{\mathcal{O}(1)}} \cdot \br{\max_{i\in \Lambda} |v_i|}^2\\
&\overset{(1)}=\beta^{\Omega(1)}\cdot e^{-\beta^{\mathcal{O}(1)}}\br{\max_{i\in \Lambda} |v_i|}^2 \geq \beta^{\Omega(1)}\cdot \frac{e^{-\beta^{\mathcal{O}(1)}}}{m}\br{\sum_i v_i^2},
\end{align*}
where we used $\beta^{-\mathcal{O}(1)}\geq e^{-\beta}$ in $(1)$. 
 Putting together the bound above with Eq.~\eqref{eq:recallnablaZ>=W'}, we find that for every $v\in \bbR^m$,  
$$ 
    v^\top \cdot \big( \nabla^2 \log Z_{\b}(\l)\big) \cdot v \geq \beta^2\Var[\quW]\geq 
     \beta^{\Omega(1)}\cdot \frac{e^{-\beta^{\mathcal{O}(1)}}}{m} \sum_{i=1}^m v_i^2.
$$
This establishes Theorem \ref{thm:strong convexity of log-partition function}.
\subsection{Some key quantities in the proof and proof sketch}

\begin{figure}
\centering
{
\includegraphics[clip, scale=0.5]{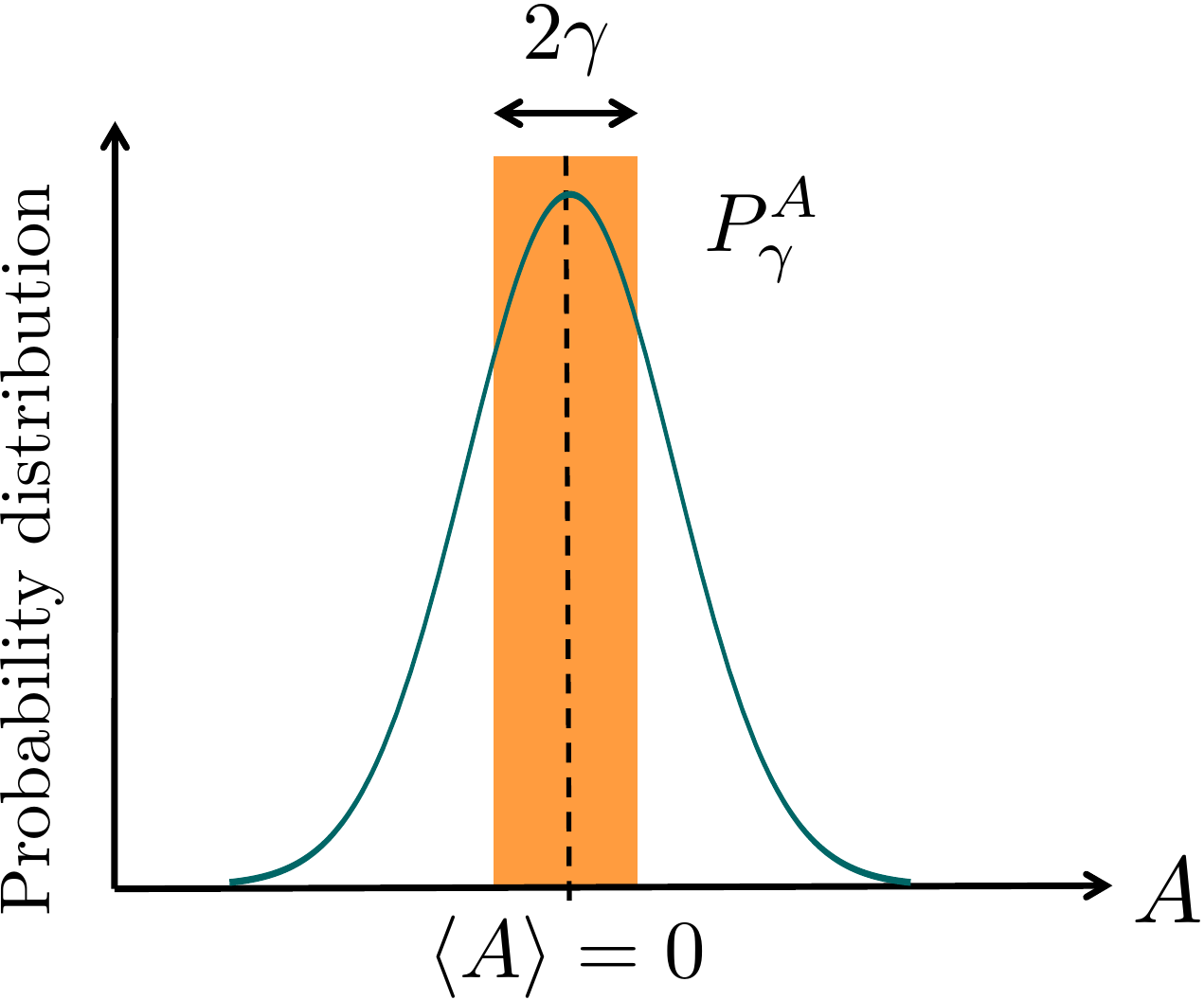}
}
\caption{Plot of the probability distribution $\Tr[\Pi_\omega\rho_\beta]$, where $\Pi_\omega$ is the projector onto the eigenvectors of $A$ with eigenvalue $\omega$. It is assumed that $\Tr[A\rho_\beta]=0$. For a $\gamma$ to be chosen in the proof, $P^A_\gamma$ is the projector onto the subspace of eigenvectors of $A$ with eigenvalue between $[-\gamma,\gamma]$. A lower bound on the variance of $A$ follows if we can show that for a constant $\gamma$, the probability mass in the colored range is small (see Equation \ref{staring_ineq_variance}).  
}
\label{fig:system}
\end{figure}

We now prove Theorem~\ref{thm:variancelowerbound}. For notational simplicity, let 
\begin{align}
H':= H- \frac{1}{\beta}\log Z_{\b},
\end{align}
which allows us to write $\rho_\beta=e^{-\beta H'}$. We will interchangeably use the frobenius norm to write $\langle A^2 \rangle =\tr(A^2\rho_\beta)= \|A\sqrt{\rho_\beta}\|_F$. 
We now define the projection operator ${P}_\gamma^A$ as follows (see Figure \ref{fig:system}):
\begin{align}
\label{eq:defnofPgamma}
{P}_\gamma^A := \sum_{\omega\in [-\gamma , \gamma]} \Pi_\omega,
\end{align}
where $\Pi_\omega$ is the projector onto the eigenspace of $A$ with eigenvalue $\omega$. 
We then define $\delta_\gamma$ by
\begin{align}
 \delta_\gamma:=1-\| {P}_\gamma^A  \sqrt{ \rho_\beta}\|_F^2. \label{def:delta_gamma}
\end{align}
Using $\delta_\gamma$, observe that we can lower bound  $\langle A^2 \rangle$ as 
\begin{align}
\begin{aligned}
\langle A^2 \rangle&= \sum_\omega \omega^2\langle \omega|\rho_\beta|\omega\rangle \geq \sum_{|\omega|\geq \gamma} \omega^2\langle \omega|\rho_\beta|\omega\rangle \geq \gamma^2 \sum_{|\omega|\geq \gamma}\langle \omega|\rho_\beta|\omega\rangle \ge \gamma^2 \delta_\gamma. 
\end{aligned}
\label{staring_ineq_variance}
\end{align}
 Let ${Q}_\gamma^A =\iden- {P}_\gamma^A$, then observe that from Eq.~\eqref{def:delta_gamma}  that
\begin{align}
\| {P}_\gamma^A  \sqrt{\rho_\beta} - \sqrt{\rho_\beta} \|_F^2 = \| {Q}_\gamma^A \sqrt{\rho_\beta}  \|_F^2= \delta_\gamma. \label{def:delta_gamma_ineq}  
\end{align}

The proof is divided into the following subsections, each of which may be of independent interest on its own.

\subsection{Reducing the global problem to a local problem}
\label{subsec:globaltolocal}

A main challenge in bounding the variance of the operator $A$ is that it is a global operator and its properties may scale badly with the system size. But since it is a linear combination of local operators, it is related to operators supported in a local region by a simple linear transform. To this end, recall the definition of $A_\locci$ (the local operator which includes essentially the terms in $A$ that have support on the $i$th site) from the subsection \ref{global_to_local}. 
Using Haar random unitaries, we obtain the integral representation of $A_\locci$ as 
 \begin{align}
A_\locci = A - \frac{1}{d} [\tr_i (A)] \otimes \iden_i= A- \int d\mu (U_{i})U_{i}^\dagger A U_{i},
\label{def_A_i_d}
\end{align}
where $\mu(U_{i})$ is the Haar measure for unitary operator~$U_{i}$ which acts on the $i$th site. Since the $A_\locci$ is obtained from a quasi-local $A$, it is quasi-local itself. Next claim will approximate $A_\locci$ by a local operator. 
\begin{claim}\label{claim:quasi-local-A_i}
For an integer $R$, let $X_{i}$ be the radius-$R$ ball \emph{around} the site $i$, i.e., $X_{i}= B(R,i)$.
There exists an operator $A_{X_i}$ supported entirely on $X_i$, such that
$$\|A_\locci-A_{X_i}\|\leq 2a_1\cdot\br{\frac{4}{a_2\tau^2}}^{\frac{1}{\tau}}\cdot e^{-\frac{a_2}{2} (R)^{\tau}}.$$
\end{claim}

\begin{proof}
Using the representation of $A$ in Theorem \ref{thm:variancelowerbound} and the fact that local operators not containing~$i$ in their support are removed, we can write
$$
A_\locci=\sum_{\substack{k, Z\subseteq \Lambda:\\ |Z|\leq k, Z\ni i}}  {g}_k\br{a_Z - \frac{1}{d} [\tr_i (a_Z)] \otimes \iden_i},
$$
Define
$$
A_{X_i}=\sum_{\substack{k, Z\subseteq X_i:\\ |Z|\leq k, Z\ni i}}  {g}_k \br{a_Z - \frac{1}{d} [\tr_i (a_Z)] \otimes \iden_i}.
$$
be the desired approximations of $A_\locci$ by removing all operators that are not contained in $X_i$. Observe that
\begin{align}
    \|A_\locci-A_{X_i}\|&\leq 2\sum_{\substack{k, Z\subseteq \Lambda:\\ Z\not\subset X_i, |Z|\leq k, Z\ni i}}  {g}_k \|a_Z\| 
    \overset{(1)}\leq 2\sum_{\substack{k,Z\subseteq \Lambda:\\ \diam(Z)\geq R, |Z|\leq k, Z\ni i}}  {g}_k\|a_Z\| \notag\\
    &\overset{(2)}\leq 2\sum_{k\geq R}  {g}_k\br{\sum_{Z: Z\ni i}\|a_Z\|}\overset{(3)}\leq 2\zeta\sum_{k\geq R}  {g}_k\notag\\
    &\leq 2\zeta a_1\sum_{k\ge R}e^{-a_2 k^{\tau}}\overset{(4)}\leq 2\zeta a_1\cdot a_2^{-\frac{1}{\tau}}\cdot \br{\frac{2}{\tau}}^{\frac{2}{\tau}}\cdot e^{-\frac{a_2}{2} (R)^{\tau}}.
    \label{approx_Xi}
\end{align}
For inequality $(1)$, note that since $Z\not\subset X_i$ and $Z\ni i$, the diameter of $Z$ (recall that $Z$ is a ball) must be larger than the radius of $X_i$, which is $R$. Inequality $(2)$ holds since $k\geq |Z| \geq \diam(Z)$, inequality $(3)$ uses Definition \ref{def:Quasi-local operators} and inequality $(4)$ uses Fact \ref{fact:integrals}. Since $\zeta=1$ for the given $A$ (see the statement of Theorem \ref{thm:variancelowerbound}),
this completes the proof.
\end{proof}
We now define $i_0$ as 
\begin{equation}
\label{eq:definitionofi0}
i_0:=\argmax_{i}  \|A_{i}\sqrt{\eta }\|_F.    
\end{equation}
The set of unitaries $U_{i_0}$ and the Haar measure $\mu(U_{i_0})$ are defined analogously. Plugging in $$
R=\br{\frac{2}{a_2}
\log\frac{8\cdot 4^{\frac{1}{\tau}}\cdot a_1}{\tau^{\frac{2}{\tau}}\cdot a_2^{\frac{1}{\tau}}\|A_{(i_0)}\sqrt{\eta }\|_F}}^{\frac{1}{\tau}}$$ in Claim \ref{claim:quasi-local-A_i} we get
 \begin{align}
\| A_\loccio - A_{X_{i_0}}\| \le \frac{1}{4}\|A_\loccio \sqrt{\eta} \|_F\label{choice_of_X_i_0}
\end{align} 
Substituting $a_2=\orderof{1/\beta}, a_1=\orderof{1}, \tau=\mathcal{O}(1)$, we find that we can ensure the condition~\eqref{choice_of_X_i_0} for
\begin{equation}
\label{eq:mainchoiceofR}
R=\diam(X_{i_0}) = \br{\beta\log\br{\frac{1}{\|A_\loccio \sqrt{\eta}\|_F}}}^{\Omega(1)}.
\end{equation}

\subsection{Variance of operators with small support: finite temperature to infinite temperature}
Having related $A$ to the operator $A_{(i_0)}$ (which is essentially supported on a small number of sites in the lattice, up to a tail decaying sub-exponentially in the radius), we now argue that it is simpler to bound the variance of $A_{(i_0)}$ in terms of its variance at infinite temperature, as long as some local rotations are allowed. In particular, we will show the existence of a unitary $U_{i_0}$ for which we can proceed in this fashion. The intuition here is that if rotations are allowed, then the eigenvectors of $A_{X_{i_0}}$ can be rearranged to yield largest possible variance with $\rho_\beta$. This turns out to be larger than the variance with $\eta$.  To make this precise, we prove the following claim.
\begin{claim}
\label{claim:existenceofuXi0}
There exists $U_{X_{i_0}}$ such that
 \begin{align}
\| U_{X_{i_0}}^\dagger A_\loccio U_{X_{i_0}} \sqrt{\rho_\beta }\|_F &\ge   \|A_\loccio\sqrt{\eta}\|_F - 2\| A_{X_{i_0}}-A_\loccio \|\ge \frac{1}{2}\|A_\loccio\sqrt{\eta}\|_F\notag,
\end{align}
where the second inequality uses Eq.~\eqref{choice_of_X_i_0}.
\end{claim}

\begin{proof}[Proof of Claim~\ref{claim:existenceofuXi0}]
Recall that the goal is to show the existence of a unitary $U_{X_{i_0}}$ satisfying
 \begin{align}
\| U_{X_{i_0}}^\dagger A_\loccio U_{X_{i_0}} \sqrt{\rho_\beta }\|_F \ge   \|A_\loccio\sqrt{\eta}\|_F - 2\| A_{X_{i_0}}-A_\loccio \|.
\end{align}
We start from the following,
 \begin{align}
\| U_{X_{i_0}}^\dagger A_\loccio U_{X_{i_0}} \sqrt{\rho_\beta} \|_F 
&= \| U_{X_{i_0}}^\dagger \bigl[(A_\loccio - A_{X_{i_0}}) +  A_{X_{i_0}} \bigr]U_{X_{i_0}} \sqrt{\rho_\beta}  \|_F \notag \\
&\ge  \| U_{X_{i_0}}^\dagger A_{X_{i_0}} U_{X_{i_0}} \sqrt{\rho_\beta}  \|_F  - \| A_\loccio - A_{X_{i_0}}\|
\label{inequality_A_X_i_0_norm0}
\end{align}
and lower-bound the norm of $\| U_{X_{i_0}}^\dagger A_{X_{i_0}} U_{X_{i_0}} \sqrt{\rho_\beta} \|_F$. For this, define
\begin{align}
\rho_{\beta, X} := \tr_{X^\co}(  \rho_{\beta}) ,
\end{align}
where $ \tr_{X^\co}$ is the partial trace operation for the Hilbert space on $X^\co$. We define the spectral decomposition of $A_{X_{i_0}}$ as 
\begin{align}
A_{X_{i_0}} = \sum_{s=1}^{\mathcal{D}_{X_{i_0}}} \varepsilon_s \ket{\varepsilon_s}\bra{\varepsilon_s},
\end{align}
where $\varepsilon_s$ is ordered as $|\varepsilon_1| \ge |\varepsilon_2| \ge |\varepsilon_3| \ge \cdots$ and~$\mathcal{D}_{X_{i_0}}$ is the dimension of the Hilbert space on~$X_{i_0}$. Additionally,  define the spectral decomposition of $\rho_{\beta, X_{i_0}}$ as 
\begin{align}
\rho_{\beta, X_{i_0}} = \sum_{s=1}^{\mathcal{D}_{X_{i_0}}} p_{s} \ket{\mu_s}\bra{\mu_s},
\end{align}
where $p_s$ is ordered as $p_1 \ge p_2 \ge p_3 \ge \cdots$ and $\ket{\mu_s}$ is the $s$th eigenstate of $\rho_{\beta, X_{i_0}}$.
We now choose the unitary operator $U_{X_{i_0}}$ such that 
\begin{align}
U_{X_{i_0}}\ket{\mu_s} = \ket{\varepsilon_s}  \quad {\rm for} \quad s=1,2,\ldots,\mathcal{D}_{X_{i_0}}
\end{align}
We then obtain
\begin{align}
U_{X_{i_0}} \rho_{\beta, X_{i_0}} U_{X_{i_0}}^\dagger = \sum_{s=1}^{\mathcal{D}_{X_{i_0}}} p_{s} \ket{\varepsilon_{s}}\bra{\varepsilon_{s}}.
\end{align}
This implies
\begin{align}
\| U_{X_{i_0}}^\dagger A_{X_{i_0}} U_{X_{i_0}}\sqrt{\rho_\beta}  \|_F^2&= \Tr[U_{X_{i_0}}^\dagger A^2_{X_{i_0}} U_{X_{i_0}}\rho_\beta]\notag\\
&= \Tr[A^2_{X_{i_0}} U_{X_{i_0}}\rho_\beta U_{X_{i_0}}^\dagger] = \Tr_{X_{i_0}}[A^2_{X_{i_0}} U_{X_{i_0}}\rho_{\beta, X_{i_0}} U_{X_{i_0}}^\dagger]\notag\\
&= \sum_{s=1}^{\mathcal{D}_{X_{i_0}}} p_{s}  \varepsilon_s^2 \ge \frac{1}{\mathcal{D}_{X_{i_0}}} \sum_{s=1}^{\mathcal{D}_{X_{i_0}}}  \varepsilon_s^2 =  \|A_{X_{i_0}}\sqrt{\eta} \|_F^2,
\label{inequality_A_X_i_0_norm}
\end{align}
where the inequality used the fact that  $p_s, \varepsilon_s$ are given in descending order. Then, the minimization problem of  $\sum_s p_s \varepsilon_s$ for ${p_s}_s$ with the constraint $p_1 \ge p_2 \ge p_3 \ge \cdots$ has a solution of $p_1=p_2=\cdots =p_{D_{X_{i_0}}}$. Using the lower bound
 \begin{align}
 \|A_{X_{i_0}}\sqrt{\eta} \|_F=\| (A_{X_{i_0}}-A_{(i_0)} + A_\loccio) \sqrt{\eta}  \|_F  \ge \|A_\loccio\sqrt{\eta }\|_F - \| A_{X_{i_0}}-A_\loccio \|,
\end{align}
we can reduce inequality~\eqref{inequality_A_X_i_0_norm} to
 \begin{align}
\| U_{X_{i_0}}^\dagger A_{X_{i_0}} U_{X_{i_0}} \sqrt{\rho_\beta}\|_F \ge \|A_\loccio\sqrt{\eta }\|_F - \| A_{X_{i_0}}-A_\loccio \|.
\label{inequality_A_X_i_0_norm2}
\end{align}
By combining the inequalities~\eqref{inequality_A_X_i_0_norm0} and \eqref{inequality_A_X_i_0_norm2}, we obtain
\begin{align*}
\| U_{X_{i_0}}^\dagger A_\loccio U_{X_{i_0}} \sqrt{\rho_\beta} \|_F &\ge  \| U_{X_{i_0}}^\dagger A_{X_{i_0}} U_{X_{i_0}} \sqrt{\rho_\beta}  \|_F  - \| A_\loccio - A_{X_{i_0}}\|\\
&\geq   \| U_{X_{i_0}}^\dagger A_{X_{i_0}} U_{X_{i_0}} \sqrt{\rho_\beta}  \|_F  - 2\| A_\loccio - A_{X_{i_0}}\|\
\end{align*}
which proves the claimed statement.
\end{proof}

\subsection{Invariance under local unitaries}

Recall that we reduced the problem of variance of $A$ to that of the operator $A_{(i_0)}$ that is essentially supported on small number of sites. But in the process, we introduced several local unitaries (c.f.~previous subsections). In order to handle the action of these unitaries, we will use two claims which show that local unitaries do not make much difference in the relative behavior of spectra of $A$ and $H'$. To elaborate, consider any local operator $U_X$ acting on constant number of sites $X$ on the state $\rho_\beta$. It is expected that the quantum state $U^\dagger_X\rho_\beta U_X $ has ``similar" spectral properties as $\rho_\beta $. So if the  eigen-spectrum of the operator $A$ is strongly concentrated for $\rho_\beta$, one would expect this behavior to hold even for $U^\dagger_X\rho_\beta U_X$. We make this intuition rigorous in the following claim.

\begin{claim} \label{claim1_lower}
Let $c_1,c_2, \lambda$ be universal constants. Let $X\subseteq \Lambda$. For every unitary $U_X$ supported on~$X$, we have 
\begin{align}
&\| {Q}_\gamma^A   U_X \sqrt{\rho_\beta}  \|_F^2\le \exp\big({\lambda |X|}\big)\delta_\gamma^{\frac{c_2}{c_2+\beta}}. 
\label{arbirtary_unitary_op}
\end{align}
\end{claim}
Let us see a simple application of the claim. It allows us to control the variance of $A$ even after local operations are applied to it. More precisely, 
\begin{align}
    \label{eq:proofsketchstep2}
\begin{aligned}
\|AU_X\sqrt{\rho_\beta}\|^2_F&= \|AP_\gamma^A U_X\sqrt{\rho_\beta}\|^2_F+ \|A(\iden-P_\gamma^A )U_X\sqrt{\rho_\beta}\|^2_F\\
&\leq \gamma^2+\|A\|^2  \cdot \|(\iden-P_\gamma^A ) U_X \sqrt{\rho_\beta}  \|_F^2.
\end{aligned}
\end{align}
 By Claim \ref{claim1_lower}, the expression on the second line is upper bounded by $\gamma^2+\|A\|^2e^{\orderof{1} |X|}\delta_\gamma^{\orderof{1}/\beta}$. This upper bound on $\|AU_X\sqrt{\rho_\beta}\|_F$ suffices to provide an inverse-polynomial \emph{lower bound} on the variance of $A^2$, since we can lower bound $\delta_\gamma$ for an appropriate choice of $\gamma$. However we now show how one can polynomially improve upon this upper bound (thereby the lower bound on variance) using the following claim. This claim, along the lines of Claim \ref{claim1_lower}, also shows that local unitaries $U_X$ do not change the desired expectation values. 

\begin{claim} 
\label{claim2_lower}
Let $X\subseteq \Lambda$. For every unitary $U_X$ supported on~$X$, we have\footnote{Explicit $\orderof{1}$ constants that appear in this inequality are made clear in the proof.} 
\begin{align}
\left \| A    {Q}_\gamma^A U_X \sqrt{\rho_\beta}  \right\|_F^2 \le    \frac{1}{\gamma}\cdot \exp\big({\orderof{1}\cdot|X|}\big) \delta_\gamma^{\orderof{1}/\beta}+\orderof{1}\cdot |X|^6 \cdot \langle A^2\rangle.
\label{unitary_A_norm_small}
\end{align}
\end{claim}
Proof of both the Claims~\ref{claim1_lower},~\ref{claim2_lower} appear in Section~\ref{subsec:2claimproofs}. An immediate corollary of this claim is the following, that improves upon Eq.~\eqref{eq:proofsketchstep2}.
\begin{cor}
\label{cor:unitnochange}
Let $X$ be a subset of $\Lambda$ of size $|X|=\mathcal{O}(1)$. For every unitary $U_X$ supported on~$X$, we have
 \begin{align}
\| A U_X \sqrt{\rho_\beta} \|_F^2 \le \gamma^2 + \frac{e^{\orderof{1}\cdot|X|} \delta_\gamma^{\orderof{1}/\beta}}{\gamma}+\orderof{1}|X|^6 \langle A^2\rangle.
\label{unitary_A_norm_small3}
\end{align}
\end{cor}
\begin{proof}
Similar to Eq.~\eqref{eq:proofsketchstep2}, we upper bound $\| A U_X \sqrt{\rho_\beta} \|_F^2$ as
 \begin{align}
\| A U_X \sqrt{\rho_\beta} \|_F^2  &= \left \| A  {P}_\gamma^A U_X \sqrt{\rho_\beta}  \right\|_F^2  +\left \| A   {Q}_\gamma^A U_X \sqrt{\rho_\beta}  \right\|_F^2
\le \gamma^2 +\left \| A   {Q}_\gamma^A U_X \sqrt{\rho_\beta}  \right\|_F^2,
\label{unitary_A_norm_small0}
\end{align}
since $Q^A_\gamma=\iden-P^A_\gamma$. 
By combining this with Claim \ref{claim2_lower}, the corollary follows.
\end{proof}

\subsection{Proof of the Theorem \ref{thm:variancelowerbound}}\label{sec:Proof of the Theorem variancelowerbound}
We are now ready to prove the main theorem statement. The main idea of the proof is the following: if the spectrum of $A$ is strongly concentrated for the Gibbs state $\rho_\beta$, the concentration can be proven to be protected to arbitrary local unitary operations (see Claims~\ref{claim1_lower} and~\ref{claim2_lower}). On the other hand, by choosing local unitary operations appropriately, we can relate the variance of the operator $A$ (rotated by certain local unitary $U_{i_0}$ on the site $i_0$) to  the variance of the operator $A_{(i_0)}$ and hence give a good lower bound to the variance (see Claim~\ref{claim:existenceofuXi0}). Combining the two results allows us to lower bound the variance of $A$ and hence prohibits the strong spectral concentration of the operator $A$. We formally prove this now.

\begin{proof}
Let $U_{X_{i_0}}$ be the unitary as chosen in Claim~\ref{claim:existenceofuXi0}. Using  Eq.~\eqref{def_A_i_d}, we obtain the following expression for~$U_{X_{i_0}}^\dagger A_{(i_0)} U_{X_{i_0}}$
 \begin{align}
U_{X_{i_0}}^\dagger A_{(i_0)} U_{X_{i_0}}  = 
U_{X_{i_0}}^\dagger A U_{X_{i_0}}  - \int d\mu (U_{i_0}) U_{X_{i_0}}^\dagger U_{i_0}^\dagger AU_{i_0} U_{X_{i_0}} .
\end{align}
Using triangle inequality, we have
 \begin{align}
  \| U_{X_{i_0}}^\dagger A_{(i_0)} U_{X_{i_0}} \sqrt{\rho_\beta} \|_F &\le   
\| U_{X_{i_0}}^\dagger A U_{X_{i_0}} \sqrt{\rho_\beta} \|_F + \int d\mu (U_{i_0}) \|  U_{X_{i_0}}^\dagger U_{i_0}^\dagger AU_{i_0} U_{X_{i_0}}\sqrt{\rho_\beta}  \|_F  \notag \\
&=\|A U_{X_{i_0}}\sqrt{\rho_\beta}  \|_F + \int d\mu (U_{i_0}) \|  AU_{i_0} U_{X_{i_0}}\sqrt{\rho_\beta}  \|_F 
\label{eq:upperboundedagAurho1}
\end{align}
This implies
\begin{align*}
\| U_{X_{i_0}}^\dagger A_{(i_0)} U_{X_{i_0}} \sqrt{\rho_\beta} \|_F^2 &\le \Big(\|A U_{X_{i_0}}\sqrt{\rho_\beta}  \|_F + \int d\mu (U_{i_0}) \|  AU_{i_0} U_{X_{i_0}}\sqrt{\rho_\beta}  \|_F\Big)^2\\
&\leq 2\|A U_{X_{i_0}}\sqrt{\rho_\beta}  \|^2_F + 2\Big(\int d\mu (U_{i_0}) \|  AU_{i_0} U_{X_{i_0}}\sqrt{\rho_\beta}  \|_F\Big)^2.
\end{align*}
Now, we can use Corollary \ref{cor:unitnochange} to obtain the upper bound 
\begin{align}
 \begin{aligned}
\| U_{X_{i_0}}^\dagger A_{(i_0)} U_{X_{i_0}} \sqrt{\rho_\beta} \|_F^2 &\le 4\gamma^2 + \frac{e^{\orderof{1} |X_{i_0}|}\delta_\gamma^{\orderof{1}/\beta}}{\gamma}+ \orderof{1}|X_{i_0}|^6\langle A^2\rangle.
\label{norm_quasi_local0}
\end{aligned}
\end{align}
Using Claim \ref{claim:existenceofuXi0}, we have 
\begin{align}
\| U_{X_{i_0}}^\dagger A_{(i_0)} U_{X_{i_0}} \sqrt{\rho_\beta} \|_F  \ge  \frac{1}{2} \|A_{(i_0)}\sqrt{\eta }\|_F. \label{norm_quasi_local}
\end{align}
 Putting together the upper bound in Eq.~\eqref{norm_quasi_local0} and the lower bound in Eq.~\eqref{norm_quasi_local}, we have
 \begin{align}
4\gamma^2 + \frac{e^{\orderof{1} |X_{i_0}|}\delta_\gamma^{\orderof{1}/\beta}}{\gamma}+ \orderof{1}|X_{i_0}|^6 \langle A^2\rangle
\ge \frac{1}{4} \|A_{(i_0)}\sqrt{\eta }\|_F^2.
\end{align}
By choosing as $\gamma^2=\|A_{(i_0)}\sqrt{\eta }\|_F^2 /32=:\gamma_0^2$, we obtain 
 \begin{align}
\frac{e^{\orderof{1} |X_{i_0}|}\delta_{\gamma_0}^{\orderof{1}/\beta}}{\gamma_0}+ \orderof{1}|X_{i_0}|^6 \langle A^2\rangle
\ge \gamma_0^2.
\end{align}
This inequality implies that either 
$$
\delta_{\gamma_0}\geq \br{\gamma_0^3e^{-\orderof{1} |X_{i_0}|}}^{\beta\cdot \mathcal{O}(1)}
$$ or 
$$
\langle A^2\rangle \geq \frac{\Omega(1)\gamma_0^2}{|X_{i_0}|^6}.
$$ 
Combining with Eq.~\eqref{staring_ineq_variance}, we conclude that 
$$
\langle A^2\rangle \geq \min\Big\{\gamma_0^2\cdot\br{\gamma_0^3e^{-\orderof{1} |X_{i_0}|}}^{\beta\cdot \mathcal{O}(1)}, \frac{\Omega(1)\gamma_0^2}{|X_{i_0}|^6} \Big\}.
$$
Eq.~\eqref{eq:mainchoiceofR} ensures that
$$
|X_{i_0}|= \orderof{1} R^D=\beta^{\Omega(1)}\log\Big(\frac{1}{\|A_{(i_0)}\sqrt{\eta}\|_F}\Big)^{\Omega(1)},
$$ 
where we have used the assumption that lattice dimension $D$ is $\orderof{1}$. Plugging in this expression for $|X_{i_0}|$ with the choice of $\gamma_0$, we find
\begin{align*}
\tr(A^2\rho_\beta)=\langle A^2\rangle &\geq \min\Big\{\|A_{(i_0)}\sqrt{\eta }\|_F^{\beta\cdot \mathcal{O}(1)}\cdot e^{-\beta\mathcal{O}(1)|X_{i_0}|}, \frac{\Omega(1)\|A_{(i_0)}\sqrt{\eta }\|_F^{\Omega(1)}}{|X_{i_0}|^6} \Big\}\\
    &\geq \min\Big\{\|A_{(i_0)}\sqrt{\eta }\|_F^{\beta^{\Omega(1)}}, \frac{\Omega(1)\|A_{(i_0)}\sqrt{\eta }\|_F^{\Omega(1)}}{\beta^{\orderof{1}}\log(\frac{1}{\|A_{(i_0)}\sqrt{\eta}\|_F})^{\orderof{1}}}\Big\}\\
    &\geq \min\Big\{\|A_{(i_0)}\sqrt{\eta }\|_F^{\beta^{\Omega(1)}}, \frac{\Omega(1)\|A_{(i_0)}\sqrt{\eta }\|_F^{\Omega(1)}}{\beta^{\orderof{1}}} \Big\}\\
    &\geq \|A_{(i_0)}\sqrt{\eta }\|_F^{\beta^{\Omega(1)}}.
\end{align*}
Since we chose $i_0$ in Eq.~\eqref{eq:definitionofi0} such that $\|A_\loccio\sqrt{\eta }\|_F=\max_i \|A_\locci\sqrt{\eta }\|_F$, this proves the theorem.
\end{proof}

\subsection{Proof of Claims \ref{claim1_lower} and \ref{claim2_lower}}
\label{subsec:2claimproofs}

\begin{proof}[Proof of Claim~\ref{claim1_lower}]
Recall that the goal is to show that for every  $X\subseteq \Lambda$ and arbitrary unitaries $U_X$,  
\begin{align}
&\| {Q}_\gamma^A   U_X \sqrt{\rho_\beta}  \|_F^2\le c_1 e^{\lambda |X|}\delta_\gamma^{\frac{c_2}{c_2+\beta}},
\label{Recall_inequality}
\end{align}
where $Q^A_\gamma=\iden-P^A_\gamma$ and $P^A_\gamma$ was defined in  Eq.~\eqref{eq:defnofPgamma}. 
To prove this inequality, we start from the following expression:
\begin{align}
\|  {Q}_\gamma^A   U_X \sqrt{\rho_\beta}  \|_F^2&= \left \| \sum_{m \in \mathbb{Z}}  {Q}_\gamma^A   U_X P_m^{H'} \sqrt{ \rho_\beta} \right \|_F^2 =\sum_{m \in \mathbb{Z}} \left \|  {Q}_\gamma^A   U_X P_m^{H'} \sqrt{ \rho_\beta} \right \|_F^2 \label{Q_A_gamma_U_X_rho_beta}
\end{align}
with 
\begin{align}
P^{H'}_m := \sum_{j:\Ene_j \in (m , m+1]} \ket{j}\bra{j}, \label{def_P_H_m}
\end{align}
where $\ket{j}$ is the eigenvector of the Hamiltonian $H'$ with $\Ene_j$ the corresponding eigenvalue.
Note that $\sum_{m\in \mathbb{Z}}P^{H'}_m = \iden$ and  we have $P^{H'}_m=0$ for $m \notin [-\|H'\| , \|H'\|]$.
For some $\Delta>0$ which we pick later, we now decompose $\left \|  {Q}_\gamma^A   U_X P_m^{H'} \sqrt{ \rho_\beta} \right \|_F^2$ as a sum of the following quantities
\begin{align}
\left \|  {Q}_\gamma^A   U_X P_m^{H'} \sqrt{ \rho_\beta} \right \|_F^2 
= \left \|  {Q}_\gamma^A \bigl(P^{H'}_{> m+\Delta } + P^{H'}_{< m-\Delta } + P^{H'}_{ [m-\Delta,m+\Delta] }\bigr)  U_X P_m^{H'} \sqrt{ \rho_\beta} \right \|_F^2 ,
\label{Decomp_Q_A_gamma_U_X_P_mH}
\end{align}
where 
\begin{align}
P^{H'}_{> m+\Delta }:=\sum_{m'>m+\Delta} P^{H'}_{m'} ,\quad 
P^{H'}_{< m+\Delta }:=\sum_{m'<m+\Delta} P^{H'}_{m'} ,\quad 
P^{H'}_{[m-\Delta,m+\Delta] }:=\sum_{m-\Delta\le m'\le m+\Delta} P^{H'}_{m'} .
\label{definition_of_P_H_range}
\end{align}
Summing over all $m\in \mathbb{Z}$ in Eq.~\eqref{Decomp_Q_A_gamma_U_X_P_mH} and using Eq.~\eqref{Q_A_gamma_U_X_rho_beta} followed by the triangle inequality gives us the following inequality
\begin{align}
\begin{aligned}
&\|  {Q}_\gamma^A   U_X \sqrt{\rho_\beta}  \|_F^2\\
&\le 
2\underbrace{\sum_{m \in \mathbb{Z}}\left \|  {Q}_\gamma^A P^{H'}_{ [m-\Delta,m+\Delta] }  U_X P_m^{H'} \sqrt{ \rho_\beta} \right \|_F^2}_{:=(1)} 
+2\underbrace{\sum_{m \in \mathbb{Z}}\left \|  {Q}_\gamma^A \bigl(P^{H'}_{> m+\Delta } + P^{H'}_{< m-\Delta } \bigr) U_X P_m^{H'} \sqrt{ \rho_\beta} \right \|_F^2}_{:=(2)}.
\end{aligned}
 \label{Q_A_gamma_U_X_rho_beta_2}
\end{align}
We first bound (1) in Eq.~\eqref{Q_A_gamma_U_X_rho_beta_2}. Note that for every  $m$, 
\begin{align}
\begin{aligned}
\left \|   {Q}_\gamma^A P^{H'}_{ [m-\Delta,m+\Delta] }  U_X P_m^{H'} \sqrt{ \rho_\beta} \right \|_F^2
&\le  \| P_m^{H'} \sqrt{ \rho_\beta} \|^2 \cdot   
\left \|   {Q}_\gamma^A P^{H'}_{ [m-\Delta,m+\Delta] }\right \|_F^2   \\
&\le  e^{-\beta m }\left \|   {Q}_\gamma^A P^{H'}_{ [m-\Delta,m+\Delta] }\right \|_F^2,
\end{aligned}
\label{upper_bound_Q_A_gamma_P_H}
\end{align}
where the first inequality used Eq.~\eqref{fact:relatingfrobeniusAB}. The expression in the last line can be upper bounded as
\begin{align*}
e^{-\beta m }\left \|   {Q}_\gamma^A P^{H'}_{ [m-\Delta,m+\Delta] }\right \|_F^2 &= e^{-\beta m }\Tr\left[{Q}_\gamma^A P^{H'}_{ [m-\Delta,m+\Delta]}\right]\\
&\leq e^{-\beta m }e^{\beta(m+\Delta+1)}\Tr\left[{Q}_\gamma^A P^{H'}_{ [m-\Delta,m+\Delta]}\rho_\beta P^{H'}_{ [m-\Delta,m+\Delta]}\right]\\
&= e^{\beta(\Delta+1)}\left \|   {Q}_\gamma^A P^{H'}_{ [m-\Delta,m+\Delta] }\sqrt{\rho_\beta}\right \|_F^2.
\end{align*} 
where the inequality follows from $$e^{-\beta(m+\Delta+1)}P^{H'}_{ [m-\Delta,m+\Delta]}\preceq P^{H'}_{ [m-\Delta,m+\Delta]}\rho_\beta P^{H'}_{ [m-\Delta,m+\Delta]}.$$ Thus we conclude, from Equation \ref{upper_bound_Q_A_gamma_P_H}, that
\begin{align*}
    \left \|   {Q}_\gamma^A P^{H'}_{ [m-\Delta,m+\Delta] }  U_X P_m^{H'} \sqrt{ \rho_\beta} \right \|_F^2 &\leq e^{\beta(\Delta+1)}\left \|   {Q}_\gamma^A P^{H'}_{ [m-\Delta,m+\Delta] }\sqrt{\rho_\beta}\right \|_F^2\\
    &=e^{\beta(\Delta+1)}\sum_{m'\in [m-\Delta,m+\Delta]}\left \|   {Q}_\gamma^A P^{H'}_{ m' }\sqrt{\rho_\beta}\right \|_F^2.
\end{align*}
So the first term (1) in Eq.~\eqref{Q_A_gamma_U_X_rho_beta_2} can be bounded by
\begin{align}
\sum_{m \in \mathbb{Z}}\left \|  {Q}_\gamma^A P^{H'}_{ [m-\Delta,m+\Delta] }  U_X P_m^{H'} \sqrt{ \rho_\beta} \right \|_F^2 &\le 
 e^{\beta(\Delta+1)}\sum_{m\in \mathbb{Z}}\sum_{m'\in [m-\Delta,m+\Delta]}\left \|   {Q}_\gamma^A P^{H'}_{ m' }\sqrt{\rho_\beta}\right \|_F^2\notag\\
 &\overset{(1)}= e^{\beta(\Delta+1)}\cdot 2\Delta\sum_{m'}\left \|   {Q}_\gamma^A P^{H'}_{ m' }\sqrt{\rho_\beta}\right \|_F^2\notag\\
 &= 2\Delta e^{\beta(\Delta+1)} \left \|   {Q}_\gamma^A\sqrt{\rho_\beta}\right \|_F^2 = 2\Delta e^{\beta(\Delta+1)}\delta_\gamma,
 \label{Q_A_gamma_U_X_rho_beta_2_first}
\end{align}
where in $(1)$ we use the fact that each $m'$ appears $2\Delta$ times in the summation $\sum_{m \in \mathbb{Z}} \sum_{m'\in [m-\Delta,m+\Delta]}$.

We now move on to upper bound $(2)$ in Eq.~\eqref{Q_A_gamma_U_X_rho_beta_2}  as follows.
We have 
\begin{align}
\left \|  {Q}_\gamma^A \bigl(P^{H'}_{> m+\Delta } + P^{H'}_{< m-\Delta } \bigr) U_X P_m^{H'} \sqrt{ \rho_\beta} \right \|_F^2 
\le \left \| \bigl(P^{H'}_{> m+\Delta } + P^{H'}_{< m-\Delta } \bigr) U_X P_m^{H'} \right \| \cdot  \| P_m^{H'} \sqrt{ \rho_\beta} \|_F^2  ,
\label{upper_bound_Q_A_gamma_P_H2}
\end{align}
where we use $\|{Q}_\gamma^A\|\le 1$. 
Using Lemma~\ref{lem:AKL16}, we obtain 
\begin{align}
\left\| \bigl(P^{H'}_{> m+\Delta } + P^{H'}_{< m-\Delta } \bigr) U_X   P^{H'}_m \right\| \le  C e^{-\lambda (\Delta -|X| ) }, \label{ineq:Arad16}
\end{align}
where $C$ and $\lambda$ are universal constants.
Plugging Eq.~\eqref{ineq:Arad16} into Eq.~\eqref{upper_bound_Q_A_gamma_P_H2}, we get
\begin{align}
\left \|  {Q}_\gamma^A \bigl(P^{H'}_{> m+\Delta } + P^{H'}_{< m-\Delta } \bigr) U_X P_m^{H'} \sqrt{ \rho_\beta} \right \|_F^2 
\le  C e^{-\lambda (\Delta -|X| ) }  \| P_m^{H'} \sqrt{ \rho_\beta} \|_F^2  .
\label{upper_bound_Q_A_gamma_P_H2_2}
\end{align}
With this, we can bound (2) in Eq.~\eqref{Q_A_gamma_U_X_rho_beta_2}  by
\begin{align}
\sum_{m \in \mathbb{Z}}\left \|  {Q}_\gamma^A \bigl(P^{H'}_{> m+\Delta } + P^{H'}_{< m-\Delta } \bigr) U_X P_m^{H'} \sqrt{ \rho_\beta} \right \|_F^2   \le 
\sum_{m \in \mathbb{Z}}C e^{-\lambda (\Delta -|X| ) }  \| P_m^{H'} \sqrt{ \rho_\beta} \|_F^2  = C e^{-\lambda (\Delta -|X| ) }  ,
\label{Q_A_gamma_U_X_rho_beta_2_second}
\end{align}
where the equality used the fact that $\sum_{m\in\mathbb{Z}}\| P_m^{H'} \sqrt{ \rho_\beta} \|_F^2=\tr (\rho_\beta)=1$. Putting together Eq.~\eqref{Q_A_gamma_U_X_rho_beta_2_first} and \eqref{Q_A_gamma_U_X_rho_beta_2_second} into Eq.~\eqref{Q_A_gamma_U_X_rho_beta_2}, we finally obtain the upper bound~of 
\begin{align}
\|  {Q}_\gamma^A   U_X \sqrt{\rho_\beta}  \|_F^2&\le 
4\Delta e^{\beta (\Delta +1)} \delta_\gamma + 2C e^{-\lambda (\Delta -|X| ) }.
\end{align}
We let $\Delta =c \beta^{-1} \log(1/\delta_\gamma)$, which gives 
\begin{align}
&\|  {Q}_\gamma^A   U_X \sqrt{\rho_\beta}  \|_F^2\le c_1 e^{\lambda |X|}\delta_\gamma^{\frac{c_2}{c_2+\beta}}, 
\label{arbirtary_unitary_op_proof}
\end{align}
for some universal constants $c_1,c_2$. This proves the claim statement.
\end{proof}

We now proceed to prove Claim~\ref{claim2_lower}. 

\begin{proof}[Proof of Claim~\ref{claim2_lower}] Recall that the aim is to prove that for every $X\subseteq \Lambda$ and unitary $U_X$ we have
$$
\left \| A    {Q}_\gamma^A U_X \sqrt{\rho_\beta}  \right\|_F^2 \le     \frac{e^{\orderof{1} |X|}\delta_\gamma^{\orderof{1}/\beta}}{\gamma}+ \frac{\orderof{1}|X|^5}{\gamma^5}\langle A^2\rangle
$$
We let $c_5,\lambda_1, \tau_1$ be $\orderof{1}$ constants as defined in Lemma \ref{multicommutator_norm_quasi_local} and $c_1,c_2, \lambda =\orderof{1}$ be constants given by Claim \ref{claim1_lower}. For the proof, we first decompose ${Q}_\gamma^A$ as 
 \begin{align}
&{Q}_\gamma^A = \sum_{s=1}^\infty P^A_s ,\quad P^A_s:= P^A_{(s\gamma,(s+1)\gamma]} +  P^A_{[-(s+1)\gamma,-s\gamma)},
\label{definition_Q_gamma_P_gamma_A}
\end{align}
where $P^A_{[a,b]}$ is defined as $P^A_{[a,b] }:=\sum_{a\le \omega\le b}\Pi_{\omega}$ (where $\Pi_{\omega}$ is the subspace spanned by the eigenvectors of $A$ with eigenvalue $\varepsilon$).
Using this notation, observe that $ \| A {Q}_\gamma^A U_X \sqrt{\rho_\beta}\|_F$ can bounded~by
 \begin{align}
\left \| A   {Q}_\gamma^A U_X \sqrt{\rho_\beta}  \right\|_F^2
=  \sum_{s=1}^\infty  \left \| A  P^A_s U_X \sqrt{\rho_\beta}  \right\|_F^2   
\le& \sum_{s=1}^\infty  \left \| A  P^A_s\right\|^2 \cdot \left\|P^A_s U_X \sqrt{\rho_\beta}  \right\|_F^2 \notag \\
\le& \gamma^2 \sum_{s=1}^\infty (s+1)^2  \left \| P^A_s U_X \sqrt{\rho_\beta}  \right\|_F^2 ,
\label{summ_s_A_uX/}
\end{align}
where we use $\|A  P^A_s\| \le \gamma(s+1)$ from the definition~\eqref{definition_Q_gamma_P_gamma_A} of $P^A_s$. 
The norm $\left \| P^A_s U_X \sqrt{\rho_\beta}  \right\|_F$ is bounded from above by
 \begin{align}
 \left \| P^A_s U_X \sqrt{\rho_\beta}  \right\|_F =  \left \| P^A_s U_X \sum_{s'=0}^\infty  P^A_{s'}  \sqrt{\rho_\beta}  \right\|_F \le  \sum_{s'=0}^\infty \left \| P^A_s U_X  P^A_{s'} \sqrt{\rho_\beta}  \right\|_F,
 \label{upper_bound_P^A_s U_X sqrt_rho_beta}
\end{align}
where we use $\sum_{s'=0}^\infty  P^A_{s'}=\iden$ in the first equation  and in the second inequality we use the triangle inequality for the Frobenius norm. Using Lemma \ref{multicommutator_norm_quasi_local} we additionally have
\begin{align}
\label{eq:PsPs'decayexponential}
\| P^A_s U_X  P^A_{s'}  \| \le  c_5|X| e^{-({\lambda_1}\gamma |s-s'|/|X|)^{1/{\tau_1}}} \quad \text{ for every } s,s'\geq 0,
\end{align}
where $c_5, {\lambda_1}$ are as given in Lemma \ref{multicommutator_norm_quasi_local}. Using this, we have
\begin{align}
\label{eq:PsuP0beta}
\begin{aligned}
    \left \| P^A_s U_X  P^A_{0} \sqrt{\rho_\beta}  \right\|_F
    &=    \left \| P^A_s U_X  P^A_{0} \sqrt{\rho_\beta}  \right\|_F^{1/2}\cdot \left \| P^A_s U_X  P^A_{0} \sqrt{\rho_\beta}  \right\|_F^{1/2}\\
    &\overset{(1)}\leq \br{2c_1 e^{\lambda|X|}\delta_\gamma^{\frac{c_2}{c_2+\beta}}}^{1/2} \cdot \| P^A_s U_X  P^A_{0} \sqrt{\rho_\beta}  \|_F^{1/2}\\
    &\overset{(2)}\leq \br{2c_1 e^{\lambda|X|}\delta_\gamma^{\frac{c_2}{c_2+\beta}}}^{1/2} \cdot \left \| P^A_s U_X  P^A_{0}\right\|^{1/2}\cdot \|\sqrt{\rho_\beta}\|_F^{1/2}\\
    &\overset{(3)}\leq \br{2c_1 e^{\lambda|X|}\delta_\gamma^{\frac{c_2}{c_2+\beta}}}^{1/2}\cdot \br{c_5|X| e^{-(\lambda_1\gamma |s|/|X|)^{1/\tau_1}}}^{1/2}\cdot 1\\
    &\overset{(4)}=(\delta'_\gamma)^{1/2}\cdot \br{c_5|X| e^{-(\lambda_1\gamma |s|/|X|)^{1/\tau_1}}}^{1/2},
    \end{aligned}
\end{align}
where inequality $(1)$ uses $\left \| P^A_s U_X  P^A_{0} \sqrt{\rho_\beta}  \right\|_F\le 2c_1 \delta_\gamma^{\frac{c_2}{c_2+\beta}}$  
\footnote{
Since $Q^A_\gamma \le \iden$ and $P^A_{0}=P^A_\gamma=\iden-Q^A_\gamma$, we have
$$
\left \| P^A_s U_X  P^A_{0} \sqrt{\rho_\beta}  \right\|_F 
\le \left\| Q^A_\gamma U_X  (\iden-Q^A_\gamma) \sqrt{\rho_\beta} \right\|_F  \le \left\| Q^A_\gamma U_X  \sqrt{\rho_\beta} \right\|_F 
+ \left\| Q^A_\gamma U_X Q^A_{\gamma}  \sqrt{\rho_\beta} \right\|_F
\le \left\| Q^A_\gamma U_X  \sqrt{\rho_\beta} \right\|_F 
+ \delta_\gamma
\le 2c_1  e^{\lambda|X|}\delta_\gamma^{\frac{c_2}{c_2+\beta}},
$$
where the first inequality used $P_s^A\leq Q^A_\gamma$ and the last inequality used $\| Q^A_\gamma U_X  \sqrt{\rho_\beta} \|_F\leq c_1e^{\lambda|X|}\delta_{\gamma}^{\frac{c_2}{c_2+\beta}}$ from Claim~\ref{claim1_lower}.},
inequality $(2)$ uses Eq.~\eqref{fact:relatingfrobeniusAB}, inequality $(3)$ uses Eq.~\eqref{eq:PsPs'decayexponential} and the fact that $\|\sqrt{\rho_\beta}\|_F=\Tr(\rho_\beta)=1$ and equality $(4)$ defines $\delta'_\gamma:=2c_1 e^{\lambda|X|}\delta_\gamma^{\frac{c_2}{c_2+\beta}}$. Using Eq.~\eqref{eq:PsuP0beta}, we obtain the following
 \begin{align}
\sum_{s'=0}^\infty  \left \| P^A_s U_X  P^A_{s'} \sqrt{\rho_\beta}  \right\|_F 
&\le \left \| P^A_s U_X  P^A_{0} \sqrt{\rho_\beta}  \right\|_F   + \sum_{s'=1}^\infty  \left \| P^A_s U_X  P^A_{s'}  \right\| \cdot \left \| P^A_{s'} \sqrt{\rho_\beta}  \right\|_F \notag\\
&\le {\delta'}_\gamma^{1/2} c_5^{1/2}|X|^{1/2}e^{-(\lambda_1\gamma s/|X|)^{1/\tau_1}/2} \notag\\
& +  \sum_{s'=1}^\infty c_5|X| e^{-(\lambda_1\gamma |s-s'|/|X|)^{1/\tau_1}}\left\| P^A_{s'} \sqrt{\rho_\beta}  \right\|_F  ,
\label{sum_s'_P_A_U_X_P_A_s'}
\end{align}
where the first term in the inequality was obtained from Eq.~\eqref{eq:PsuP0beta} and the second term was obtained from Eq.~\eqref{eq:PsPs'decayexponential}. 

We now upper bound the summation in the second term of Eq.~\eqref{sum_s'_P_A_U_X_P_A_s'} by using the Cauchy--Schwarz inequality as follows: 
\begin{align}
&\sum_{s'=1}^\infty e^{-(\lambda_1\gamma |s-s'|/|X|)^{1/\tau_1}} \left\| P^A_{s'} \sqrt{\rho_\beta}  \right\|_F \notag\\
&=\sum_{s'=1}^\infty e^{-(\lambda_1\gamma |s-s'|/|X|)^{1/\tau_1}/2}\cdot  \big(e^{-(\lambda_1\gamma |s-s'|/|X|)^{1/\tau_1}/2} \left\| P^A_{s'} \sqrt{\rho_\beta}  \right\|_F\big) \notag \\
&\le  \left(\sum_{s'=1}^\infty e^{-(\lambda_1\gamma |s-s'|/|X|)^{1/\tau_1}} \right)^{1/2}  \left(\sum_{s'=1}^\infty e^{-(\lambda_1\gamma |s-s'|/|X|)^{1/\tau_1}}\left\| P^A_{s'} \sqrt{\rho_\beta}  \right\|_F ^2 \right)^{1/2} \notag \\
&\overset{(1)}\le \br{\frac{4\tau_1|X|}{\lambda_1\gamma}\br{2\tau_1}^{\tau_1}}^{1/2}\cdot \left(\sum_{s'=1}^\infty p_{s'}^A e^{-(\lambda_1\gamma |s-s'|/|X|)^{1/\tau_1}}   \right)^{1/2},
\label{The summation in the second term of sum_s'_P_A_U_X_P_A_s'}
\end{align}
where $p_{s'}:= \left\| P^A_{s'} \sqrt{\rho_\beta}  \right\|_F ^2$ and we used Fact \ref{fact:integrals} in inequality $(1)$.
Note that  $\sum_{s'=1}^\infty p_{s'} =\| {Q}^A_{\gamma}\sqrt{\rho_\beta}  \|_F ^2$ because of $P_0^A=P^A_{(0,\gamma]} +  P^A_{[-\gamma,0)}={P}^A_{\gamma}$.
 We can obtain the following upper bound by combining the equations Eq.~\eqref{upper_bound_P^A_s U_X sqrt_rho_beta}, \eqref{sum_s'_P_A_U_X_P_A_s'} and \eqref{The summation in the second term of sum_s'_P_A_U_X_P_A_s'}:
\begin{align*}
&\left \| P^A_s U_X \sqrt{\rho_\beta}  \right\|_F^2 \\
&\le  \Big(\sum_{s'=0}^\infty \left \| P^A_s U_X  P^A_{s'} \sqrt{\rho_\beta}  \right\|_F\Big)^2\\
&\leq \Big({\delta'}_\gamma^{1/2} c_5^{1/2}|X|^{1/2}e^{-(\lambda_1\gamma s/|X|)^{1/\tau_1}/2}+\sum_{s'=1}^\infty c_5|X| e^{-(\lambda_1\gamma |s-s'|/|X|)^{1/\tau_1}}\left\| P^A_{s'} \sqrt{\rho_\beta}  \right\|_F\Big)^2 \\
&\leq \Big({\delta'}_\gamma^{1/2} c_5^{1/2}|X|^{1/2}e^{-(\lambda_1\gamma s/|X|)^{1/\tau_1}/2}  +  c_5|X|\br{\frac{4\tau_1|X|}{\lambda_1\gamma}\br{2\tau_1}^{\tau_1}}^{1/2}\cdot \left(\sum_{s'=1}^\infty p_{s'}^A e^{-(\lambda_1\gamma |s-s'|/|X|)^{1/\tau_1}}   \right)^{1/2}\Big)^2\\
&\leq \underbrace{2 c_5{\delta'}_\gamma|X| e^{-(\lambda_1\gamma s/|X|)^{1/\tau_1}}}_{:=f_1(s)} +\underbrace{\frac{8c_5^2\tau_1|X|^3}{\lambda_1\gamma}\br{2\tau_1}^{\tau_1}\cdot \left(\sum_{s'=1}^\infty p_{s'}^A e^{-(\lambda_1\gamma |s-s'|/|X|)^{1/\tau_1}}  \right). }_{:=f_2(s)} 
\end{align*}
Recall that the goal  of this claim was to upper bound Eq.~\eqref{summ_s_A_uX/}, which we can rewrite now as
\begin{equation}
\left \| A   {Q}_\gamma^A U_X \sqrt{\rho_\beta}  \right\|_F^2
\le \gamma^2 \sum_{s=1}^\infty (s+1)^2  \left \| P^A_s U_X \sqrt{\rho_\beta}  \right\|_F^2 \leq \gamma^2\sum_{s=1}^\infty (s+1)^2f_1(s)+\gamma^2\sum_{s=1}^\infty (s+1)^2f_2(s).  
    \label{eq:AQurhobetabound1}
\end{equation}
We bound each of these terms separately. In order to bound the first term observe that
\begin{align*}
\begin{aligned}
    \gamma^2\sum_{s=1}^\infty (s+1)^2f_1(s)&=2\gamma^2 c_5\delta'_\gamma|X|^2\sum_{s=1}^\infty (s+1)^2 e^{-(\lambda_1\gamma s/|X|)^{1/\tau_1}} \\
    &\overset{(1)}\leq 2\gamma^2c_5{\delta'}_\gamma|X|\cdot 8\tau_1\cdot\br{\frac{(3\tau_1)^{\tau_1}|X|}{\lambda_1\gamma}}^{3}\leq \frac{16c_5\delta'_\gamma|X|^4\tau_1(3\tau_1)^{3\tau_1}}{\lambda_1^3\gamma},
\end{aligned}
\end{align*}
where inequality $(1)$ uses Fact \ref{fact:integrals}. We now bound the second term in Eq.~\eqref{eq:AQurhobetabound1} as follows
\begin{align}
    \begin{aligned}
       \gamma^2\sum_{s=1}^\infty (s+1)^2f_2(s)&=\gamma^2\cdot\frac{8c_5^2\tau_1|X|^3}{\lambda_1\gamma}\br{2\tau_1}^{\tau_1}
\sum_{s=1}^\infty (s+1)^2\left(\sum_{s'=1}^\infty p_{s'}^A e^{-(\lambda_1\gamma |s-s'|/|X|)^{1/\tau_1}}  \right).  \\
&=\frac{8\gamma c_5^2\tau_1|X|^3\br{2\tau_1}^{\tau_1}}{\lambda_1}
 \sum_{s'=1}^\infty p_{s'}^A \Big(\sum_{s=1}^\infty (s+1)^2 e^{-(\lambda_1\gamma |s-s'|/|X|)^{1/\tau_1}}\Big)\\
&\le \frac{8\gamma c_5^2\tau_1|X|^3\br{2\tau_1}^{\tau_1}}{\lambda_1}
 \sum_{s'=1}^\infty p_{s'}^A(2s')^2 \Big(\sum_{s=1}^\infty (1+|s-s'|)^2 e^{-(\lambda_1\gamma |s-s'|/|X|)^{1/\tau_1}}\Big)\\
 &\overset{(1)}\le \frac{8\gamma c_5^2\tau_1|X|^3\br{2\tau_1}^{\tau_1}}{\lambda_1}\cdot 16\tau_1\cdot\br{\frac{\br{3\tau_1}^{\tau_1}|X|}{\lambda_1\gamma}}^{3}\sum_{s'=1}^\infty p_{s'}^A(2s')^2,
 \end{aligned}
 \end{align}
  where inequality $(1)$ follows from Fact \ref{fact:integrals}. Further upper bound this expression by simplifying the pre-factors, we get
 \begin{align}
 \begin{aligned}
    \gamma^2\sum_{s=1}^\infty (s+1)^2f_2(s)&\leq \frac{512 c_5^2|X|^6\tau_1^2(3\tau_1)^{4\tau_1}}{\lambda_1^4\gamma^2}\sum_{s'=1}^\infty p_{s'}^A (s')^2\\
 &= \frac{512 c_5^2|X|^6\tau_1^2(3\tau_1)^{4\tau_1}}{\lambda_1^4}\sum_{s'=1}^\infty (\gamma s')^2 p_{s'}^A \\
 &=\frac{512 c_5^2|X|^6\tau_1^2(3\tau_1)^{4\tau_1}}{\lambda_1^4}\sum_{s'=1}^\infty  \left\|P^A_{s'}(\gamma s') \sqrt{\rho_\beta}  \right\|_F ^2\\
 &\overset{(2)}\le \frac{512 c_5^2|X|^6\tau_1^2(3\tau_1)^{4\tau_1}}{\lambda_1^4}\sum_{s'=0}^\infty  \left\|P^A_{s'}A \sqrt{\rho_\beta}  \right\|_F ^2 = \frac{512 c_5^2|X|^6\tau_1^2(3\tau_1)^{4\tau_1}}{\lambda_1^4}\langle A^2\rangle,
 \end{aligned}
\end{align}
 In inequality $(2)$, we used $P^A_{s'}(\gamma s')\preceq P^A_{s'}$ from the definition~\eqref{definition_Q_gamma_P_gamma_A} of $P^A_{s'}$.
By combining the above inequalities altogether, we prove Eq.~\eqref{unitary_A_norm_small}.
\end{proof}
\bibliographystyle{alpha}
 \bibliography{refs}

\addtocontents{toc}{\setcounter{tocdepth}{1}}


\appendix
\section{Proof of Fact~\ref{fact:integrals}}\label{sec:proof of integral facts}
Here we restate and prove the following fact.
\begin{fact}[Restatement of Fact~\ref{fact:integrals}]

Let $a,c,p>0$ be reals and $b$ be a positive integer. Then
\begin{enumerate}
\item[1)] $\sum_{j=0}^{\infty} e^{-cj} \leq \frac{e^c}{c}$.
\item[2)] $\sum_{j=0}^{\infty} j^be^{-cj^p} \leq \frac{2}{p}\cdot\br{\frac{b+1}{cp}}^{\frac{b+1}{p}}$.
\item[3)] $\sum_{j=0}^{\infty} e^{-c(a+j)^p} \leq e^{-\frac{c}{2}a^p}\br{1+\frac{1}{p }\br{\frac{2}{cp}}^{\frac{1}{p}}}$.
\end{enumerate}
\end{fact}

\begin{proof}
The first summation follows from
$$\sum_{j=0}^{\infty} e^{-cj}= \frac{1}{1-e^{-c}}=\frac{e^c}{e^c-1}\leq \frac{e^c}{c}.$$
For the second sum, notice that the function $t^be^{-ct^p}$ achieves the maximum at $t^*=\br{\frac{b}{cp}}^{\frac{1}{p}}$. Then
\begin{align*}
    \sum_{j=0}^{\infty} j^b e^{-cj^p} &\leq t^*\br{t^*}^b e^{-c\br{t^*}^p}+\int_{0}^{\infty} t^be^{-ct^p} dt\\
    &= \br{\frac{b}{cp}}^{\frac{b+1}{p}}e^{-\frac{b}{p}}+\frac{1}{(b+1)c^{\frac{b+1}{p}}}\int_{0}^{\infty} e^{-y^{\frac{p}{b+1}}} dy\\
    &=\br{\frac{b}{e^{\frac{b}{b+1}}cp}}^{\frac{b+1}{p}}+\frac{1}{pc^{\frac{b+1}{p}}}\Gamma\br{\frac{b+1}{p}}\\
    &\leq \br{\frac{b}{e^{\frac{b}{b+1}}cp}}^{\frac{b+1}{p}}+\frac{1}{pc^{\frac{b+1}{p}}}\br{\frac{b+1}{p}}^{\frac{b+1}{p}}\leq \frac{2}{p}\cdot\br{\frac{b+1}{cp}}^{\frac{b+1}{p}}. 
\end{align*}
 For the third sum, we will use the identity 
$$(a+j)^p\geq 2^{p-1}\br{a^p+j^p}\geq \frac{1}{2}\br{a^p+j^p}.$$
This is clearly true if $p\geq 1$. For $p<1$, we use concavity. Now, consider the following chain of inequalities and change of variables:
\begin{align*}
    \sum_{j=0}^{\infty} e^{-c(a+j)^p} &\leq  e^{-\frac{c}{2} a^p}\sum_{\ell=0}^{\infty} e^{-\frac{c}{2}\ell^p}\\
    &\leq e^{-\frac{c}{2}a^p}\br{1+\int_{0}^{\infty} e^{-\frac{c}{2}t^p} dt}\\
    &= e^{-\frac{c}{2}a^p}\br{1+\frac{2^{\frac{1}{p}}}{c^{\frac{1}{p}}}\int_{0}^{\infty} e^{-y^p} dy}\\
    &=e^{-\frac{c}{2}a^p}\br{1+\frac{2^{\frac{1}{p}}}{pc^{\frac{1}{p}}}\Gamma\left(\frac{1}{p}\right)}\leq e^{-\frac{c}{2}a^p}\br{1+\frac{1}{p c^{\frac{1}{p}}}\br{\frac{2}{p}}^{\frac{1}{p}}}.
\end{align*}
This completes the proof.
\end{proof}

\section{Fourier transform of $\pmb{\tanh(\beta\omega/2)/(\beta\omega/2)}$} \label{sec:Hamiltonian construction}
We here derive the Fourier transform of 
\begin{align*}
\tilde{f}_\beta(\omega)= \frac{\tanh(\beta\omega/2)}{\beta\omega/2},
\end{align*}
which is
 \begin{align*}
f_\beta(t):=\frac{1}{2\pi} \int_{-\infty}^\infty e^{i\omega t} \tilde{f}_\beta(\omega) d\omega.
\end{align*}

\begin{figure}
\centering
\subfigure[Integral path $C^+$ of $\omega$ for $t>0$]
{\includegraphics[clip, scale=0.3]{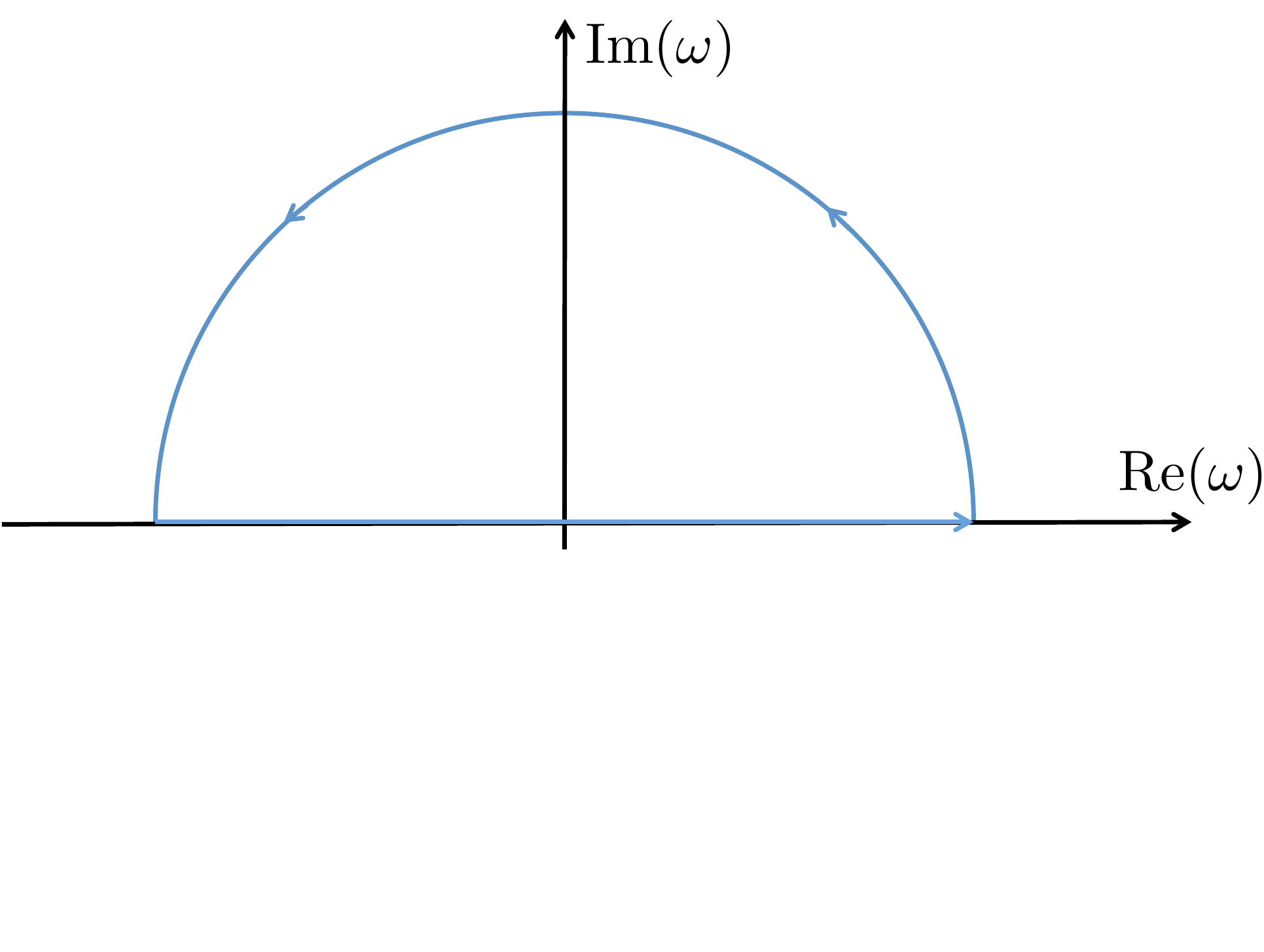}
}
\subfigure[Integral path $C^-$ of $\omega$ for $t<0$]
{\includegraphics[clip, scale=0.3]{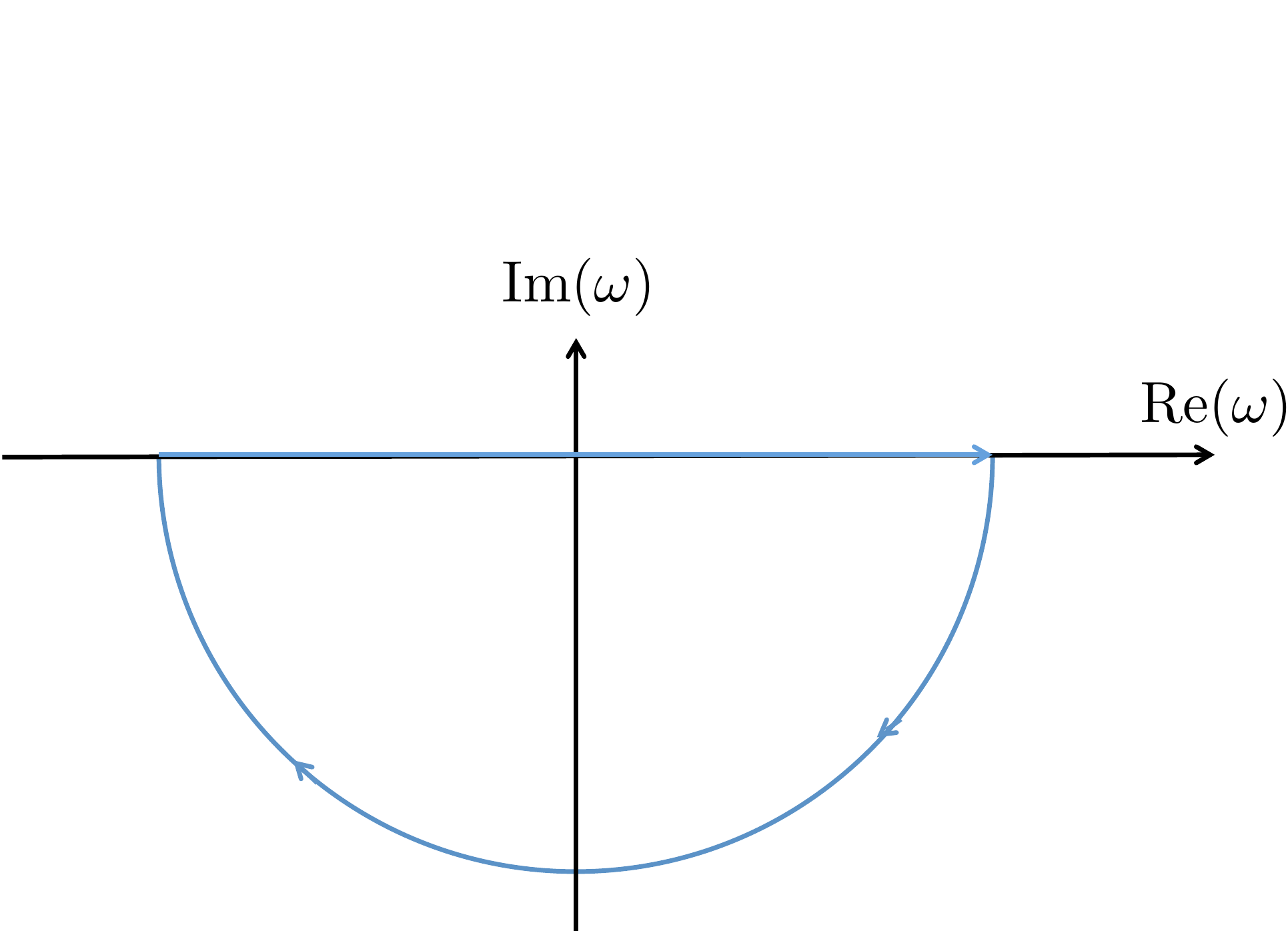}
}
\caption{Cauchy's integral theorem for the calculation of the Fourier transform.}
\label{fig:int_path}
\end{figure}

For the calculation of the Fourier transform, we first consider the case of $t>0$.
By defining $C^+$ as a integral path as in Fig.~\ref{fig:int_path} (a), we obtain
\begin{align}
\frac{1}{2\pi}  \int_{-\infty}^\infty e^{i\omega t} \tilde{f}_\beta(\omega) d\omega
&=\frac{1}{2\pi}   \int_{C^+} e^{i\omega t}\tilde{f}_\beta(\omega) d\omega \notag \\
&=i \sum_{m=0}^\infty {\rm Res}_{\omega=i\pi+2i m\pi} [e^{i\omega t} \tilde{f}_\beta(\omega)].
\end{align}
Note that the singular points of $[e^{i\omega t}\tilde{f}_\beta(\omega)]$ are given by $\beta \omega= i\pi (2 m+1)$ with $m$ integers.
We can calculate the residue as 
\begin{align}
{\rm Res}_{\beta \omega=i\pi+2i m\pi} [e^{i\omega t} \tilde{f}_\beta(\omega)] = \frac{4e^{-(2m+1)\pi t/\beta}}{\beta\pi} \frac{-i}{2m+1}
\end{align}
We thus obtain
\begin{align}
f_\beta(t)=\frac{1}{2\pi}  \int_{-\infty}^\infty e^{i\omega t} \tilde{f}_\beta(\omega)  d\omega
&=\frac{4}{\beta \pi}  \sum_{m=0}^\infty  \frac{ e^{-(2m+1)\pi t/\beta}}{2m+1}.
\end{align}
for $t>0$.

We can perform the same calculation for $t<0$.
In this case, we define $C^-$ as a integral path as in Fig.~\ref{fig:int_path} (b), and obtain
\begin{align}
f_\beta(t)&=\frac{1}{2\pi}   \int_{C^-} e^{i\omega t}\tilde{f}_\beta(\omega) d\omega 
=-i \sum_{m=0}^\infty {\rm Res}_{\omega=-i\pi-2i m\pi} [e^{i\omega t} \tilde{f}_\beta(\omega)] .
\end{align}
By using 
\begin{align}
{\rm Res}_{\omega=-i\pi-2i m\pi} [e^{i\omega t}\tilde{f}_\beta(\omega) ] = \frac{4e^{(2m+1)\pi t/\beta}}{\beta\pi} \frac{i}{2m+1} ,
\end{align}
we have
\begin{align}
f_\beta(t)=\frac{1}{2\pi}  \int_{-\infty}^\infty e^{i\omega t} \tilde{f}_\beta(\omega)  d\omega
&= \frac{4}{\beta\pi}  \sum_{m=0}^\infty  \frac{ e^{(2m+1)\pi t/\beta}}{2m+1} .
\end{align}
for $t<0$.
By combining the above expressions for $f_\beta(t)$, we arrive at 
\begin{align}
f_\beta(t)= \frac{4}{\beta\pi}  \sum_{m=0}^\infty  \frac{ e^{-(2m+1)\pi |t|/\beta}}{2m+1} .
\end{align}
The summation is calculated as
\begin{align}
\sum_{m=0}^\infty  \frac{ e^{-(2m+1)x}}{2m+1} = \int_x^\infty \sum_{m=0}^\infty   e^{-(2m+1)x' } dx' 
=  \int_x^\infty \frac{1}{e^{x'} -e^{-x'}} dx' = \frac{1}{2}\log\frac{e^x+1}{e^x-1}
\end{align}
for $x>0$, which yields
\begin{align}
f_\beta(t)= \frac{2}{\beta\pi}\log \frac{e^{\pi |t|/\beta}+1}{e^{\pi |t|/\beta}-1} .
\end{align}
Since 
$$\log \frac{e^{\pi |t|/\beta}+1}{e^{\pi |t|/\beta}-1}\leq \frac{2}{e^{\pi |t|/\beta}-1},$$
$f_\beta(t)$ shows an exponential decay in $|t|$.

\section{Derivation of the sub-exponential concentration} \label{sec:Derivation of norm quasi local}

Recall that the goal in this appendix is to prove the following lemma.
\begin{lem} [Restatement of Lemma~\ref{multicommutator_norm_quasi_local}]
\label{res:multicommutator_norm_quasi_local}
Let $A$ be a $(\tau, a_1, a_2, 1)$-quasi-local operator with $\tau<1$, as given in Eq.~\eqref{quasi_locality of A}. For an arbitrary operator $O_X$ supported on a subset $X\subseteq \Lambda$  with $|X|=k_0$ and $\|O_X\|=1$, we~have
\begin{align}
\| P^A_{\ge x+y} O_X  P^A_{\le x}  \|  \le c_5\cdot  k_0 \exp \Big(-(\lambda_1 y/k_0)^{1/\tau_1}\Big), 
\label{res:Concentration_lemma_subexponential1}
\end{align}
where $\tau_1:=\frac{2}{\tau}-1$ and $c_5$ and $\lambda_1$ are constants which only depend on $a_1$ and $a_2$. 
In particular, the $a_2$ dependence of $c_5$ and $\lambda_1$ is given by 
$c_5 \propto a_2^{2/\tau}$ and $\lambda_1\propto a_2^{-2/\tau}$ respectively.
\end{lem}
Before proving this lemma, let us elaborate upon the method. Recall that
\begin{align}
P^A_{\le x} = \sum_{\omega\leq x} \Pi_\omega ,\quad 
P^A_{>y} = \sum_{\omega> y} \Pi_\omega,
\end{align}
where $\Pi_\omega$ is the projector onto the eigenvalue $\omega$ eigenspace of $A$. One way to prove the upper bound in the estimation of the norm~\eqref{res:Concentration_lemma_subexponential1} is to utilize the technique in Ref.~\cite{Arad_connecting_global_local_dist} (i.e., Lemma~\ref{lem:AKL16}). The argument proceeds by considering
\begin{align}
\| P^A_{\ge x+y} O_X  P^A_{\le x}  \| 
&= \| P^A_{\ge x+y} e^{-\nu A}  e^{\nu A}O_X   e^{-\nu A}  e^{\nu A}P^A_{\le x}  \| \notag \\
&\le \| P^A_{\ge x+y} e^{-\nu A}\| \cdot \|  e^{\nu A}O_X   e^{-\nu A} \| \cdot \| e^{\nu A}P^A_{\le x}  \| \notag \\
&\le e^{-\nu x}\|  e^{\nu A}O_X   e^{-\nu A} \|, \label{Itai's_technique}
\end{align}
which reduces the problem to estimation of the norm $\|  e^{\nu A}O_X   e^{-\nu A} \|$.
Additionally, by definition of~$A$ in Theorem~\ref{thm:variancelowerbound} we have
\begin{align}
A= \sum_{\ell=1}^n  {g}_{\ell} \bar{A}_\ell ,
\end{align}
where $\bar{A}_\ell$ is $\k$-local and $ {g}_{\ell} $ is sub-exponentially decaying function for $\ell$ (as made precise in Eq.~\eqref{quasi_locality of A}), namely ${g}_{\ell} =\exp(-\orderof{\ell^{1/D}})$. 
In this case, for $\nu=\orderof{1}$, 
the norm of the imaginary time evolution can be finitely bounded only in the case $D=1$~\cite{kuwahara2016asymptotic}. 
That is, the norm $\|  e^{\nu A}O_X   e^{-\nu A} \|$ diverges to infinity for $D\ge 2$. However, our main contribution in this section is that  we are able to prove the lemma statement  \emph{without} going through the inequalities in~\eqref{Itai's_technique} (which in turn used earlier results of~\cite{kuwahara2016asymptotic,Arad_connecting_global_local_dist}). We now give more details.

\subsection{Proof of Lemma~\ref{res:multicommutator_norm_quasi_local}} \label{Proof of Lemma multicommutator_norm_quasi_local}
In order to estimate the norm, we need to take a different route from \eqref{Itai's_technique}.  
Let $I$ be any interval of the real line and $P^A_I$ be the projector onto the eigenspace of $A$ with eigenvalues in $I$. Using the operator inequality $$P^A_{\ge z} (A-\omega\iden)^m\succeq (z-\omega)^m P^A_{\ge x+y},$$ we obtain
\begin{align}
\|  (A-\omega\iden )^m O_X  P^A_I \|\ge \|  P^A_{\ge z} (A-\omega)^m O_X  P^A_I \| \ge (z-\omega)^m \|  P^A_{\ge z} O_X  P^A_I \|,
\end{align}
hence
\begin{align}
\|  P^A_{\ge z} O_X  P^A_I \| \le \frac{\|  (A-\omega)^m O_X  P^A_I \|}{(z-\omega)^m}.
\label{Markov_ineq_norm}
\end{align}
Our strategy to establish Eq.~\eqref{res:Concentration_lemma_subexponential1} will be to expand 
\begin{align}
\|  P^A_{\ge x+y} O_X  P^A_{\le x} \| \le \sum_{j=0}^\infty \|  P^A_{\ge x+y} O_X P^A_{I_j} \| ,
\label{P^A_ge y O_X  P^A_le x_upper_bound}
\end{align}
for carefully chosen intervals $I_j:=(x- a_1 (j+1), x- a_1 j]$ (the term $a_1$ is as given in the statement of Lemma \ref{res:multicommutator_norm_quasi_local}). Towards this, let us fix an arbitrary $\omega$, an interval $I:=(\omega-a_1,\omega]$ and prove an upper bound on $\|  P^A_{\ge \omega+\theta} O_X P^A_{I_j} \|$ (for all $\theta$). We show the following claim.
\begin{claim}
\label{clm:genericquasiprojupbound}
There is a constant $c_6$ such that
\begin{align}
\|  P^A_{\ge \omega + \theta} O_X  P^A_I \| \le  \frac{1}{\tau}\exp \left[-[\theta/(ec_6 k_0)]^{1/\tau_1} +1\right].
\label{Upper_bound_for_P_A_I}
\end{align}
\end{claim}
The claim is proved in subsection \ref{proof_of_genericquasiprojupbound}. Let us use the claim to establish the lemma.
In the inequality~\eqref{P^A_ge y O_X  P^A_le x_upper_bound}, we need to estimate $\|  P^A_{\ge x+y} O_X P^A_{I_j}\|$ with $I_j:=(x- (j+1)a_1, x-ja_1]$. 
Setting $\omega=x-j a_1$ and $\theta=y+j a_1$ in Claim \ref{clm:genericquasiprojupbound}, we have 
\begin{align}
\|  P^A_{\ge x+y} O_X P^A_{I_j} \|  \le \frac{1}{\tau}\exp \left\{-\left(\frac{y + a_1 j}{ec_6 k_0} \right)^{1/\tau_1}+1 \right\}.
\label{Upper_bound_for_P_A_I_j_2}
\end{align}
In order to complete the bound on Equation \ref{P^A_ge y O_X  P^A_le x_upper_bound}, we need to take summation with respect to $j$. We have
 \begin{align}
\sum_{j=0}^\infty \|  P^A_{\ge x+y} O_X P^A_{I_j} \|  \le \sum_{j=0}^\infty  \frac{1}{\tau}\exp \left\{-\left(\frac{y + a_1 j}{ec_6 k_0} \right)^{1/\tau_1}+1 \right\}\leq \frac{1}{\tau}e^{-\frac{1}{2}\br{\frac{y}{ec_6 k_0}}^{1/\tau_1}}\br{1+\frac{ec_6k_0\tau_1}{a_1}\br{4/\tau}^{1/\tau_1}},
\end{align}
where in last inequality we used Fact \ref{fact:integrals} (3) with $c=(ec_6k_0/a_1)^{-1/\tau_1}$, $p=1/\tau_1$ and $a=y/a_1$. This gives the form of \eqref{res:Concentration_lemma_subexponential1} and completes the proof.

\subsection{Proof of Claim \ref{clm:genericquasiprojupbound}}
\label{proof_of_genericquasiprojupbound}
From Equation \ref{Markov_ineq_norm}, it suffices to upper bound $\|(A-\omega)^m O_X  P^A_I \| $.
Abbreviate $\tilde{A}:= A-\omega\iden$. Introduce the multi-commutator
$$\ad_{\tilde{A}}^s(O_X ):=\underbrace{[\tilde{A},\ldots [\tilde{A},[\tilde{A}, O_X]]\ldots]}_{s\text{ times}}.$$
Consider the following identity,
\begin{align}
\tilde{A}^m O_X  P^A_I  = \sum_{s=0}^m \binom{m}{s} \ad_{\tilde{A}}^s(O_X )  \tilde{A}^{m-s}  P^A_I.
\end{align}
This shows that
\begin{align}
\| \tilde{A}^m O_X  P^A_I \|  &\le  \sum_{s=0}^m \binom{m}{s} \| \ad_{\tilde{A}}^s(O_X ) \| \cdot \| \tilde{A}^{m-s}  P^A_I \| \le  \sum_{s=0}^m \binom{m}{s} a_1^{m-s} \| \ad_{\tilde{A}}^s(O_X ) \| ,
\label{upp_tilde_A_^m O_X _P^A_I}
\end{align}
where we use $\| \tilde{A}^{m-s}  P^A_I \|=\|  (A-\omega)^{m-s}P^A_I \| \le a_1^{m-s}$.
The remaining task is to estimate the upper bound of $\| \ad_{\tilde{A}}^s(O_X ) \|= \| \ad_{A}^s(O_X ) \|$. This is done in the following claim.
\begin{claim} \label{cl:multicommutator_norm_quasi_local}
Let $A$ be an operator that is given by the form~\eqref{quasi_locality of A}.
Then, for an arbitrary operator $O_X$ which is supported on a subset $X$ ($|X|=k_0$), the norm of the multi-commutator $\ad^s_A(O_X)$ is bounded  from above by
\begin{align}
\| \ad_{A}^s (O_X) \| \le  \frac{(2a_1)^s(2k_0)^se^s}{\tau}\cdot\br{\frac{2}{a_2\tau}}^{\frac{2s}{\tau}}\cdot \br{s^{\tau_1}}^s \quad {\rm for} \quad  s\le m, \label{multicommutator_norm_quasi_local_ineq1}
\end{align}
where the constants $a_1$ and $a_2$ have been defined in Eq.~\eqref{quasi_locality of A}, and $\Gamma(\cdot)$ is the gamma function.
\end{claim}
\noindent
By applying the inequality~\eqref{multicommutator_norm_quasi_local_ineq1} to \eqref{upp_tilde_A_^m O_X _P^A_I}, we obtain 
\begin{align*}
\| \tilde{A}^m O_X  P^A_I \| &\le \sum_{s=0}^m \binom{m}{s} a_1^{m-s} \frac{(2a_1)^s(2k_0)^se^s}{\tau}\cdot\br{\frac{2}{a_2\tau}}^{\frac{2s}{\tau}}\cdot \br{s^{\tau_1}}^s \\
&\le \sum_{s=0}^m \binom{m}{s} (2a_1)^m \frac{(2ek_0)^m}{\tau}\cdot\br{\frac{2}{a_2\tau}}^{\frac{2m}{\tau}}\cdot \br{m^{\tau_1}}^m\\
&=  (4a_1)^m \frac{(2ek_0)^m}{\tau}\cdot\br{\frac{2}{a_2\tau}}^{\frac{2m}{\tau}}\cdot \br{m^{\tau_1}}^m=\frac{1}{\tau}\left [ 8e a_1k_0 [2/(a_2\tau)]^{2/\tau}m^{\tau_1}\right]^m.
\end{align*}
Therefore, setting $z=\omega+\theta$ in the inequality~\eqref{Markov_ineq_norm}, we obtain
\begin{align}
\|  P^A_{\ge \omega + \theta} O_X  P^A_I \| \le \frac{\|  \tilde{A}^m O_X  P^A_I \|}{\theta^m} 
&\le  \frac{1}{\tau}\left [ 8e a_1k_0 [2/(a_2\tau)]^{2/\tau}\frac{m^{\tau_1}}{\theta}\right]^m \\
&\le \frac{1}{\tau}\left( \frac{c_6 k_0 m^{\tau_1}}{\theta} \right)^m,
\end{align}
where $c_6:= 8ea_1 [2/(a_2\tau)]^{2/\tau}$. Let us choose $m=\tilde{m}$ with $\tilde{m}$ the minimum integer such that
\begin{align}
 \frac{c_6 k_0 \tilde{m}^{\tau_1}}{\theta} \le 1/e.
\end{align}
The above condition is satisfied by $\tilde{m}^{\tau_1} \le \theta/(ec_6 k_0)$, which implies
\begin{align}
\tilde{m} =\left\lfloor [\theta/(ec_6 k_0)]^{1/\tau_1} \right \rfloor,
\end{align}
where $\lfloor\cdot  \rfloor$ is the floor function.
From this choice, the claim concludes.

\subsection{Proof of Claim~\ref{cl:multicommutator_norm_quasi_local}}  \label{proof_multi_commutator_quasi_local}
Recall that we need to show, for an arbitrary operator $O_X$ which is supported on $k_0$ sites, the norm of the multi-commutator $\ad^s_A(O_X)$ is bounded  by
$$
\| \ad_{A}^s (O_X) \| \le  \frac{(2a_1)^s(2k_0)^se^s}{\tau}\cdot\br{\frac{2}{a_2\tau}}^{\frac{2s}{\tau}}\cdot \br{s^{\tau_1}}^s \quad {\rm for} \quad  s\le m.
$$
We start from the following expansion:
\begin{align}
\ad_{A}^s (O_X) = \sum_{k_1,k_2,\ldots,k_s} {g}_{k_1} {g}_{k_2}\cdots  {g}_{k_s}  [\bar{A}_{k_s},[\bar{A}_{k_{s-1}}, \cdots [ \bar{A}_{k_1},O_{X}]\cdots] . \notag 
  \end{align}
By using Lemma 3 in Ref.~\cite{AnnFlouqet} and setting $\zeta=1$ (see Definition \ref{quasi_locality of A}) we obtain
\begin{align}
&\|[[\bar{A}_{k_s},[\bar{A}_{k_{s-1}}, \cdots [ \bar{A}_{k_1},O_{X}]\cdots]  \| \le  2^{s} k_0 (k_0+k_1) (k_0+k_1+k_2) \cdots  (k_0+k_1+k_2+\cdots + k_{s-1}).
  \end{align}
Recall that we set $\|O_X\|=1$ and $|X|=k_0$.
The norm of $\ad_{A}^s (O_X) $ is bounded from above by
\begin{align}
&\| \ad_{A}^s (O_X) \|  \notag \\
\le&  \sum_{k_1,k_2,\ldots,k_s=1}^\infty 2^{s}  {g}_{k_1} {g}_{k_2}\cdots  {g}_{k_s} k_0 (k_0+k_1) (k_0+k_1+k_2) \cdots  (k_0+k_1+k_2+\cdots + k_{s-1})   \notag \\
=&\sum_{K\ge s}\ \sum_{\substack{k_1+k_2+\ldots+k_s=K\\k_1\ge1,k_2\ge1,\ldots,k_s\ge1}}2^{s}  {g}_{k_1} {g}_{k_2}\cdots  {g}_{k_s} k_0 (k_0+k_1) (k_0+k_1+k_2) \cdots  (k_0+k_1+k_2+\cdots + k_{s-1}) , 
\label{Inequality_quasi_ad_norm1}
  \end{align}
where the summation over $K$ starts from $s$ because each of $\{k_j\}_{j=1}^{s}$ is larger than $1$. Now, using the expression $\log[ {g}_{k}/a_1]=- a_2 k^{\tau}$ for $\tau\le1$, we have 
$\sum_{j=1}^{s} \log( {g}_{k_j}/a_1) \le \log( {g}_{k_1+k_2+\cdots +k_s}/a_1)$. This follows from $\sum_{j=1}^s k_j^{\tau} \geq \br{k_1+k_2+\cdots +k_s}^{\tau}$. Thus,  using $k_1+k_2+\cdots +k_s=K$, the summand in the inequality \eqref{Inequality_quasi_ad_norm1} is upper-bounded by
\begin{align}
& {g}_{k_1} {g}_{k_2}\cdots  {g}_{k_s} k_0 (k_0+k_1) (k_0+k_1+k_2) \cdots  (k_0+k_1+k_2+\cdots + k_{s-1}) \le a_1^s ( {g}_{K}/a_1)  k_0 (k_0+K)^{s-1} ,
\label{Inequality_quasi_ad_norm2}
  \end{align}
where we use the inequality $k_1+k_2+\cdots + k_j  \le  K$ for $j=1,2,\ldots,s-1$.
By combining the two inequalities~\eqref{Inequality_quasi_ad_norm1} and \eqref{Inequality_quasi_ad_norm2}, we obtain
\begin{align*}
\| \ad_{A}^s (O_X) \| 
\le&\sum_{K\ge s}\ \sum_{\substack{k_1+k_2+\ldots+k_s=K\\k_1\ge1,k_2\ge1,\ldots,k_s\ge1}}(2 a_1)^s  ( {g}_{K}/a_1)  k_0 (k_0+K)^{s-1} \\
\overset{(1)}\le & \sum_{K\ge s}\multiset{s}{K-s}(2 a_1)^s ( {g}_{K}/a_1)  k_0 (k_0+K)^{s-1} \\ 
=& \sum_{K\ge s}\binom{K-1}{s-1}(2 a_1)^s( {g}_{K}/a_1)  k_0 (k_0+K)^{s-1} \\ 
\overset{(2)}\le & (2a_1)^s(2k_0)^s\sum_{K\ge s}  \frac{e^sK^s}{s^s} ( {g}_{K}/a_1) (K)^{s-1} \\
\overset{(3)}=& \frac{(2a_1)^s(2k_0)^se^s}{s^s} \sum_{K\ge s} K^{2s-1} e^{-a_2 K^{\tau}}\leq  \frac{(2a_1)^s(2k_0)^se^s}{s^s} \sum_{K\ge 0} K^{2s-1} e^{-a_2 K^{\tau}}\\
\overset{(4)}\leq& \frac{(2a_1)^s(2k_0)^se^s}{s^s\tau}\cdot\br{\frac{2s}{a_2\tau}}^{\frac{2s}{\tau}}= \frac{(2a_1)^s(2k_0)^se^s}{\tau}\cdot\br{\frac{2}{a_2\tau}}^{\frac{2s}{\tau}}\cdot \br{s^{\frac{2}{\tau}-1}}^s.
\end{align*}
where in $(1)$, $\multiset{}{}$ denotes the multi-combination, namely $\multiset{n}{m}=\binom{n+m-1}{n-1}$, in $2$ we upper bound $\binom{K-1}{s-1}\leq \frac{e^sK^s}{s^s}$, $k_0+K\leq 2k_0 K$, in $(3)$ we use the sub-exponential form of $ {g}_{K}$ in Eq.~\eqref{quasi_locality of A} and in $(4)$ we use Fact \ref{fact:integrals}. Since $\tau_1=\frac{2}{\tau}-1$, this proves the statement.

\section{Quasi-locality of $\pmb{\widetilde{W}}$}
\label{append:Wquasi}
We here aim to obtain $(\tau, a_1, a_2, \zeta)$-quasi-locality of the operator $\quW$, where $\{\tau, a_1, a_2, \zeta\}$ defined in Definition~\ref{def:Quasi-local operators}. In particular, we will show that
$$
\big(\tau, a_1, a_2, \zeta\big)=\br{1/D, \orderof{1}, \orderof{1/\beta}, \orderof{\beta^{2D+1}}  \br{\max_{j\in \Lambda}v_j}}
$$
suffices to prove the quasi-locality of $\quW$. Recall the definition of $\quW$:
\begin{align*}
\quW= \int_{-\infty}^\infty f_\beta(t)\ e^{-iHt}\ W\ e^{iHt} dt,
\end{align*}
where
\begin{align*}
f_\beta(t)= \frac{2}{\beta\pi}\log \frac{e^{\pi |t|/\beta}+1}{e^{\pi |t|/\beta}-1}
\end{align*}
and
$$W=\sum_{i\in \Lambda} v_i E_i.$$ We write
$$\quW=\sum_i v_i \int_{-\infty}^\infty f_\beta(t)\ e^{-iHt}\ E_i\ e^{iHt} dt.$$
Abbreviate 
$$
\tilde{E}_i(t):= e^{-iHt}\ E_i\ e^{iHt}
$$
and recall that $\tilde{E}_i= \int_{\infty}^{\infty} f_{\b}(t) \tilde{E}_i(t)$. Moreover, (with some abuse of notation) let $B(r,i)$ denote the ball of radius $r$ such that: the centre of $B(r,i)$ coincides with the the center of the smallest ball containing $E_i$. We assume that $r$ ranges in the set $\{m_i,m_i+1,\ldots ,n_i\}$, where $m_i$ is the radius of the smallest ball containing $E_i$ and $n_i$ is the number such that $B(n_i,i)=\Lambda$. Define
$$\tilde{E}^r_i(t):= \Tr_{B(r,i)^c}[\tilde{E}_i(t)]\otimes \frac{\iden_{B(r,i)^c}}{\Tr[\iden_{B(r,i)^c}]}, \quad \tilde{E}^0_i(t)=0,
$$
i.e., $\tilde{W}^r_i(t)$ traces out all the qudits in $\tilde{E}_i(t)$ that are at outside the $B(r,i)$-ball around $\tilde{E}^r_i$. From \cite{BravyiHV06}, we have
$$
\|\tilde{E}_i(t)-\tilde{E}^r_i(t)\|
\leq \|E_i\| \min \Big\{1,  c_3 r^{D-1} e^{-c_4(r-\vL|t|)} \Big\}
$$
which in particular implies
$$
\|\tilde{E}_i^r(t)-\tilde{E}^{r-1}_i(t)\|
\leq 2\min \Big\{1,  c_3 r^{D-1}e^{-c_4(r-\vL|t|)} \Big\}, 
$$
where we use $\|E_i\|=1$, $\vL$ is the Lieb-Robinson velocity (as defined in Fact~\ref{fact:LRB}) and $c_3, c_4$ are constants. We note that the $2\min\{1,\cdot\}$ is derived from the trivial upper bound
$\|\tilde{E}_i^r(t)-\tilde{E}^{r+1}_i(t)\|\le~2$.
This allows us to write the following quasi-local expression:
$$\tilde{E}_i(t) = \sum_{r=m_i}^{n_i}\br{\tilde{E}^{r}_i(t)-\tilde{E}^{r-1}_i(t)}.$$
Using this, we can now write the quasi-local representation of $\tilde{E}_i$ as follows.
\begin{align*}
  \int_{-\infty}^\infty f_\beta(t)\tilde{E}_i(t)dt&= \int_{-\infty}^\infty f_\beta(t)\sum_{r=m_i}^{n_i}\br{\tilde{E}^{r}_i(t)-\tilde{E}^{r-1}_i(t)}.
\end{align*}
To see that it is quasi-local, observe that the term with radius $r$ has norm
\begin{align*}
&\int_{-\infty}^\infty f_\beta(t)\Big\|\tilde{E}^{r}_i(t)-\tilde{E}^{r-1}_i(t)\Big\| \notag \\
&\leq c_3 r^{D-1}  e^{-c_4 r}\cdot \int_{-r/\vL}^{r/\vL}e^{c_4\vL|t|-\pi |t|/\beta }dt + \int_{r/\vL}^\infty e^{-\pi |t|/\beta} dt 
+ \int_{-\infty}^{-r/\vL} e^{-\pi |t|/\beta} dt \notag \\
&\le 2c_3 r^{D-1} e^{-c_4r} \frac{e^{|c_4 \vL-\pi/\beta|r/\vL}-1}{|c_4 \vL-\pi/\beta|} + 2 \frac{e^{-\pi r/(\beta \vL)}}{\pi/\beta} \notag \\
&\le 2c_3 r^{D-1} (r/\vL)e^{-\min(\pi r/(\beta \vL),c_4r)}  + 2(\beta/\pi)e^{-\pi r/(\beta \vL)},
\end{align*}
where we use $(e^{xy}-1)/x \le ye^{xy}$ for $x\ge0$ and $y\ge0$. Define 
$$a_{B(r,i)}:=e^{\pi r/(2\beta \vL)}\int_{-\infty}^\infty f_\beta(t)\br{\tilde{E}^{r}_i(t)-\tilde{E}^{r-1}_i(t)}.$$
Here, the operator $a_{B(r,i)}$ is supported on the subset $B(r,i)$.
Then, from $|B(r,i)|=\orderof{r^D}$, the quasi-local representation of $\quW$ is given as
\begin{align*}
    \quW= \sum_{i\in \Lambda} v_i\sum_{r=m_i}^{n_i}e^{-\pi r/(2\beta \vL)}a_{B(r,i)} = \sum_{i\in \Lambda} \sum_{r=m_i}^{n_i} e^{-\orderof{|B(r,i)|^{\frac{1}{D}}}}v_i a_{B(r,i)}, 
\end{align*}
with $e^{-\orderof{|B(r,i)|^{\frac{1}{D}}}}$ decaying sub-exponentially with rate $\tau=1/D$, for all $i\in \Lambda$. We also obtain the parameter $\zeta$ in Eq.~\eqref{quasi_locality of A} by
\begin{align}
\sum_{r,j:B(r,j) \ni i}v_j\|a_{B(r,j)}\|
\leq \sum_{r} c_5 r^D\sum_{j: B(r,j) \ni i} v_je^{-\pi r/(2\beta \vL)}
&\leq \br{\max_{j\in \Lambda}v_j}\sum_{r} c_5 c_B r^{2D} e^{-\pi r/(2\beta \vL)}\notag \\
&\leq 2c_Bc_5\br{\frac{2D+1}{\pi/(2\beta \vL )}}^{2D+1} \br{\max_{j\in \Lambda}v_j}, \notag 
\end{align}
where we define $c_B$ such that $|B(r,j)| \le c_B r^D$ and we used Fact \ref{fact:integrals} (2) with $p=1$, $b=2D$ and $c=\pi/(2\beta \vL)$. This completes the representation and shows that $\quW$ is a $\br{1/D, \orderof{1}, \orderof{1/\beta}, \orderof{\beta^{2D+1}}  \br{\max_{j\in \Lambda}v_j}}$-quasi-local.

\section{Proof of Lemma \ref{lem:Wprimelowb}}
\label{append:lowerboundwprimei}
Recall that the goal in this section is to prove that for $\quW$ defined in Lemma~\ref{lem:variance lower bound on Hessian} we have
$$
\max_{i\in \Lambda} \Tr[(\quW_\locci)^2\eta] =   \frac{\Omega(1)}{\br{\beta \log(\beta)+1}^{2D+2}} \br{\max_{i\in \Lambda}v_i^2},
$$
where $\eta$ is the maximally mixed state. In this direction, we will now prove that
\begin{align} 
\max_{i\in \Lambda} \| \quW_\locci \sqrt{\eta} \|_F \ge   \frac{c_7}{\br{\beta \log(\beta)+1}^{D+1}}    \max_{i\in \Lambda}(|v_i|)  ,\label{main_uppe_bound_Wi'}
\end{align}
for a constant $c_7=\mathcal{O}(1)$. For convenience, let us define ${\rm argmax}_{i\in \Lambda} |v_i|=i_+$, or equivalently $|v_{i_+}|= \max_{i\in \Lambda} |v_i|$.
We denote the ball region $B(r,i_+)$ by $B_{r}$ for the simplicity, where $r$ is fixed later. Let us consider $\quW[B_r]$ which is defined as follows:
\begin{align} 
\quW[B_r]:= \int_{-\infty}^\infty f_\beta(t)e^{-iHt} W[B_r] e^{iHt} dt ,\quad W[B_r] :=\sum_{i\in B_r} v_i E_i.
\label{W'_explicit_form_B_r_i0}
\end{align}
Since $\quW[B_r]$ is obtained from $W[B_r]$ in an equivalent manner as $\quW$ is obtained from $W$, the following claim follows along the same lines as Theorem \ref{claim_W_norm_W'_norm}. We skip the very similar proof.
\begin{claim} \label{claim_Wr_norm_W'r_norm}
It holds that
\begin{align*} 
\|\quW[B_r]\|_F^2 \ge \frac{\mathcal{D}_\Lambda}{c_5[\beta \log(r)+1]^2}\sum_{i\in B_r} v_i^2  ,
\end{align*} 
where $c_5$ is a constant of $\orderof{1}$.
\end{claim}
Since the new operator $\quW[B_r]$ well approximates the property of $\quW$ around the site $i_+$, as long as $r$ is sufficiently large, we expect that $\quW[B_r]_{\locciplus}$ and $\quW_{\locciplus}$ are close to each other. The claim below makes this intuition rigorous: 
\begin{claim} \label{claim_W_i'_B_r_i0_W_i}
It holds that
\begin{align} 
\| \quW_{\locciplus} - \quW[B_r]_{\locciplus}  \| \le c_1 |v_{i_+}|\beta^D e^{-c_2 r/\beta},
\end{align}
where $c_1,c_2$ are constants of $\orderof{1}$.
\end{claim}

This claim implies that the contribution of all the terms in $\quW_{\locciplus}$ which are not included in the $B_r$ ball around $i_+$ decays exponentially with $r$. Hence, 
 \begin{align} 
\| \quW_{\locciplus}\|_F =  \| \quW_{\locciplus}  - \quW[B_r]_{\locciplus} + \quW[B_r]_{\locciplus} \|_F 
&\ge \| \quW[B_r]_{\locciplus} \|_F -  \| \quW_{\locciplus}  - \quW[B_r]_{\locciplus}\|_F  \notag \\
&\ge \| \quW[B_r]_{\locciplus} \|_F -  \sqrt{\mathcal{D}_\Lambda} \| \quW_{\locciplus}  - \quW[B_r]_{\locciplus}\| \notag  \\
&\ge \| \quW[B_r]_{\locciplus} \|_F -  \sqrt{\mathcal{D}_\Lambda} c_1 |v_{i_+}| \beta^D e^{-c_2 r/\beta}, 
\label{upper_bound_W_i_+_F_2}
\end{align}
where we use $\|\quW_{\locciplus}  - \quW[B_r]_{\locciplus}\|_F \le \sqrt{\mathcal{D}_\Lambda} \|\quW_{\locciplus}  - \quW[B_r]_{\locciplus}\|$ in the second inequality. Second, we consider the approximation of $\quW[B_r]$ by $\quW[B_r,B_{r'}]$ which are supported on $B_{r'}$:
 \begin{align} 
\quW[B_r,B_{r'}]: =\tr_{B_{r'}^\co} (\quW[B_r]) \otimes \frac{\iden_{B_{r'}^\co}}{d^{|B_{r'}^\co|}}.
\end{align}
Because of the quasi-locality of $\quW$, we expect $\quW[B_r,B_{r'}]\approx \quW[B_r]$ for $r' \gg r$. This is shown in the following lemma:
 \begin{claim} \label{claim_diff_W'B_r - W'B_r,B_r'}
The norm difference between $\quW[B_r,B_{r'}]$ and $ \quW[B_r] $ is upper-bounded as
\begin{align} 
\| \quW[B_r] - \quW[B_r,B_{r'}]\| \le c_3 |v_{i_+}| r^D \beta e^{-c_4  |r'-r|/\beta}
\label{main_claim_diff_W'B_r - W'B_r,B_r'}
\end{align}
and
\begin{align} 
\| \quW[B_r]_{\locciplus} - \quW[B_r,B_{r'}]_{\locciplus}\| \le 2c_3 |v_{i_+}| r^D \beta e^{-c_4  |r'-r|/\beta},
\label{main_claim_diff_W'B_r - W'B_r,B_r'iplus}
\end{align}
where $c_3,c_4$ are constants of $\orderof{1}$.
\end{claim}

The claim reduces the inequality~\eqref{upper_bound_W_i_+_F_2} to 
 \begin{align} 
\| \quW_{\locciplus}\|_F &\ge \| \quW[B_r]_{\locciplus} \|_F -  \sqrt{\mathcal{D}_\Lambda} c_1 |v_{i_+}| \beta^D e^{-c_2 r/\beta} \notag \\
&\ge \| \quW[B_r,B_{r'}]_{\locciplus} \|_F -  \sqrt{\mathcal{D}_\Lambda}\| \quW[B_r]_{\locciplus} - \quW[B_r,B_{r'}]_{\locciplus}\| -  \sqrt{\mathcal{D}_\Lambda} c_1 |v_{i_+}| \beta^D e^{-c_2 r/\beta}  \notag \\
&\ge \| \quW[B_r,B_{r'}]_{\locciplus} \|_F -   \sqrt{\mathcal{D}_\Lambda} c_1 |v_{i_+}|\beta^D e^{-c_2 r/\beta}  -  2\sqrt{\mathcal{D}_\Lambda} c_3 |v_{i_+}|r^D \beta e^{-c_4  (r'-r)/\beta}.
\label{upper_bound_W_i_+_F_3}
\end{align}
Next, we relate the norm of $\quW[B_r,B_{r'}]_{\locciplus}$ to that of $\quW[B_r,B_{r'}]$ using Claim \ref{claim_global_norm_local_norm}. By recalling that $\quW[B_r,B_{r'}]_{\locciplus} $ is supported on $B_{r'}$, this gives  
 \begin{align} 
\| \quW_{\locciplus}[B_r,B_{r'}] \|_F \ge \frac{1}{|B_{r'}|}\| \quW[B_r,B_{r'}] \|_F,
\end{align}
which reduces the inequality~\eqref{upper_bound_W_i_+_F_3} to 
 \begin{align} 
\begin{aligned}
\| \quW_{\locciplus}\|_F \ge& \frac{1}{|B_{r'}|}\| \quW[B_r,B_{r'}] \|_F  -   \sqrt{\mathcal{D}_\Lambda} c_1 |v_{i_+}| \beta^D e^{-c_2 r/\beta}  -  2\sqrt{\mathcal{D}_\Lambda} c_3 |v_{i_+}| r^D  \beta e^{-c_4  (r'-r)/\beta}
\\
\ge& \frac{1}{|B_{r'}|}\| \quW[B_r] \|_F  -   \sqrt{\mathcal{D}_\Lambda} c_1 |v_{i_+}|\beta^D e^{-c_2 r/\beta}  -  (2+1/|B_{r'}|) \sqrt{\mathcal{D}_\Lambda} c_3 |v_{i_+}|r^D  \beta e^{-c_4  (r'-r)/\beta},
\label{upper_bound_W_i_+_F_4}
\end{aligned}
\end{align}
where in the second inequality we apply Claim~\ref{claim_diff_W'B_r - W'B_r,B_r'} to $\| \quW[B_r,B_{r'}] \|_F$. 
Finally, we use the lower bound given in Claim \ref{claim_Wr_norm_W'r_norm} and the inequality $\sum_{i\in B_r} v_i^2 \ge v_{i_+}^2$ (since $i_+ \in B_r $) to obtain 
\begin{align*} 
\|\quW[B_r]\|_F^2 \ge \frac{\mathcal{D}_\Lambda}{c_5[\beta \log(r)+1]^2} v_{i_+}^2 .
\end{align*} 
This reduces the inequality~\eqref{upper_bound_W_i_+_F_4} to the following: 
\begin{align*} 
\frac{\|\quW_{\locciplus}\|_F }{\sqrt{\mathcal{D}_\Lambda}}
&\ge  \frac{|v_{i_+}|}{c_8\br{r'}^D\sqrt{c_5}[\beta \log(r)+1]}  -    c_1 |v_{i_+}|\beta^D e^{-c_2 r/\beta}  -  3c_3 |v_{i_+}| r^D  \beta e^{-c_4 (r'-r)/\beta},
\end{align*}
where we used $|B_{r'}|\leq c_8 \br{r'}^D$, for some constant $c_8$. By choosing $r'=2r$ and $r=\orderof{1}\cdot D\beta \log(\beta)+1$, we have 
 \begin{align} 
\frac{\|\quW_{\locciplus}\|_F }{\sqrt{\mathcal{D}_\Lambda}} = \| \quW_{\locciplus} \sqrt{\eta} \|_F \ge  \frac{c_7|v_{i_+}|}{\br{\beta \log(\beta)+1}^{D+1}},
\end{align}
for some constant $c_7$. By using the inequality $\max_{i\in \Lambda}\| \quW_\locci\|_F \ge\| \quW_{\locciplus}\|_F $, we obtain the main statement.
This completes the proof. $\square$

\subsection{Proof of Claims~\ref{claim_W_i'_B_r_i0_W_i},~\ref{claim_diff_W'B_r - W'B_r,B_r'}}

\begin{proof}[Proof of Claim \ref{claim_W_i'_B_r_i0_W_i}] 
Recall that the goal is to prove 
$$
\| \quW_{\locciplus} - \quW[B_r]_{\locciplus}  \| \le c_1 |v_{i_+}|\beta^D e^{-c_2 r/\beta}
$$
 for constants $c_1,c_2$. We start from the integral representation of $\quW_{\locciplus}$:
\begin{align} 
\quW_{\locciplus}= \quW- \int \mu (U_{i_+}) U_{i_+}^\dagger \quW  U_{i_+},
\end{align}
where $\mu(U_{i_+})$ is the Haar measure for unitary operator~$U_{i_+}$ which acts on the $i_+$th site.
This yields  
 \begin{align} 
\quW_{\locciplus} - \quW[B_r]_{\locciplus}=  \quW[B^\co_r]- \int \mu (U_{i_+}) U_{i_+}^\dagger \quW[B^\co_r]  U_{i_+} .
\end{align}
We thus obtain 
 \begin{align} 
\| \quW_{\locciplus} - \quW[B_r]_{\locciplus}\| 
&\le  \sup_{U_{i_+}} \| [U_{i_+},  \quW[B^\co_r]] \| \notag \\
&\le  \int_{-\infty}^\infty f_\beta(t)  \sum_{j\in B_r^\co} |v_j|  \sup_{U_\locci}\| [U_{i_+},  e^{-iHt} E_j e^{iHt}]  \|  dt  \notag \\
&\le  |v_{i_+}|\sum_{j\in B_r^\co}  \int_{-\infty}^\infty f_\beta(t) \min( e^{-c(\dist(i_+,j) - \vL t)} ,1) dt  ,
\end{align}
where we use $|v_j |\le |v_{i_+}|$ and the Lieb-Robinson bound (Fact \ref{fact:LRB}) for the last inequality. Because the function $ f_\beta(t)$ decays as $e^{-\orderof{t/\beta}}$ and $\dist(i_+,j) \ge r$ for $j\in B_r^\co$, we have 
 \begin{align} 
 |v_{i_+}| \sum_{j\in B_r^\co}\int_{-\infty}^\infty f_\beta(t) \min( e^{-c(\dist(i_+,j) - \vL t)} ,1) dt  \le c_1 |v_{i_+}|\beta^D e^{-c_2 r/\beta}.
\end{align}
This completes the proof.
\end{proof}

\begin{proof}[Proof of Claim~\ref{claim_diff_W'B_r - W'B_r,B_r'}] Recall that we wanted to show
$$
\| \quW[B_r] - \quW[B_r,B_{r'}]\| \le c_3 |v_{i_+}| r^D \beta e^{-c_4  |r'-r|/\beta}.
$$
In order to prove this, we also utilize the integral representation of $\quW[B_r,B_{r'}]$:
\begin{align} 
\quW[B_r,B_{r'}]:=  \int \mu (U_{B_{r'}^\co}) U_{B_{r'}^\co}^\dagger \quW[B_r]   U_{B_{r'}^\co},
\end{align}
which yields an upper bound of $\| \quW[B_r] - \quW[B_r,B_{r'}]\|$ as 
 \begin{align} 
\| \quW[B_r] - \quW[B_r,B_{r'}]\| \le   \int \mu (U_{B_{r'}^\co})  \| [\quW[B_r] ,  U_{B_{r'}^\co}] \|.
\end{align}
From the definition~\eqref{W'_explicit_form_B_r_i0} of $\quW[B_r]$ and the Lieb-Robinson bound (Fact \ref{fact:LRB}), we obtain
 \begin{align} 
\int \mu (U_{B_{r'}^\co})  \| [\quW[B_r] ,  U_{B_{r'}^\co}] \|
&\le \int \mu (U_{B_{r'}^\co})   \int_{-\infty}^\infty f_\beta(t)\sum_{j\in B_r} |v_j|\cdot  \| [ e^{-iHt} E_j  e^{iHt} , U_{B_{r'}^\co}] \| \notag \\
&\le   |v_{i_+}| \int_{-\infty}^\infty f_\beta(t)  \sum_{j\in B_r}  \min(  e^{-c(r'-r- \vL t)}  ,1)  dt \notag \\
&\le  c'_3 |v_{i_+}| \cdot  |B_r| \cdot \beta e^{-c_4 r/\beta},
\end{align}
where $\partial B_{r'}^\co$ is the surface region of $B_{r'}^\co$. Since $ |B_r| \propto r^D$, we obtain the main inequality~\eqref{main_claim_diff_W'B_r - W'B_r,B_r'}. Now, since
$$\quW[B_r]_{\locciplus}= \quW[B_r]- \int \mu (U_{i_+}) U_{i_+}^\dagger \quW[B_r]  U_{i_+}$$
and
$$\quW[B_r, B_{r'}]_{\locciplus}= \quW[B_r, B_{r'}]- \int \mu (U_{i_+}) U_{i_+}^\dagger \quW[B_r, B_{r'}]  U_{i_+},$$
we obtain the second inequality~\eqref{main_claim_diff_W'B_r - W'B_r,B_r'iplus} due to
$$\|\quW[B_r]_{\locciplus}-\quW[B_r, B_{r'}]_{\locciplus}\|\leq 2\| \quW[B_r] - \quW[B_r,B_{r'}]\|.$$
This completes the proof. 
\end{proof}

\section{Proof of Theorem \ref{thm:lower_bound}}
\label{learning_lower_bound}
For convenience of the reader, we restate the theorem here.

\begin{thm}[Restatement of Theorem~\ref{thm:lower_bound}]
The number of copies $N$ of the Gibbs state needed to solve the $\HLP$ and outputs a $\hat{\mu}$ satisfying $\|\hat{\mu}-\mu\|_2\leq \varepsilon$ with probability $1-\delta$ is lower bounded by
$$
N\geq \Omega\Big(\frac{\sqrt{m}+\log(1-\delta)}{\beta\varepsilon} \Big).
$$
\end{thm}
\begin{proof}
In order to prove the lower bound, we consider learning the parameters $\mu \in \bbR^m$ of the following class of one-local Hamiltonians on $m$ qubits: 
$$
H(\mu)=\sum_{i=1}^m \mu_i \ketbra{1}{1}_i.
$$ 
Let $T_m:\{\mu\in \mathbb{R}^m_+: \sum_i \mu^2_i\leq 100\varepsilon^2\}$ be an orthant of the hypersphere of radius $\theta$ in $\mathbb{R}_+^m$. We have the following claim. 
\begin{claim}
There exists a collection of $2^m$ points in $T_m$, such that the $\ell_2$ distance between each pair is $\geq \varepsilon$.
\end{claim}
\begin{proof}
Pick $2^m$ points uniformly at random in $T_m$. By union bound, the probability that at least one pair is at a distance of at most $\varepsilon$ is at most $\br{2^m}^2$ times the probability that a fixed pair of points is at a distance of at most $\varepsilon$. But the latter probability is upper bounded by the ratio between the volume of a hypersphere of radius $\varepsilon$ and the volume of $T_m$, which is $\frac{\varepsilon^m}{\br{10\varepsilon}^m/2^m}= \frac{1}{5^m}$. Since $\br{2^m}^2\frac{1}{5^m}<1$, the claim concludes.
\end{proof}
Let these set of $2^m$ points be $S$. For some temperature $\beta>0$ and unknown $\mu\in S$, suppose $\mathcal{A}$ is an algorithm that is given $N$ copies of $\rho_\b(\mu)$ and, with probability $1-\delta$, outputs $\mu'$ satisfying 
$\|\mu'-\mu\|_2\leq \varepsilon$. We now use $\mathcal{A}$ to assign the estimated $\hat{\m}$ to exactly one of the parameters $\mu$. Once the learning algorithm obtains an output $\mu'$, we can find the closest point in $S$ (in $\ell_2$ distance) as our estimate of $\mu$, breaking ties arbitrarily. With probability $1-\delta$, the closest $\mu\in S$ to $\mu'$ is the correct $\mu$ since by the construction of $S$, $\|\mu'-\mu\|_2\leq \e$. Thus, the algorithm $\mathcal{A}$ can be used to solve the problem of estimating the parameters $\mu$ themselves (not only approximating it). We furthermore show that the number of samples required to estimate $\mu \in S$ is large using lower bounds in the quantum state discrimination. We will directly use the lower bound from \cite{Hayashi_Kawachi_Kobayashi06} (as given in \cite{Harrow_Winter_12}). Before we plug in their formula, we need to bound the maximum norm of $\r_{\b}(\m)$ for $\m \in S$. That is,
\begin{align*}
    \max_{\mu\in S} \{2^m\|\rho_\b(\mu)\|\}&= \max_{\mu\in S} 2^m\br{\bigotimes_{i=1}^m\left\|\frac{1}{1+e^{-\beta\mu_i}}\ketbra{0}{0}+\frac{e^{-\beta\mu_i}}{1+e^{-\beta\mu_i}}\ketbra{1}{1}\right\|}\\
    &= \max_{\mu\in S} \br{\bigotimes_{i=1}^m\left|\frac{2}{1+e^{-\beta\mu_i}}\right|}\\
    &=\max_{\mu\in S} \br{\bigotimes_{i=1}^m\left|\frac{2e^{\beta\mu_i}}{e^{\beta\mu_i}+1}\right|}\\
    &\leq\max_{\mu\in S} \br{\bigotimes_{i=1}^m\left|\frac{2e^{\beta\mu_i}}{2}\right|}\\
    &=\max_{\mu\in S} \br{e^{\beta \sum_{i=1}^m\mu_i}} \leq e^{\beta\sqrt{m}\sqrt{\sum_{i=1}^m\mu_i^2}}\leq e^{\beta\sqrt{m}\cdot 10\varepsilon},
\end{align*}
since $\sum_i \mu^2_i\leq 100\varepsilon^2$ for all $i\in S$. Thus, the lower bound for state identification of $\{H(\mu):\mu\in S\}$ in \cite[Equation~2]{Harrow_Winter_12} (cf. \cite{Hayashi_Kawachi_Kobayashi06} for the original statement) implies that 
$$
N\geq \frac{\log|S| + \log(1-\delta)}{\log\br{\max_{\mu\in S} \{2^m\|\rho(\mu)_\beta\|\}}}= \frac{m\log 2+\log(1-\delta)}{10\sqrt{m}\beta\varepsilon}=\mathcal{O}\Big(\frac{\sqrt{m}+\log(1-\delta)}{\varepsilon\beta}\Big).
$$
This establishes the lower bound.
\end{proof}

\end{document}